\documentclass[12pt]{amsart}
\usepackage{amsfonts,amssymb,graphicx,mystyle, xcolor, appendix, thm-restate}
\usepackage{tikz, pstricks, egameps}
\usepackage{chngcntr}
\usetikzlibrary{arrows.meta, positioning}

\counterwithin{figure}{section}

 \let\backslash=\setminus \let\ge=\geqslant \let\le=\leqslant

\def\a{\alpha} \def\b{\beta} \def\d{\delta} \def\e{\varepsilon}
   \def\g{\gamma} 
\def\l{\lambda}   \def\s{\sigma} 
 \def\t{\tau} \def\z{\zeta}

\def\D{\Delta} \def\G{\Gamma} \def\S{\Sigma}
 \def\E{\mathcal{E}}  \def\K{\mathcal{K}}
\def\U{\mathcal{U}}
\def\Q{\mathcal{Q}}
 \def\M{\mathcal{M}} 
\def\T{\mathcal{T}}  \def\P{\mathcal{P}}
\def\calQ{\mathcal{Q}} \def\GG{\mathcal{G}}  
 \def\Re{\mathbb{R}} 
\def\BR{\text{BR}}
\def\bbP{\mathbb{P}}

\def\X {\mathcal X}

\def\bbB{\mathbb{B}}

\def\bbBR{\mathbb{BR}}
\def\bbF{\mathbb{F}}

\def\sn{\scriptscriptstyle \mathcal{N}}
\def\so{\scriptscriptstyle O}
\def\Z{\mathcal{Z}}

\def\proj{\textrm{proj}}

\def\Lmath{\mathcal{L}}

\begin{document} \openup 1\jot

\title[Axiomatic Equilibrium Selection]{Axiomatic Equilibrium Selection \\ in Generic Extensive-Form Games}
\author[S.~Govindan]{Srihari Govindan}
\address{Department of Economics, University of Rochester, Rochester, NY 14627, USA.}
\email{s.govindan@rochester.edu}
\author[R.~Wilson]{Robert Wilson}
\address{Stanford Graduate School of Business, Stanford, CA 94305, USA.}
\email{rwilson@stanford.edu}

\begin{abstract}
	A solution concept that is a refinement of Nash equilibria selects for each finite game a nonempty collection of closed and connected subsets of Nash equilibria as solutions.  We impose three axioms for such solution concepts.  The axiom of backward induction requires each solution to contain a quasi-perfect equilibrium.  Two invariance axioms posit that solutions of a game are the same as those of a game obtained by the addition of strategically irrelevant strategies and players. Stability \cite{M1989, M1991} satisfies these axioms; and any solution concept that satisfies them must, for generic extensive-form games, select as solutions sets that contain stable sets; in particular, the outcomes they select have to be stable in such games.  A strengthening of the two invariance axioms provides an analogous axiomatization of components of equilibria with a nonzero index. 	
\end{abstract}

\thanks{We are grateful to Paulo Barelli for his help with this project, and to  Lucas Pahl and Carlos Pimienta for their comments. We thank four anonymous referees.}

\maketitle

\section{Introduction}

The 1986 paper on strategic stability by Kohlberg and Mertens (\cite{KM1986}, henceforth KM) is seminal in the literature on equilibrium selection.  Hitherto, refinements were proposed on an {\it ad hoc} basis, with a view to capturing a particular line of game-theoretic reasoning---see \cite{HK2002} for a comprehensive survey.  By considering, and attempting to reconcile, at once multiple criteria, KM laid out an agenda for a methodical approach to equilibrium selection. The paper had a profound effect on both foundational theory (see, for instance, \cite{R1992, D2024}) and the application of refinements to economic models of, among others, limit pricing \cite{BR1991}, signaling games \cite{CK1987}, adverse selection \cite{G1992}, and voting \cite{DDL2006}.

KM argued that refinements of Nash equilibria should derive from axioms adapted from single-person decision theory, and they listed properties that the axioms should imply.  The solution concept they defined, stability, satisfied most but not all of their properties.  Mertens (\cite{M1989, M1991}) showed that his reformulated definition of stability satisfied all those properties, and more. Neither of these papers provided an axiomatic approach (KM did not ``feel ready for such an [axiomatic] approach"). In Govindan and Wilson (\cite{GW2012}, henceforth GW) we provided a first step in implementing KM's program, in generic two-player games. In this paper we take another step by allowing for any finite number of players. Specifically, we show that, in generic $N$-player games, axioms capturing backward induction and invariance of solutions to the addition of irrelevant strategies or players imply that selected solutions must contain stable sets; thus in generic games the outcomes selected must be stable.  Stronger versions of the invariance axioms characterize essentiality.\footnote{As explained in our discussion of the invariance axioms, such a strengthening reveals a shortcoming of essentiality in comparison to stability.}

A solution concept for finite games that refines the set of Nash equilibria selects for each finite normal form game a collection of subsets of Nash equilibria, called solutions. As KM showed, solutions cannot be singletons if one wants to ensure that every game has a solution where players play undominated strategies. As in GW, we define solutions to be nonempty collections of closed and connected subsets of Nash equilibria. Connectedness forces a solution to lie within a single component of the Nash equilibria and, thus, is necessary for a selection among equilibrium outcomes and payoffs. Asking for solutions to be closed subsets, while a purely topological property, seems to be a mild technical requirement, which, actually, is not needed for our results concerning outcomes. 

We impose  three axioms for a solution concept. The first axiom, Axiom B, captures the logic of backward induction: every solution of the normal form of an extensive-form game should contain a mixed strategy profile that is equivalent to a quasi-perfect equilibrium of the extensive form.   

The other two axioms concern the property of invariance for solution concepts. We apply the principle that a subset of equilibria selected by a refinement should inherit the invariance properties of Nash equilibria. We invoke two invariance properties. Each is an instance of the general decision-theoretic requirement that there are no framing effects (Tversky and Kahneman \cite{TK1981}); that is, equilibria are not affected by embedding the given game in a larger one that does not alter the original players' strategies and payoffs.\footnote{The term ``framing effects'' refers to the empirical phenomenon that some people react differently to the same information depending on how it is presented.} 
For each of these invariance properties, we invoke an axiom requiring that each player's preference ranking over his actions is unaffected by the embedding.

The first invariance axiom (Axiom D$^*$) considers the addition of weakly dominated strategies that are strategically irrelevant---see Definition \ref{def saids}---and requires that the solutions of a game be the same as those of a game obtained by adding such strategies. The second invariance axiom (Axiom I$^*$) considers the addition of players whose strategies are strategically irrelevant---see Definition \ref{def strong embedding}---and requires a similar invariance of solutions.

Stable sets as defined by Mertens \cite{M1989, M1991} satisfy all three axioms.  Our main result, Theorem \ref{thm main 1}, concerns the implications of the axioms for an extensive-form game with perfect recall and generic payoffs.  It shows that any solution concept that satisfies our three axioms must for generic games select as solutions sets that contain stable sets.  Such games have finitely many equilibrium outcomes, i.e.,~probability distributions over the terminal nodes induced by the equilibria. The fact that solutions must be connected implies that there is a unique equilibrium outcome associated with each solution. A corollary to Theorem \ref{thm main 1} shows that if a solution concept satisfies the axioms, then every solution induces a stable outcome, i.e.,~an outcome associated with a stable set.

How does Theorem \ref{thm main 1} compare with the result in GW?  That paper, too, imposed three axioms: Axiom B, a weaker version of Axiom I$^*$, and Admissibility, which requires each solution to consist only of undominated strategy profiles.  Its result, however, is a little stronger: any solution concept that satisfies those three axioms must select from among the stable sets. If we eliminate the admissibility axiom from their list, then we get Theorem \ref{thm main 1} for two-player games.   Alternatively, if we impose a minimality axiom, that solutions are as ``small'' as possible, then for any generic game we must get a selection of its stable sets.\footnote{In the two-player case, admissibility can be replaced with minimality to produce the same result.}

One could insist on a stronger invariance principle by considering embeddings where the best-reply correspondence is preserved on the restriction of the larger game to the original game and allowing it to be arbitrary otherwise---whereas, by contrast, our original invariance principle would ask for the best-reply correspondences of the two games to be related globally over the strategy spaces.  Such an expansive reading of invariance (see Definitions \ref{def waids} and \ref{def weak embedding}) results in stronger versions of the invariance axioms, called Axioms D and I. Replacing Axioms D$^*$ and I$^*$ with these axioms produces an analogous characterization of essentiality, which is the solution concept that selects components of equilibria with nonzero indices---see Theorem \ref{thm main 2}.  In this case, solutions have to be essential sets in the reduced normal form of the game.  

A feature that both essentiality and stability share is that they insist on an algebraic-topological version of robustness---essentiality in (co)homology---which is much stronger than the usual $\e$-$\delta$ version. The proofs of our two main theorems, of course, lay out the logic behind why the axioms imply this form of robustness. Here, we offer a heuristic explanation.  As Mertens \cite{M1991} points out in his concluding remarks in Section 7, essentiality allows us to solve many fixed-point problems, the existence of a backward induction equilibrium being an exemplar of such a problem. Our finding, both in GW and here, is that more than just being exemplary, backward induction gives us the canonical problem to solve.  The invariance axioms generate a whole family of such problems (one for each larger game we obtain) that is so rich that its solution represents a sufficient condition for essentiality.   

The main point of divergence between essentiality and stability is the difference in their domains of perturbations.  Essentiality admits all possible payoff perturbations of the game, while stability allows strategy perturbations, which correspond to a much smaller class of payoff perturbations. Comparisons of their respective invariance axioms allow us to recast this technical difference between the concepts in a decision-theoretic framework. Intuitively speaking, when we enlarge a game $G$ to a larger game $\bar G$ with more strategies or players, stability requires the solutions of the two games to be the same when there is a one-to-one correspondence between the beliefs---both on- and off-the-equilibrium-path---and the preference orderings of actions of the original players in the two games, whereas essentiality does not insist on this relation ``away'' from the face of the strategy space of $\bar G$ that corresponds to the original game $G$.    Even without the axioms, we preferred stability over essentiality since the latter violates admissibility, which we view as fundamental.  The axiomatic differences merely solidify our preference for stability.  That said, an axiomatization of essentiality is of independent interest and may appeal to other game theorists. Besides, it is also possible that further axioms might result in selecting subsets of essential sets (e.g., stable sets residing in essential components). For these reasons, we provide a parallel treatment of essentiality in this paper.

While the results of this paper can be seen as providing a systematic reason for using stability in selecting equilibria in generic games, the fact remains that 
verifying the conditions of stability in applications is a forbidding task: one needs to construct the graph of equilibria over perturbed games and then check topological and algebraic properties of the graph and its projection to perturbations. However, there are at least a few ways in which, in most cases, we can approximate the results of using stability, especially in generic extensive-form games.  Firstly, applying properties like invariance, iterated elimination of dominated strategies, and forward induction can already be successful in whittling down the set of potential solutions.  Secondly, by focusing attention on connected sets of admissible equilibria, we could narrow the selections further.  Thirdly, a context-specific set of perturbations can be used to test whether one or more of these selections could be eliminated. De Sinopoli et al. \cite{DDL2006} use a combination of these ideas in their analysis. Similarly, an approach using the second and third suggestions was useful in the analysis of the Chain Store Paradox in Govindan \cite{G1995}.

The rest of the paper is organized as follows.  Section 2 lays out the notation that we use and also the definitions of the key concepts. Sections 3 through 5 describe the axioms.  Section 6 states the two main results of the paper, for stability and essentiality. The result for stability is proved in Section 9, while the proof of the other result is in the online appendix. Sections 7 and 8 do the required preparatory work for Section 9: the former gives a characterization of consistent beliefs for extensive-form games, which is used in the latter to give a characterization of stability for generic extensive-form games.  The proofs of the results in these two sections are also in the appendix.  Finally, Section 10 offers a few concluding remarks concerning the question of where we go from here. 

\section{Finite Games}

A finite $N$-player game $(\N, S, G)$ in normal form is specified as follows.  The set of players is $\N = \{\, 1, \ldots, N \, \}$. For each player $n$, his set of pure strategies is a finite set $S_n$.   The payoff function is $G: S \to \Re^{\N}$, where $S = \prod_n S_n$.  Quite often, we refer to such a game simply by its payoff function $G$. We use $\S_n$ to denote the set of $n$'s mixed strategies, and $\S = \prod_n \S_n$. 

We denote by $\G$ the game tree of a game in extensive form with perfect recall. For each player $n$, $H_n$ is the collection of his information sets. For each $n$ and $h_n \in H_n$,  $A_n(h_n)$ is the set of actions available at $h_n$.  Assume that the actions are all labeled differently across all $n$, $h_n$. For each $a_n$,  let $h_n(a_n)$ be the information set such that $a_n \in A_n(h_n)$, and let $A_n = \cup_{h_n} A_n(h_n)$ and $A = \cup_n A_n$. The payoff function is $u: Z \to \Re^N$, where $Z$ is the set of terminal nodes of $\G$.
$\G(u)$ refers to the extensive-form game with payoff function $u$.  We say that $u$ is in ``generic'' position if it lies outside a lower-dimensional semialgebraic subset of $\Re^{N \times Z}$.\footnote{In proving our main results, we invoke genericity in several intermediary steps, and each time the set where the property of genericity holds is different; however, as there are only finitely many such instances, it is clear that everything in this paper holds for games in an open set whose complement is a closed, lower-dimensional semialgebraic set.}

A pure strategy for player $n$ in an extensive-form game $\G$ is a function that assigns to each $h_n \in H_n$ an action $a_n \in A_n(h_n)$. We again use $G$ to represent the normal form of the extensive-form game $\G(u)$ with pure strategies $S_n$  for each $n$.  A behavioral strategy for $n$ is a function that assigns to each $h_n \in H_n$ a probability distribution over $A_n(h_n)$. Let $B_n$ be the set of $n$'s behavioral strategies, and $B = \prod_n B_n$.
A mixed strategy induces a payoff-equivalent behavioral strategy, and vice versa, by Kuhn's Theorem \cite{K1953}.

A third representation of strategies---which has the dimensional advantage of behavioral strategies over mixed strategies while retaining the linearity of payoffs that mixed strategies enjoy---invokes the concept of enabling strategies (cf.~Govindan and Wilson \cite{GW2002}) which we now describe.  Say that $a_n \in A_n$ is a last action of player $n$ if it is the last action for $n$ on some path from the root of the tree to a terminal node.  Let $L_n$ be the set of last actions of player $n$.  For each $a_n \in L_n$, let $S_n(a_n)$ be the set of pure strategies of $n$ that choose all actions leading up to, and including, $a_n$.  Then there is a linear map $\pi_n: \S_n \to {[0, 1]}^{L_n}$ given by $\pi_{n, a_n}(\s_n) = \sum_{s_n \in S_n(a_n)} \s_{n, s_n}$ for all $a_n \in L_n$. 
$P_n = \pi_n(\S_n)$ is the set of $n$'s enabling strategies, and $P = \prod_n P_n$. 
Each $P_n$, and thus $P$, is a compact and convex polyhedron.  
For each $n$ and $z$, let $a_n(z)$ be the last action of $n$ on the path to $z$, in case $n$ has at least one action on this path; and for $p_n \in P_n$, we use $p_n(a_n(z))$ to denote the probability of $a_n(z)$ under $p_n$, with the convention that $p_n(a_n(z)) = 1$ if $n$ has no action on this path.  Use $p_0(z)$ to denote the product of the probabilities of Nature's choices on the path to $z$, letting this probability be $1$ if no such choice exists.  Given $p$, for each $n$ and $z$, let
$p_{-n}(z) = p_0(z)\prod_{m \neq n}p_m(a_m(z))$.  The game $\G(u)$ can be defined in strategic form over $P$ with player $n$'s payoff from a profile $p$ being $\sum_{z \in Z}p_{-n}(z)p_n(a_n(z))u(z)$---see Pahl \cite{P2023}.

We now provide the definitions of three refinements of Nash equilibria that are important in this paper.  The first is the concept of quasi-perfect equilibria, due to van Damme \cite{vD1984}, that is invoked in our Backward Induction Axiom.  The other two concepts, which we characterize axiomatically, are essentiality and stability. 

Given a completely mixed behavioral strategy profile $b_{-n}$ of $n$'s opponents and an information set $h_n$ of $n$, say that an action $a_n$ at $h_n$ is optimal against $b_{-n}$ if at $h_n$ there exists a continuation strategy that chooses action $a_n$ at $h_n$ and is optimal against $b_{-n}$.  

\begin{definition}\label{def quasi-perfection}
	A behavioral strategy profile $b$ is a \emph{quasi-perfect equilibrium} of a game $\G(u)$ if there exists a sequence $b^k$ of completely mixed behavioral strategies converging to $b$ such that for each $n$ and $h_n$, $b_n$ assigns positive probability only to actions in $A_n(h_n)$ that are optimal against $b_{-n}^k$ for all $k$.  
\end{definition}

Given a game $G$ and a component $\S^*$ of Nash equilibria of $G$, we can assign an index, denoted $\text{Ind}(\S^*)$, which is a measure of its robustness to perturbations---see G\"{u}l et al.~\cite{GPS1993} or Ritzberger \cite{R1994}. This index is essentially independent of the map used in its definition (DeMichelis and Germano \cite{DG2000}).  In fixed-point theory there is an equivalence between robustness of fixed points and nonzero index (O'Neill \cite{O1953}).  Govindan and Wilson \cite{GW2005} show that this is ``almost'' the case with games: if we allow for the addition of duplicate strategies and then consider the perturbations of those as well, we obtain the concept of uniform hyperstability, which is equivalent to nonzero index. Recently,  Pahl and Pimienta \cite{PP2024, PP2026} have proved an exact equivalence for generic games in extensive form; that is, a component has a nonzero index iff it is hyperstable in the sense of KM.\footnote{ We conjecture that the equivalence would not extend to all finite games and thus that there is a ``gap'' between these two robustness ideas in game theory.}  The first of these two papers deals with the two-player case and the second one generalizes the result to $N$ players. 

To avoid multiplying terminology, we simply use the term essentiality rather than uniform hyperstability in the following definition.

\begin{definition}\label{def essentiality}
	A component of equilibria is \emph{essential} if its index is nonzero. 
\end{definition}

We turn now to the concept  of stability, as reformulated by Mertens.  As a lead-up to it, we present the original definition of Kohlberg and Mertens---in a slightly different form from theirs, to make the link with Mertens's definition clearer---followed by a strengthening before ending up with the final definition.

Given a normal form game $G$, for each $0 < \e \le 1$, let $T_\e$ be the set of all vectors $\d\t$, with $0 \le \d \le \e$ and $\t \in \S$, and let $\partial T_\e$ be its boundary.  For each $\d\t \in T_1$, we have a perturbed game $G(\d\t)$ where each player $n$ is constrained to choosing $\s_n \in \S_n$ such that $\s_n \ge \d\t_n$.  Let $E$ be the graph of equilibria over $T_1$, i.e.,~the set of $(\eta, \s) \in T_1 \times \S$ such that $\s$ is an equilibrium of the perturbed game $G(\eta)$.  Let $\proj: E \to T_1$ be the natural projection. For each $X \subseteq E$ and $\e > 0$, let $(X_\e, \partial X_\e) = \proj^{-1}(T_\e, \partial T_\e) \cap X$ and let $X_0 = \proj^{-1}(0) \cap X$. 

\begin{definition}
	$\S^*$ is KM-stable if it is minimal with respect to the following property.  For each $\e > 0$, there exist $\d > 0$ and a closed subset $X_\e$ of $E$ such that: 
	\begin{enumerate}
		\item $\{\, \s \in \S \mid (\eta, \s) \in X_\e  \, \text{for some} \, \eta  \in T_\e\}$ is contained in the $\d$-neighborhood of $\S^*$;
		\item The projection map from $X_\e$ to $T_\e$ is surjective.
	\end{enumerate}
\end{definition}

The $\e$-$\delta$ definition of stability captures the requirement that a solution be robust to strategy perturbations of the game $G$.  However, this notion of robustness is not enough to obtain game-theoretic properties like backward induction. Proving that a solution has such properties, typically, is equivalent to showing that certain fixed-point problems have solutions.  One way to guarantee that is to replace surjectivity of the projection from $X_\e$ in the definition with the following fixed-point property: every continuous map $g: X_\e \to T_\e$ has a point of coincidence with $\proj$, i.e., there exists $(\eta, \s) \in X_\e$ such that $\eta = g(\eta, \s)$. This property is a generalization of the fixed-point property for a space $X$, which is that every continuous self-map of $X$ has a fixed point. (If $X_\e = T_\e$ and $\proj$ is the identity, that is exactly what we would get.) If $\proj: X_\e \to T_\e$ has the fixed-point property, then it can be shown that there is an $\e$-proper equilibrium that is within $\d$ of $\S^*$---see Theorem 4 of Section 3 of \cite{M1989}---and thus $\S^*$ has a proper equilibrium.  The examples where KM-stable sets do not contain a sequential equilibrium are instances where this fixed-point property does not hold.

The fixed-point property above  is equivalent to saying that the projection from $X_\e$ to $T_\e$ is essential in homotopy, i.e., the projection map cannot be deformed to a map to the boundary $\partial T_\e$: equivalently, there does not exist a continuous map $h$ from $X_\e$ to $\partial T_\e$ such that the restriction of $h$ to $\partial X_\e$ equals the projection from $\partial X_\e$ to $\partial T_\e$.  

Even a definition using essentiality in homotopy  does not give us other game-theoretic properties like invariance under addition of duplicate strategies.  These considerations motivated Mertens to further strengthen essentiality, by requiring essentiality in cohomology. This requirement says that for all generic perturbations $\eta \in T_\e \backslash \partial T_\e$, the perturbed game has the same nonzero number of equilibria in $X_\e$, modulo orientation, which in some sense gives us an index for the stability of $X_\e$.  When $X_\e$ and $T_\e$ have the same dimension, this distinction makes no difference (as happens with Proposition \ref{prop Q}). But for truly nongeneric equilibrium correspondences the two notions of essentiality yield different solutions.

Apart from strengthening essentiality, Mertens dispensed with the minimality condition of the definition of KM stability, replacing it with a connectedness  requirement  by asking for  $X_\e \backslash \partial X_\e$ to be connected and dense in $X_\e$. 

One final point to consider when using cohomology is the specification of the coefficient module, the group or ring or field that is used for the algebraic count of the number of equilibria described above---we could do the count in integers, or real numbers, or integers mod 2, etc. The decomposition property for solution concepts---i.e., the solutions of a composition of independent games are the product of their solutions---forces us to use coefficients in a field with characteristic $p$ (which is either prime or zero). When $p$ is zero, we get the field $\mathbb{Q}$ of rational coefficients (one could equivalently use $\mathbb{R}$), while using a prime $p$ means we are using $\mathbb{Z}_p$, the set of integers modulo $p$.  Of these, the $0$-stable sets seem to be the most convincing ones: heuristically, if $p$ is, say, two, it would rule out stable sets with index $2$ and there seems to be no intuitive reason for doing so. For generic extensive-form games, any $p$-stable set is already a $0$-stable set. And since for each prime $p$, there exist $0$-stable sets that are not $p$-stable, our main theorem holds only for $0$-stability. Throughout this paper, stability refers to $0$-stability.  

Assembling all the above ideas together,  we have the following definition.

\begin{definition}\label{def stability}
	$\S^* \subseteq \S$ is a \emph{stable set} if there exists $X \subseteq E$ such that:
	\begin{enumerate}
		\item $\S^* = \{\, \s \mid (0, \s) \in X \,\}$;
		\item connexity: for each neighborhood $V$ of $X_0$ in $X$, there is a connected component of $V \backslash \partial X_1$ whose closure is a neighborhood of $X_0$ in $X$.
		\item essentiality: $\proj: (X_\e, \partial X_\e) \to (T_\e, \partial T_\e)$ is essential in \v{C}ech cohomology with rational coefficients.
	\end{enumerate}
\end{definition}

The axiomatic approach to refinements, as envisaged by KM, aims to select from the Nash equilibria of a game, using a set of axioms of rationality. From this viewpoint, we have the following definition for a solution concept.

\begin{definition}\label{def solution concept}
	A {\it solution concept} $\varphi$ is a correspondence that assigns to each finite game $G$ a nonempty collection $\varphi(G)$ of closed and connected subsets of the Nash equilibria of $G$, called {\it solutions}.  
\end{definition}

Our focus in this paper is on two solution concepts. The first is called essentiality, which assigns to each game $G$ its essential components of equilibria; and the second is called stability, which assigns to each game $G$ its stable sets. The next three sections describe the axioms we impose on a solution concept $\varphi$.

\section{Backward Induction}
The Backward Induction Axiom states that a solution should contain a ``backward induction'' equilibrium. As KM argue, the concept of a sequential equilibrium captures backward induction rationality in games with imperfect information and is thus the right notion of a backward induction equilibrium.  Ideally, we would like to state the axiom using sequential equilibria.  Unfortunately, our method of proof requires us to use quasi-perfect equilibria in the axiom.\footnote{When payoffs are generic, all sequential equilibria are quasi-perfect (see \cite{HKS2002, PS2014}).  However, that equivalence is not useful here, as the embeddings we construct could be nongeneric games.}

Our desire to invoke sequential equilibria to capture the backward induction property stems from our view that it is the mildest generalization of subgame perfection, which in turn generalizes backward induction in perfect information games.  However, as a refinement, quasi-perfection is superior to sequential equilibrium, since it yields sequential equilibria where a player's continuation strategy at each information set is admissible. Figure \ref{fig:figure_3_1} displays an example from Mertens \cite{M1995} that makes this point.
	
	\begin{figure}[htbp]
		\centering
		\begin{tikzpicture}[scale=1.2, thick]
			\coordinate (root) at (0, 0);
			\coordinate (top_node) at (3, 1.8);
			\coordinate (bot_node) at (3, -1.8);
			
			\coordinate (t1) at (5, 2.7);
			\coordinate (t2) at (5, 0.9);
			\coordinate (t3) at (5, -0.9);
			\coordinate (t4) at (5, -2.7);
			
			\draw (root) -- node[above left, midway] {$U$} (top_node);
			\draw (root) -- node[below left, midway] {$D$} (bot_node);
			
			\draw (top_node) -- node[above left, midway] {$u$} (t1);
			\draw (top_node) -- node[below left, midway] {$d$} (t2);
			
			\draw (bot_node) -- node[above left, midway] {$u$} (t3);
			\draw (bot_node) -- node[below left, midway] {$d$} (t4);
			
			\node[left=3pt] at (root) {$I$};
			\node[above=3pt] at (top_node) {$I$};
			\node[below=3pt] at (bot_node) {$II$};
			
			\node[right=3pt] at (t1) {$1, 1$};
			\node[right=3pt] at (t2) {$0, 0$};
			\node[right=3pt] at (t3) {$1, 1$};
			\node[right=3pt] at (t4) {$0, 0$};
		\end{tikzpicture}
		\caption{}
		\label{fig:figure_3_1}
	\end{figure}
	
In this game, the only quasi-perfect outcome is for player I to play $U$ and then $u$.  But the outcome $(D, u)$ is also a sequential equilibrium outcome.   As Mertens points out, the latter is also extensive-form perfect, which makes quasi-perfection superior to it. In fact, in the same paper,  Mertens  constructs another example where the set of trembling-hand perfect equilibria consists entirely of profiles that are equivalent to weakly dominated strategies and is disjoint from the set of quasi-perfect equilibria.\footnote{This example implies that stability would not satisfy a reformulated Axiom B that uses trembling-hand perfection.  However, essentiality satisfies either version of Axiom B and Theorem \ref{thm main 2} goes through if we use trembling-hand perfection.}  

Given the above considerations, we formalize  Axiom B  as follows.

\begin{axiom}[Axiom B]\label{axm B}
	If $G$ is the normal form of an extensive-form game $\G$ with perfect recall, then every solution contains a strategy profile that is equivalent to a quasi-perfect equilibrium of $\G$.  
\end{axiom}

\section{Addition of Dominated Strategies}
The Axiom of Iterated Dominance figures prominently in the list of desiderata drawn up by KM  for stable sets.  The axiom states that a stable set of a game $G$ contains a stable set of any game $G'$ obtained by the deletion of (weakly) dominated strategies. KM consider, but eventually drop, a dual requirement on the addition of a dominated strategy. Their reasoning is that such an axiom would be inconsistent with admissibility, as the following example, which is taken from KM, demonstrates.  In the two-player game:
\[
\left(
\begin{array}{ll}
	3,2 & 2,2 
\end{array}
\right)
\]	
every strategy profile  is reasonable as a solution.  But when we add a strictly dominated strategy:
\[
\left(
\begin{array}{ll}
	3,2 & 2,2 \\
	1,1 & 1,0 
\end{array}
\right)
\]	
only the (Top, Left) equilibrium is admissible. Similarly, (Top, Right) is the unique admissible equilibrium of another game obtained by adding a strictly dominated strategy.

We propose two counterarguments that form the bases for two versions of a related axiom concerning the addition of dominated strategies. (1) One could drop the axiom of admissibility. (2) Alternatively, one could restrict the class of games obtained by the addition of dominated strategies.  We elaborate on these ideas in turn, after we consider the more general problem of the relationship between the solutions of any two games where one is obtained from the other by adding strategies. Studying the problem at this abstract level will hopefully persuade the reader of the reasonableness of our axiom about dominated strategies and also of how it naturally leads to the formulation of the invariance axiom in the next section, where we add players to the game.

Consider the following question.  Suppose we have two games $(\N, S, G)$ and $(\N, \bar S, \bar G)$ with $S \subsetneq \bar S$ and the restriction of the payoff function $\bar G$ to $S$ is $G$;  under what conditions could we require axiomatically that the two games have the same solutions?  If for each $n$, the strategies in $\bar S_n \backslash S_n$ are duplicates of mixed strategies in $\S_n$, then the axiom of invariance, as postulated by KM, says that the games should have the same solutions---see also Mertens \cite{M2003} for a discussion of ordinality in games.  But what if the additional strategies are not duplicates? Before we proceed, as a matter of convention, when $(\N, \bar S, \bar G)$ is obtained from $(\N, S, G)$ by the addition of strategies, there is a natural inclusion $i: \S \to \bar \S$,  where we set the probability of the strategies not in $S$ to be zero.  This inclusion allows us to view $\S$ as a subset of $\bar \S$, which is how we treat these sets in this section. 

Return now to the question concerning the relationship between $(\N, S, G)$ and $(\N, \bar S, \bar G)$ when the latter is not obtained from the former by the addition of duplicate strategies. Since we are dealing with the question of selecting from Nash equilibria, we must at least insist on Property P1 below. 

\smallskip

\noindent {\bf Property P1:} {\slshape The two games $G$ and $\bar G$ have the same set of Nash equilibria.}

\smallskip
As a matter of clarification, concerning our convention about the inclusion of $\S$ in $\bar \S$, formally what Property P1 says is that $\bar \s$ is a  Nash equilibrium of $\bar G$ iff it is of the form $i(\s)$ where $\s$ is a Nash equilibrium of $G$. 

\begin{figure}[htbp]
	\centering
		\begin{tikzpicture}[auto, scale=1.5]
			\filldraw (0,2) circle (1pt);
			\draw (0,2) -- (2,2);
			\draw (0,2) -- (1,1);
			\filldraw (2,2) circle (1pt);
			\draw (2,2) -- (3,3);
			\draw (2,2) -- (4,2);
			\filldraw (4,2) circle (1pt);
			\draw (4,2) -- (5,3);
			\draw (4,2) -- (5,1);
			\draw (5,3) -- (6,3.5);
			\draw (5,3) -- (6,2.5);
			\draw (5,1) -- (6,1.5);
			\draw (5,1) -- (6,0.5);
			\draw (0,2.2) node{$1$};
			\draw (2,2.2) node{$2$};
			\draw (4,2.2) node{$1$};
			\draw (4.8,2) node{$2$};
			\draw[dotted] (5,1) -- (5,3);
			\filldraw (5,1) circle (1pt);
			\filldraw (5,3) circle (1pt);
			\draw (1.4,1) node{$3,6$};
			\draw (3.4,3) node{$4,3$};
			\draw (6.4,3.5) node{$1,1$};
			\draw (6.4,2.5) node{$2,4$};
			\draw (6.4,1.5) node{$2,2$};
			\draw (6.4,0.5) node{$1,1$};
		\end{tikzpicture}
	\caption{}	
	\label{fig two outside options}		
\end{figure}

Using the techniques of Govindan et al.~\cite{GLP2023}, we can construct games $G$ and $\bar G$ with the same set of equilibria, but where the indices of some of the equilibria are changed.  Thus, a solution concept that picks equilibria with a positive index would violate an axiom that relied on Property P1 alone.  A more concrete example that we now describe also points to a need for something stronger than P1.  The game $G$ with the extensive form in Figure \ref{fig two outside options} gives the two players, by turns, an outside option; and if neither player exercises his, then they play a $2 \times 2$ game. $G$ has two components of equilibria that yield the payoffs $(3,6)$ and $(4,3)$, respectively.   Only the outcome $(3,6)$ survives the iterative elimination of weakly dominated strategies: we first eliminate (In, Up) for player 2, then (In, Down) for player 1, then Out for player 2, and finally we end up with $(3,6)$.  Therefore, the component giving  $(3,6)$ has index $+1$, and the second outcome has index zero.

Now add a move for player 1 to obtain the game $\bar G$ in Figure \ref{fig two outside options again}. This game has two components as well, with the same outcomes.  But now $(4,3)$ is the outcome that survives iterative elimination of dominated strategies:  (In, Up) is dominated by a $(1/3, 2/3)$-mixture of  (In, Middle) and (In, Down) for player 1; eliminating that allows us to eliminate both strategies of player 2 that play In, leading finally to the outcome $(4,3)$.  Thus, we have that $(4,3)$ has  index $+1$, while $(3,6)$ has index zero. Adding a strategy for player 1 has reversed the indices of the two outcomes, thus reversing their essentiality.

Of course, even though the equilibrium payoffs of the two games are the same, the game $\bar G$ has a strictly larger set of equilibria and, therefore, does not provide conclusive proof that essentiality would require an axiom based on more than P1. But the example is suggestive of the need for a stronger link between $G$ and $\bar G$. A connection between $G$ and $\bar G$ stronger than P1 exists when the following additional property holds.

\smallskip

\begin{figure}
	\centering
		\begin{tikzpicture}[auto, scale=1.5]

\filldraw (0,2) circle (1pt);
\draw (0,2) -- (2,2);
\draw (0,2) -- (1,1);
\filldraw (2,2) circle (1pt);
\draw (2,2) -- (3,3);
\draw (2,2) -- (4,2);
\filldraw (4,2) circle (1pt);
\draw (4,2) -- (5,3);
\draw (4,2) -- (5,2);
\draw (4,2) -- (5,1);
\draw (5,3) -- (6,3.3);
\draw (5,3) -- (6,2.7);
\draw (5,2) -- (6,2.3);
\draw (5,2) -- (6,1.7);
\draw (5,1) -- (6,1.3);
\draw (5,1) -- (6,0.7);
\draw (0,2.2) node{$1$};
\draw (2,2.2) node{$2$};
\draw (4,2.2) node{$1$};
\draw (4.8,2.4) node{$2$};
\draw[dotted] (5,1) -- (5,3);
\filldraw (5,1) circle (1pt);
\filldraw (5,2) circle (1pt);
\filldraw (5,3) circle (1pt);
\draw (1.4,1) node{$3,6$};
\draw (3.4,3) node{$4,3$};
\draw (6.4,3.3) node{$1,1$};
\draw (6.4,2.7) node{$2,4$};
\draw (6.4,2.3) node{$2,2$};
\draw (6.4,1.7) node{$1,1$};
\draw (6.4,1.3) node{$1,1$};
\draw (6.4,0.7) node{$3,0$};
\end{tikzpicture}
\caption{}
\label{fig two outside options again}
\end{figure}

\smallskip

\noindent {\bf Property P2:} {\slshape The restriction of the best-reply correspondence of $\bar G$ to $\S$ is the best-reply correspondence of $G$.}

\smallskip

P2 is satisfied when, for instance, all the strategies that are added are strictly dominated (in which case P1 holds as well).  A weaker assumption than P2 would be the following.

\smallskip

\noindent {\bf Property P3:} {\slshape The best-reply correspondence of $G$ is a selection of the restriction of the best-reply correspondence of $\bar G$ to $\S$.}

\smallskip

P3 holds when, for example, the added strategies are weakly dominated (but now P3 is independent of P1).  {\it Prima facie}, it seems reasonable to impose an axiom that says that when P1 and P3 hold, $G$ and $\bar G$ have the same solutions.  Essentiality does satisfy such an axiom. Rather than imposing that axiom directly, we choose to impose a weaker one, asking for P1 to hold when we have added weakly dominated strategies.  Before stating the axiom, we formalize what it means to add a dominated strategy.

\begin{definition}\label{def ads}
	A game $(\N, \bar S, \bar G)$ is obtained from $(\N, S, G)$ by the addition of (weakly) dominated strategies if:
	\begin{enumerate}
		\item $S \subsetneq \bar S$, and for each $n$, the strategies in $\bar S_n \backslash S_n$ are (weakly) dominated;
		\item the restriction of $\bar G$ to $S$ is $G$.
	\end{enumerate}
\end{definition}	

To restate what we observed earlier, if $\bar G$ is obtained from $G$ by the addition of strictly dominated strategies, both P2 and P1 hold.  When some of the added strategies are only weakly dominated, P3 holds but P1 might not: it is possible that $\bar G$ has a strictly bigger set of Nash equilibria.  For each solution of $G$ one could insist that there be a solution of $\bar G$ that contains it (the reverse direction of the one proposed by KM).  We do not know the implications of this axiom. For example, we do not know if essentiality satisfies it, or if it is even consistent with existence.   However, when the addition of weakly dominated strategies does not expand the set of Nash equilibria, i.e.,~when P1 holds, the picture is clearer and our Axiom D imposes this restriction. 

\begin{definition}\label{def waids}
	$(\N, \bar S, \bar G)$ is obtained from $(\N, S, G)$ by the addition of weakly irrelevant dominated strategies if:
	\begin{enumerate}
		\item $\bar G$ is obtained from $G$ by the addition of weakly dominated strategies;
		\item $\bar G$ and $G$ have the same set of Nash equilibria.
	\end{enumerate}
\end{definition}

Condition (1) of Definition \ref{def waids} implies P3, and its condition (2) is P1. With this definition of weakly irrelevant dominated strategies, we have the following axiom.

\begin{axiom}[Axiom D] \label{axm D plain}
	$\varphi(G) = \varphi(\bar G)$ whenever $\bar G$ is obtained from $G$ by the addition of weakly irrelevant dominated strategies.
\end{axiom}

As we show in Appendix A, essentiality satisfies Axiom D. Of course, as the example of KM discussed at the beginning of this section shows, Axiom D comes at the expense of admissibility. Our own view is that admissibility is the more important of the two axioms as an expression of rationality, and therefore we are not entirely satisfied with Axiom D.  There are other grounds to be skeptical of this axiom.  Suppose $\bar G$ is obtained from $G$ by the addition of strictly dominated strategies.  The restriction of the best-reply correspondence of $\bar G$ to $\S$ is the best-reply correspondence of $G$.  But what is the behavior of the correspondence on $\bar \S \backslash \S$?  This set contains the set of all completely mixed strategy profiles, which generate the set of consistent beliefs when $\bar G$ has a nontrivial extensive form.  Absent a way to relate such profiles in $\bar G$ with profiles in $G$, there is no meaningful way to compare beliefs in the two games.  Applications have taught us that identification of the ``rational'' outcomes among all Nash outcomes in games quite often turns on the beliefs supporting these outcomes. In not allowing us to create a correspondence between beliefs in equivalent games, Axiom D relies on an equivalence relation among games that is too ``coarse.'' This inability to relate beliefs between the two games is the game-theoretic cost of working with the much richer set of payoff perturbations.   

The second approach to getting around the problem identified by KM is motivated by the criticism made in the last paragraph. What is required is a refinement of the notion of equivalent games, defined by combining P1 and P3, that would make beliefs across equivalent games comparable, thus making preference rankings over strategies in $S_n$ comparable. We propose one approach to formulating such an equivalence now.

 Suppose for each player $n$, there is a map $f_n: \bar S_n \to \S_n$ such that $f_n(s_n) = s_n$ for each $s_n \in S_n$; that is, $f_n$ maps each pure strategy in the original game to itself and maps every strategy in $\bar S_n \backslash S_n$ to a mixed strategy $\s_n \in \S_n$.  The map $f_n$ extends to a map from $\bar \S_n$ to $\S_n$ by linear interpolation, and for notational simplicity we refer to the extension by $f_n$ as well.  Let $f = \prod_n f_n: \bar \S \to \S$ be the product of these maps.  Suppose that player $n$ believes his  opponents are playing according to $\bar \s_{-n} \in \bar \S_{-n}$. Then $f_{-n}(\bar \s_{-n})$ is his belief about how his opponents are playing in $G$; we can now ask that his preferences in $\bar G$ over his strategies in $S_n$ are the same as those in $G$ against $f_{-n}(\bar \s_{-n})$.  This gives us the following property.


\smallskip

\noindent {\bf Property P4:} {\slshape For each $n$, there is a  map $f_n: \bar S_n \to \S_n$ such that: (a) $f_n(s_n) = s_n$ for all $s_n \in S_n$;  (b) for each $m$, denoting still by $f_m$ the linear extension of $f_m$ to $\bar \S_m$, for each $n$, $\s_n, \t_n \in \S_n$ and $\bar \s_{-n}$, $\bar G_n(\s_n , \bar \s_{-n}) \ge \bar G_n(\t_n, \bar \s_{-n})$ iff  $G_n(\s_n, f_{-n}(\bar \s_{-n})) \ge G_n(\t_n, f_{-n}(\bar \s_{-n}))$.}

\smallskip


We could have just asked for continuous maps $f_n: \bar S_n \to \S_n$ for each $n$, rather than linear maps, but the linearity of $f$ is an implication of P4 holding universally assuming just continuity.  Specifically, suppose we have continuous maps $(f_n)$ that satisfy condition (a) of Property P4  and such that for each payoff function $G$ there is a $\bar G$ for which condition (b) of Property P4 holds as well.\footnote{If there exists a $\bar G$, then there exists one that is obtained from $G$ by the addition of weakly dominated strategies.}  It is then easy to show that the $f_n$'s must be linear in $\s_n$. Thus, if the maps $f_n$ are to serve ``globally'' on the space of games (with strategy sets $\bar \S$ and $\S$) as a way to generate equivalence classes, then we get linearity as a consequence.

To highlight what P4 entails, consider the example from KM discussed at the beginning of the section. Suppose we have the following payoffs when we add a dominated strategy for player 1:
\[
\left(
\begin{array}{ll}
	3,2 & 2,2 \\
	1,a & 1,b 
\end{array}
\right)
\]	
Regardless of what $a$ and $b$ are, this new game is obtained from the original game by the addition of weakly irrelevant dominated strategies, in the sense of Definition \ref{def waids}. However, when $a = b =2$ (or even simply $a = b$) then,  from player 2's perspective, the newly added strategy is just a duplicate of 1's strategy in the original game; his payoffs have been duplicated, or at least his preference ranking is preserved; and all beliefs that can be generated in the new game can be generated in the original game (and vice versa).  It is reasonable to insist that the two games have the same set of solutions under the additional restriction on the payoffs of player 2.\footnote{Section 4, p. 710 of Mertens \cite{M1991} makes this argument: if we duplicate a strategy for a player and then subtract a constant from it just for this player, to make it strictly dominated, the resulting game should have the same solutions as the original game. De Stefano Grigis \cite{D2014} also considers the addition of strictly dominated strategies in his KM-stability-style refinement called G-stable sets.}  

One way to get around KM's objection, then, is to ask that $G$ and $\bar G$ have the same solutions when $\bar G$ is obtained from $G$ by the addition of weakly dominated strategies and P4 holds with the maps $f_n$ being linear.  However, we want to weaken two aspects of P4.  First,  we could allow $f_n$ to be a  function of $\bar s$ and not just $\bar s_n$; second, we could insist on condition (b) of P4 holding not globally, but in a neighborhood of $\S$ in $\bar \S$, since this set is what is relevant for beliefs at the equilibria (which, given P1, are in $\S$). Thus, we have:
\smallskip

\noindent {\bf Property P5:} {\slshape There exists a  map $f: \bar S \to \S$ such that: 
	\begin{enumerate}
		\item[(a)] $f(s) = s$ for all $s \in S$; 
		\item[(b)] there exists a neighborhood $\bar U$ of $\S$ in $\bar \S$ such that, denoting still by $f$ the multilinear extension of $f$ to $\bar \S$,  for each $n$, $\bar \s \in \bar U$, $\s_n, \t_n \in \S_n$, $\bar G_n(\s_n , \bar \s_{-n}) \ge  \bar G_n(\t_n , \bar \s_{-n})$ iff $G_n(\s_n, f_{-n}(\bar \s)) \ge  G_n(\t_n, f_{-n}(\bar \s))$.
	\end{enumerate}
}
\smallskip

The map $f$ creates a ``dictionary'' that identifies for each $\bar \s \in \bar \S$ an equivalent profile in $G$.\footnote{We have imposed the assumption of multilinearity in P5.  As with P4, it is a consequence of P5 holding over a ``large'' domain of games.  Given a continuous function $f$, suppose for all $G$ in a linear subspace of the space of games over $S$ generated by a game form, there is a $\bar G$ for which P5 holds, then we could replace $f$ with a multilinear function $g$ for which P5 holds as well.}  
It allows us to map beliefs in $\bar G$ to those in $G$ and to compare the preference orderings over a player's strategies against equivalent beliefs. Obviously, given the assumption of independence of the players' actions,  a player $n$'s beliefs about the actions of his opponents are independent of his choice. Given that we are allowing for $f_m$ to depend on $\bar \s_n$ for $m \neq n$, the inclusion of $\bar \s_n$ in the profile $\bar \s$ in Condition (b) of P5 makes sure that this dependence of $f_{-n}$ on $\bar \s_n$ is decision-theoretically irrelevant.  Unless the payoffs of the games $G$ and $\bar G$ are pathological, condition (b) of Property P5, therefore, would amount to saying that $m$'s choice is payoff irrelevant to $n$ if $f_m$ depends on $\s_n$.

Property P5 allows us to strengthen the requirements for obtaining a game $\bar G$ from $G$ by the addition of irrelevant strategies in the following definition.

\begin{definition}\label{def saids}
	$(\N, \bar S, \bar G)$ is obtained from $(\N, S, G)$ by the addition of strongly irrelevant dominated strategies if:
	\begin{enumerate}
		\item $\bar G$ is obtained from $G$ by the addition of weakly irrelevant dominated strategies;
		\item there exist a neighborhood $\bar U$ of $\S$ in $\bar \S$ and a multilinear function $f: \bar \S \to \S$ such that:
		\begin{enumerate}
			\item[(a)] the restriction of $f$ to $\S$ is the identity; 
			\item[(b)] for each $n$, $\bar \s \in \bar U$, $\s_n, \t_n \in \S_n$, $\bar G_n(\s_n , \bar \s_{-n}) \ge  \bar G_n(\t_n , \bar \s_{-n})$ iff $G_n(\s_n, f_{-n}(\bar \s)) \ge  G_n(\t_n, f_{-n}(\bar \s))$.
		\end{enumerate}   
	\end{enumerate}
\end{definition}

Suppose $\bar G$ is obtained from $G$ by the addition of weakly irrelevant dominated strategies.  What is the effect of imposing condition (2) of Definition \ref{def saids} as well?  Suppose we have an equilibrium profile $\s^*$ in $G$.  Consistent beliefs---or, more generally, conditional probability systems or lexicographic probability systems, which provide alternative hypotheses for how players might play if they deviate from $\s^*$---are constructed, when $G$ has a nontrivial extensive form, by taking sequences of completely mixed strategies converging to $\s^*$.  We could construct these beliefs both in $G$ and in $\bar G$. If $\bar G$ is to add no additional decision-theoretically relevant data for players to reason about the equilibrium $\s^*$, there has to be an equivalence of these beliefs, and preferences over actions given these equivalent beliefs should remain the same.   condition (2) of the definition delivers precisely this property. The function $f$, because it is defined on a neighborhood $\bar U$ of $\S$ in $\bar \S$, allows us to map sequences in $\bar \S$ converging to $\s^*$ to sequences that do the same in $\S$; condition (2b) of the definition says that preferences over pure strategies in $S$ against such sequences (and hence against the limiting beliefs) are preserved.  

Condition (2), therefore, is entirely about each player's individual decision problem concerning beliefs and preferences; in particular, it  does not contrive a fixed-point relationship between these games that could be exploited in the axiomatic characterization of stability. 

To better understand the definition, we provide some examples. As might be expected, there is a substantial difference between two-player games and $N$-player games with $N > 2$. Hence, we treat them separately.

\begin{example}
	Fix a two-player  game tree $\G$ and consider the space of all games $G$ with this tree.  Let $\bar S$ be obtained from $S$ by the addition of strategies. Given a multilinear function $f: \bar \S \to \S$, suppose for each $G$ there is a $\bar G$ for which P5 holds. Then, if $n$'s strategy is payoff relevant to $m$, $f_m(\bar s_m, \s_n)$ is independent of $\s_n$ for each newly added strategy $\bar s_m$ of $m$; it could still depend on $\bar s_n \notin S_n$, but that is payoff irrelevant.   Thus, from $n$'s perspective, each $\bar s_m \notin S_m$ corresponds to a duplicate in $G$.
	
	For example, consider the game $\bar G$
	\[
	\left(
	\begin{array}{lll}
		4,2 & 0,0 & 4,1 \\
		0,0 & 2,4 & 0,1 \\
		1,2 & 1,0 & 0,0
	\end{array}
	\right)
	\]	
	The last strategy of each player is dominated by a mixture of the other two; a player's payoff if he plays an undominated strategy against the dominated strategy of the other is what he would get if the other player were to play his first strategy. The payoffs when both play a dominated strategy are irrelevant.  If we let $G$ be the game obtained by the deletion of the last row and the last column, then $\bar G$ is obtained from $G$ by the addition of  strongly irrelevant dominated strategies. Again, by contrast, if we are considering only weakly irrelevant dominated strategies, the payoffs to player 2 (resp. 1) against the last row (resp. column)  could be completely arbitrary.
	
	To see how $D^\ast$ differs from a simple invariance axiom that asks for solutions to remain solutions after adding duplicate strategies, if the payoffs to  player 1 from the last row (resp. last column) are $(4, 0, 4)$ (resp. $(2, 0, 2)$), then we get an equivalent game.  Therefore, the only difference here is the payoffs of the player for whom we have added a strategy. $\Box$
\end{example}

\begin{example}\label{example bq game}
	By way of an example with more than two players, view the Beer-Quiche game as a three-player game G (as in Figure \ref{fig BQ}, where Weak, Strong, and Receiver are players 1, 2, and 3, respectively).

\begin{figure}
\centering
\begin{tikzpicture}[auto, scale=1.5]
\filldraw (4,2) circle (1pt);
\draw (4,2) -- (4,3.5);
\draw (4,2) -- (4,0.5);
\filldraw (4,3.5) circle (1pt);
\draw (2,3.5) -- (4,3.5);
\draw (4,3.5) -- (6,3.5);
\filldraw (2,3.5) circle (1pt);
\filldraw (6,3.5) circle (1pt);
\filldraw (4,0.5) circle (1pt);
\draw (2,0.5) -- (4,0.5);
\draw (4,0.5) -- (6,0.5);
\filldraw (2,0.5) circle (1pt);
\filldraw (6,0.5) circle (1pt);
\draw[dotted] (2,0.5) -- (2,3.5);
\draw[dotted] (6,0.5) -- (6,3.5);
\draw (2,0.5) -- (1,1);
\draw (2,0.5) -- (1,0);
\draw (2,3.5) -- (1,4);
\draw (2,3.5) -- (1,3);
\draw (6,0.5) -- (7,1);
\draw (6,0.5) -- (7,0);
\draw (6,3.5) -- (7,4);
\draw (6,3.5) -- (7,3);
\draw (4.3,2.75) node{$0.1$};
\draw (4.3,1.25) node{$0.9$};
\draw (4,3.8) node{Weak};
\draw (4,0.2) node{Strong};
\draw (2.5,2) node{Receiver};
\draw (5.5,2) node{Receiver};
\draw (1.5,4) node{F};
\draw (6.5,4) node{F};
\draw (1.5,1) node{F};
\draw (6.5,1) node{F};
\draw (1.5,3) node{R};
\draw (6.5,3) node{R};
\draw (1.5,0) node{R};
\draw (6.5,0) node{R};
\draw (0.6,4) node{$0,0,1$};
\draw (0.6,3) node{$2,0,0$};
\draw (0.6,1) node{$0,1,0$};
\draw (0.6,0) node{$0,3,1$};
\draw (7.4,4) node{$1,0,1$};
\draw (7.4,3) node{$3,0,0$};
\draw (7.4,1) node{$0,0,0$};
\draw (7.4,0) node{$0,2,1$};
\draw (3,3.7) node{Beer};
\draw (5,3.7) node{Quiche};
\draw (3,0.3) node{Beer};
\draw (5,0.3) node{Quiche};
\end{tikzpicture}
\caption{}
\label{fig BQ}
\end{figure}

Define a game $\bar G$ as follows. The strategy sets of Strong and Receiver are as in $G$. Weak has his two strategies as before but also a third strategy, denoted $s_w^0$. This strategy corresponds to the uniform mixture $\s_w^*$ over Quiche and Beer from the perspective of the other players, while for Weak it is dominated up to $\e > 0$ by mixing uniformly between Quiche and Beer.  $\bar G$ in extensive form is played as follows (see Figure \ref{fig BQ again}).    The game begins with Weak first deciding whether he wants to play  $s_w^0$ or not. Thus, in Figure \ref{fig BQ again} the game has as its root the node where Weak decides whether to play $s_w^0$. If he plays $s_w^0$, then the mixed strategy $\s_w^*$ is executed. If he chooses not to, then the moves of the game $G$ occur. Specifically, Nature chooses Weak and Strong to play with probabilities $.1$ and $.9$, respectively.  When chosen, Strong and Weak decide whether to play Beer or Quiche.  Player 3 then moves, having only observed Beer or Quiche.  

There is a multilinear map $f$ such that $\bar G$ is obtained from $G$ by the addition of strongly irrelevant dominated strategies.  The map is specified by the following values at a pure strategy profile $\bar s$.  $f_n(\bar s) = \bar s_n$ if $n$ is the Receiver or if $\bar s_w \neq s_w^0$.  Otherwise it is the uniform mixture over Beer and Quiche. The payoffs for the Receiver at the terminal nodes following Weak's choice of $s_w^0$ initially and Nature's choice of either Quiche or Beer  are .1 and .9 for $F$ and $R$, respectively.  These represent his expected payoffs in the original game if Weak and Strong had independently mixed uniformly over Beer and Quiche.  If we had used unrelated payoffs here, the game would still be obtained by the addition of weakly irrelevant dominated strategies, but not by the addition of strongly irrelevant ones. (Less important are Strong's payoffs here, as they do not affect his preferences over his actions.) 

The example highlights a couple of aspects of Property P5. First, in this example, unlike with P4, condition (b) of P5 cannot be reduced to a statement of the form: $\bar G_n(\s_n, \bar \s_{-n}) = G_n(\s_n, f_{-n}(\bar \s))$. Second, observe that the preference orderings are preserved as long as Weak does not play his dominated strategy; in other words, the neighborhood $\bar U$ of $\S$ is a strict subset of $\bar \S$.  $\Box$
\end{example}

The definition of strongly irrelevant dominated strategies now gives us the following variant of Axiom D, called Axiom D$^*$. As the set of embeddings is smaller, Axiom D$^*$ is weaker than Axiom D.

\begin{axiom}[Axiom D$^*$] \label{axm Dstar}
	$\varphi(G) = \varphi(\bar G)$ whenever $\bar G$ is obtained from $G$ by the addition of strongly irrelevant dominated strategies.
\end{axiom}

\begin{example}\label{example bq Axiom D}
	As a precursor to how Axiom D$^*$ and Axiom B are useful in eliminating outcomes that are not stable, let us return to Example \ref{example bq game}, and in particular consider the following game $\tilde G$ that has the same normal form as the game $\bar G$ of Figure \ref{fig BQ again}. In this game, there is no move of Nature.  Players Strong and Weak simultaneously decide first.  For Weak, first he decides whether to play Quiche or not.  If he decides not to play Quiche, then he decides between $s_w^0$ and Beer.  Strong just decides between Beer and Quiche.  The Receiver then moves, having just seen the signal of Beer or Quiche.   The payoffs for the Receiver in $\tilde G$ are computed by taking the expectations with respect to Nature's move. 
	
	Quiche is not a sequential equilibrium outcome in the game $\tilde G$ when $\e$ is sufficiently small.  Indeed, consider any equilibrium supporting Quiche.  At the node where Weak decides between Beer and $s_w^0$, the latter is strictly better for all small $\e$.  Then from 3's perspective, it is as if when he sees Beer, he believes it is Strong with a high probability and he would retreat, thus upsetting the equilibrium. 
	
	In this example, Strong's strategy is payoff-irrelevant to Weak, since they do not move in succession in $G$ (the same reasoning that made Strong's payoff following $s_w^0$ irrelevant to him).  For an arbitrary extensive-form game, however, we would not be able to use this feature.  What works is an addition not just of dominated strategies, but also of players, as considered in the next section about our final axiom.\footnote{The definition of G-stability in De Stefano Grigis \cite{D2014} requires the use of both the addition of strictly dominated strategies and irrelevant players in order to get a player-splitting property: in games like signaling games where different agents of a player do not move in succession the solutions of the normal and agent-normal form of the game are the same.}  $\Box$
\end{example}

\begin{figure}
	\centering
		\begin{tikzpicture}[auto, scale=1.5]
			\filldraw (4,2) circle (1pt);
			\draw (4,2) -- (4,3.5);
			\draw (4,2) -- (4,0.5);
			\filldraw (4,3.5) circle (1pt);
			\draw (2,3.5) -- (4,3.5);
			\draw (4,3.5) -- (6,3.5);
			\filldraw (2,3.5) circle (1pt);
			\filldraw (6,3.5) circle (1pt);
			\filldraw (4,0.5) circle (1pt);
			\draw (2,0.5) -- (4,0.5);
			\draw (4,0.5) -- (6,0.5);
			\filldraw (2,0.5) circle (1pt);
			\filldraw (6,0.5) circle (1pt);
			\draw[dotted] (2,0.5) -- (2,3.5);
			\draw[dotted] (6,0.5) -- (6,3.5);
			\draw (2,0.5) -- (1,1);
			\draw (2,0.5) -- (1,0);
			\draw (2,3.5) -- (1,4);
			\draw (2,3.5) -- (1,3);
			\draw (6,0.5) -- (7,1);
			\draw (6,0.5) -- (7,0);
			\draw (6,3.5) -- (7,4);
			\draw (6,3.5) -- (7,3);
			\draw (4.3,2.75) node{$0.1$};
			\draw (4.3,1.25) node{$0.9$};
			\draw (4,3.8) node{Weak};
			\draw (4,0.2) node{Strong};
			\draw (2.5,2) node{Receiver};
			\draw (5.5,2) node{Receiver};
			\draw (1.5,4) node{F};
			\draw (6.5,4) node{F};
			\draw (1.5,1) node{F};
			\draw (6.5,1) node{F};
			\draw (1.5,3) node{R};
			\draw (6.5,3) node{R};
			\draw (1.5,0) node{R};
			\draw (6.5,0) node{R};
			\draw (0.6,4) node{$0,0,1$};
			\draw (0.6,3) node{$2,0,0$};
			\draw (0.6,1) node{$0,1,0$};
			\draw (0.6,0) node{$0,3,1$};
			\draw (7.4,4) node{$1,0,1$};
			\draw (7.4,3) node{$3,0,0$};
			\draw (7.4,1) node{$0,0,0$};
			\draw (7.4,0) node{$0,2,1$};
			\draw (2.7,3.7) node{Beer};
			\draw (5,3.7) node{Quiche};
			\draw (3,0.3) node{Beer};
			\draw (5,0.3) node{Quiche};
			\draw (3,4.5) -- (4,2);
			\filldraw (3,4.5) circle (1pt);
			\draw (3,4.5) -- (4,5.5);
			\filldraw (4,5.5) circle (1pt);
			\draw (4,5.5) -- (6,5.5);
			\filldraw (6,5.5) circle (1pt);
			\draw (4,5.5) -- (2,5.5);
			\filldraw (2,5.5) circle (1pt);
			\draw[dotted] (2,3.5) -- (2,5.5);
			\draw[dotted] (6,3.5) -- (6,5.5);
			\draw (2.6,4.5) node{Weak};
			\draw (3,5.7) node{$\frac{1}{2}$};
			\draw (5,5.7) node{$\frac{1}{2}$};
			\draw (2,5.5) -- (1,6);
			\draw (2,5.5) -- (1,5);
			\draw (6,5.5) -- (7,6);
			\draw (6,5.5) -- (7,5);
			\draw (3.2,5.1) node{$s^0_w$};
			\draw (6.5,6) node{F};
			\draw (1.5,6) node{F};
			\draw (6.5,5) node{R};
			\draw (1.5,5) node{R};
			\draw (0.4,6) node{$-\e,0.9,0.1$};
			\draw (0.2,5) node{$.2-\e,2.7,0.9$};
			\draw (7.8,6) node{$.1-\e,0,0.1$};
			\draw (7.8,5) node{$.3-\e,1.8,0.9$};
			
		\end{tikzpicture}
		\caption{}
		\label{fig BQ again}
	\end{figure}

\section{Invariance}

An invariance axiom says that the solutions of a game $G$ are the projections of the solutions of a game $\tilde G$ that ``embeds'' it. In KM, embedding just meant the addition of duplicate strategies. In GW we allowed for both the addition of duplicate strategies and other players  whose choices were payoff irrelevant to the original players.  Here, we are going to take an even more expansive notion of an embedding.  As with the addition of dominated strategies, there are going to be two versions of an embedding that nest within one another, yielding two axioms. 

We first recall the definition of an embedding in GW and then give the modified version.  

\begin{definition}\label{def embedding}
	A game $\bar G: \bar \S \equiv \bar \S_\N \times \bar \S_O \to \Re^{\N \cup O}$, where $\bar \S_{\N} = \prod_{n \in \N} \bar \S_n$, and multilinear maps $f_n: \bar \S_n \times \bar \S_O \to \S_n$, $n \in \N$, embed a game $G: \S \to \Re^{\N}$ if:
	\begin{enumerate}
		\item for each $n$ and $\bar \s_O \in \bar \S_O$, the map $f_n(\cdot, \bar \s_O): \bar \S_n \to \S_n$ is surjective;
		\item $\bar G = G \circ f$, where $f = \prod_n f_n$. 
	\end{enumerate}
\end{definition}

As a matter of terminology, we call the players in $\N$ the ``insiders,'' and those in $O$, the ``outsiders.''  As a matter of notation, for any profile $\bar \s_{\N}$ of the insiders, there is an induced game among the outsiders, denoted $\bar G_O^{\bar \s_{\N}}$.   While the payoffs of the outsiders could depend on the strategies of the insiders---and thus $\bar G_O^{\bar \s_\N}$ is a nontrivial function of $\bar \s_\N$---the properties of $f$ in the definition of an embedding ensure that the insiders' strategic interaction is not affected by the outsiders. To see this, it is best to consider an equivalent definition of the function $f$. Proposition 3.2 of GW shows that condition (1) of Definition \ref{def embedding} holds iff there is a bijection $i_n$ between a subset $\bar S_n^*$ of $\bar S_n$ and $S_n$  such that for each $n$ and $\bar s_n \in \bar S_n^*$, we have $f_n(\bar s_n, \cdot) = i_n(\bar s_n)$.  That is, there is a subset of pure strategies whose payoff implications for the insiders are not distorted by the actions of the outsiders, while the strategies in $\bar S_n \backslash \bar S_n^*$ are mapped to mixed strategies in a way that possibly depends on the choice of the outsiders. The payoff evaluations in $\bar G$ are the same as in $G$ using $f$ to identify equivalent profiles. Thus, the outsiders are strategically irrelevant for the insiders.  It is as if the insiders are ``upstream'' players and the outsiders, ``downstream.''

If $\bar G$ embeds $G$ via the map $f$, then the Invariance Axiom in GW requires that the solutions of $G$ be the $f$-images of the solutions of $\bar G$. We want to strengthen the axiom by allowing for a larger class of games to embed a given game $G$.  The idea we use draws on the procedure we used in dealing with the addition of dominated strategies. Suppose $\bar G$ embeds $G$ via a map $f$ in the sense of Definition \ref{def embedding}.  Now take a game $\tilde G$ obtained from $\bar G$ by adding strategies only for the outsiders. We use $\tilde \S$ to denote the strategy space of $\tilde G$---for each $n$, $\tilde \S_n = \bar \S_n$. Suppose, too, that there is an extension of each $f_n$ to a multilinear map $\bar f_n$ over $\bar \S_n \times \tilde \S_O$.\footnote{Multilinearity in this setting, like in Section 4, is also a consequence of continuity and a requirement that the property hold universally over the domain of games being considered.}   When could we consider the newly added strategies irrelevant? That is, when is it reasonable to insist that  the $\bar f$-images of solutions of $\tilde G$ are the same as the $f$-images  of the solutions of $\bar G$?   Obviously $\bar G$ and $\tilde G$ must satisfy Property P1, i.e.,~they must have the same set of Nash equilibria.  But if this property has to hold for some extension $\tilde G$ for any arbitrary specification of payoffs for $G$, then, as any profile in $\S$ could arise as an equilibrium for the upstream players,  Property P1 must hold in a stronger sense:

\smallskip

\noindent {\bf Property P6:} {\slshape For each $\bar \s_\N \in \bar \S_\N$, the games $\bar G_O^{\bar \s_\N}$ and $\tilde G_O^{\bar \s_\N}$ have the same set of Nash equilibria.}

\smallskip

If we want a counterpart of Axiom $D^*$, then a version of Property P5 is necessary as well:

\smallskip

\noindent {\bf Property P7:} {\slshape There exist a neighborhood $\tilde U$ of $\bar \S$ in $\tilde \S$ and, for each $n$, a  multilinear map $\tilde f_n: \tilde \S_n \times \tilde \S_O \to \S_n$ such that:
	\begin{enumerate}
		\item[(a)] the restriction of $\tilde f_n$ to $\bar \S_n \times \bar \S_O$ is $f_n$; 
		\item[(b)] for each $\tilde \t_n \in \tilde \S_n$, $\tilde \s \in \tilde U$,  $\tilde G_n(\tilde \s) \ge \tilde G_n(\tilde \t_n, \tilde \s_{-n})$ iff $G_n(\tilde f(\tilde \s)) \ge G_n(\tilde f(\tilde \t_n, \tilde \s_{-n}))$. 
	\end{enumerate}			
}

\smallskip

We combine Definition \ref{def embedding} with Properties P6 and P7 to allow for a richer class of embeddings.  Rather than first embedding $G$ via $f$ in a game $\bar G$ in the sense of Definition \ref{def embedding} and then extending $\bar G$ to a game $\tilde G$, we will directly state the properties for embedding $G$ in a game $\tilde G$ with the ``subgame'' $\bar G$ being defined implicitly. 

\begin{definition}\label{def strong embedding}
	A game $\tilde G: \tilde \S_\N \times \tilde \S_O \to \Re^{\N \cup O}$, where $\tilde \S_\N = \prod_{n \in \N} \tilde \S_n$, and multilinear maps $f_n: \tilde  \S_n \times \tilde \S_O \to \S_n$, $n \in \N$, strongly embed a game $G: \S \to \Re^{\N}$ if:
	\begin{enumerate}
		\item there is a face $\bar \S_O$ of $\tilde \S_O$ such that for each $\tilde \s_\N \in \tilde \S_\N$, the set of Nash equilibria of the induced game $\tilde G_O^{\tilde \s_\N}$ among the outsiders $O$ is contained in  $\bar \S_O$;
		\item for each $n$, $\tilde \s_{O} \in \bar \S_O$, $f_n(\cdot, \tilde \s_O)$ is surjective; 
		\item there exists a neighborhood $\tilde U$ of $\tilde \S_{\N} \times \bar \S_O$ in $\tilde \S_\N \times \tilde \S_O$ such that for each $\tilde \t_n \in \tilde \S_n$, $\tilde \s \in \tilde U$,  $\tilde G_n(\tilde \s) \ge \tilde G_n(\tilde \t_n, \tilde \s_{-n})$ iff $G_n(f(\tilde \s)) \ge G_n(f(\tilde \t_n, \tilde \s_{-n}))$. 
	\end{enumerate}
\end{definition}

Condition (1) of the definition says that for equilibrium analysis, we can restrict ourselves to the face $\bar \S_O$ of the strategy space of the outsiders.  What is special about this face?  On this face, condition (2) of the definition says that we have the invariance notion as defined by GW---see Definition \ref{def embedding}. As we observed following Definition \ref{def embedding}, one implication of condition (2) is that there is a subset of pure strategies of each insider that is in one-to-one correspondence with those in $S_n$ and the other pure strategies in $\tilde S$ are duplicates.  Thus, if $\bar \S_O = \tilde \S_O$, we have exactly Definition \ref{def embedding}. What condition (3) adds is a richer set of perturbations, in exactly the way condition (2) of Definition \ref{def saids} was used to generate beliefs.  As there, these conditions concern themselves purely with the individual decision problems of the insiders and in no way set up a fixed-point problem.

\begin{example}\label{example invariance}
	We now give an example of a strong embedding. Let $G$ be the Beer-Quiche game of Example \ref{example bq game}. Suppose now that $G_O$ is a completely independent game played by two players $o_1$ and $o_2$ whose payoffs are as follows:
	\[
	\left(
	\begin{array}{ll}
		2,2 & 2,2\\
		3,0 & 0,1  \\
		0,2 & 3,0   
	\end{array}
	\right)
	\]	
	The game $G_O$ has a unique equilibrium outcome $(2,2)$.  Suppose that when player $o_1$ plays his second strategy, Middle, in the game $G$ Strong loses agency, in that regardless of what he chooses, Beer is what the Receiver observes. Similarly, if $o_1$ chooses Bottom, then Weak's choice is overridden and Beer is chosen. It seems a reasonable requirement of a solution concept $\varphi$ that the solutions of the game $\bar G$ involving both games $G$ and $G_O$ are just the products of the solutions of the two games. $\Box$
\end{example}

With this definition of an embedding, we have the following invariance axiom.\footnote{One could create a {\it portmanteau} axiom that combines Axioms $I^*$ and $D^*$ to allow for a simultaneous addition of players and strategies.  But that seems too heavy-handed for our purposes: even a stronger version of the axiom about adding strategies that dispenses with the requirement of dominance is already too much. Such a strengthening might be needed to extend the axiomatization to games with nongeneric payoffs.}

\begin{axiom}[Axiom I$^*$]\label{axm Istar}
	The $\varphi$-solutions of a game $G$ are the $f$-images of the $\varphi$-solutions of any game $\tilde G$ that strongly embeds it via a map $f$, i.e., $\varphi(G) = \{\, f(\tilde \S) \mid \, \tilde \S \in \varphi(\tilde G)\, \}$.
\end{axiom}

As we did in the previous section, if we ask for condition (3) of Definition \ref{def strong embedding} to hold only for profiles in $\bar \S$, then we get a larger class of embeddings.

\begin{definition}\label{def weak embedding}
	A game $\tilde G: \tilde \S_\N \times \tilde \S_O \to \Re^{\N \cup O}$, where $\tilde \S_\N = \prod_{n \in \N} \tilde \S_n$, and  multilinear maps $f_n: \tilde  \S_n \times \tilde \S_O \to \S_n$, $n \in \N$, weakly embed a game $G: \S \to \Re^{\N}$ if:
	\begin{enumerate}
		\item there is a face $\bar \S_O$ of $\tilde \S_O$ such that for each $\tilde \s_\N \in \tilde \S_\N$, the set of Nash equilibria of the induced game $\tilde G_O^{\tilde \s_\N}$ among the outsiders $O$ is contained in  $\bar \S_O$;
		\item for each $n$, $\tilde \s_{O} \in \bar \S_O$, $f_n(\cdot, \tilde \s_O)$ is surjective; 
		\item  for each $\tilde \t_n \in \tilde \S_n$, $\tilde \s_\N \in \tilde \S_\N$, $\tilde \s_O \in \bar \S_O$,  $\tilde G_n(\tilde \s) \ge \tilde G_n(\tilde \t_n, \tilde \s_{-n})$ iff $G_n(f(\tilde \s)) \ge G_n(f(\tilde \t_n, \tilde \s_{-n}))$.
	\end{enumerate}
\end{definition}

Weak embeddings allow us to formulate a stronger axiom of invariance. 

\begin{axiom}[Axiom I] \label{axm I plain}
	The $\varphi$-solutions of a game $G$ are the $f$-images of the $\varphi$-solutions of any game $\tilde G$ that weakly embeds it via a map $f$.
\end{axiom}	

\begin{example}\label{example Axiom I}
As with Axiom D, a big problem with Axiom I is that it does not allow us to compare off-equilibrium-path beliefs in the games $\tilde G$ and $G$. To see this point, let us return to Example \ref{example invariance}. As there, $o_1$'s choice of $M$ (resp. $B$) overrides Strong's (resp. Weak's)  choice and implements the choice of Beer.  If we put arbitrary payoffs for the Receiver after these forced choices, then there is no way to relate the Receiver's decision problem after seeing Beer to his problem in the original Beer-Quiche game. We could map beliefs in the obvious way, but his preferences are not preserved. $\Box$
\end{example}

\section{The Main Results}

We are now ready to state the main theorems of the paper.  We start with stability.

\begin{theorem}\label{thm main 1}
	Stability satisfies Axioms B, $D^\ast$, and $I^\ast$.  For a generic game $\G(u)$, every solution assigned by any other solution concept that also satisfies these axioms must contain a stable set. 
\end{theorem}

For generic $u$, $\G(u)$ has finitely many Nash equilibrium outcomes. Since solutions are connected sets, to each solution there corresponds a unique outcome. Thus, in this case, we can talk of the outcomes selected by a solution concept. The outcomes associated with the solutions selected by stability for a generic game are called stable outcomes.  We now have the following corollary.

\begin{corollary}\label{corollary 1}
	Any solution concept that satisfies Axioms B, $D^\ast$, and $I^\ast$ must, for each generic game $\G(u)$, select from among its stable outcomes. 
\end{corollary}

\begin{example}\label{example main}
	We illustrate the main idea in proving Theorem \ref{thm main 1} using an example.  Theorem \ref{thm main 2} is proved analogously, but using payoff perturbations instead of strategy perturbations.  
	
	Consider the Beer-Quiche game of Example \ref{example bq game}. The Quiche outcome is not stable as it relies on the Receiver assigning more probability to Weak than to Strong if he sees Beer, yet Beer is never a best reply to the equilibrium outcome.  Backward induction does not help by itself since all beliefs are consistent.  As we saw in Example \ref{example bq Axiom D}, we could eliminate this outcome with just Axioms B and D$^*$.  Here we will use Axiom I$^*$ as well and give a different construction.
	
	Take the game $G_O$ developed in Example \ref{example invariance}.  It is useful to represent $G_O$ as an outside option game where the outside option gives $(2,2)$ and otherwise there is a $2 \times 2$ subgame.  In the unique belief supporting the outcome $(2,2)$ as a sequential equilibrium,  player $o_2$ assigns  probability $2/3$ to $M$ and probability $1/3$ to $B$.
	
	As in Example \ref{example invariance}, embed $G$ in the game $\tilde G$ using $G_O$: the choice of $M$ (resp.~$B$) triggers Beer for Strong (resp.~Weak).  Consider the game $\bar G$ obtained from $\tilde G$ by adding a dominated strategy for Weak as in Example \ref{example bq Axiom D}, but now the dominated strategy  immediately overrides $o_1$'s choice (but not Strong's choice) and implements Bottom for $o_1$ (and, from the Receiver's perspective, the strategy Beer for Weak as well). 
	
	We claim that there is no sequential equilibrium of $\bar G$ that induces Quiche.  Indeed, observe first that in any such equilibrium, $o_2$'s beliefs must assign probability $2/3$ to $M$. There are two sources of the belief: the direct choices of $o_1$ and the choice induced by Weak playing his dominated strategy.  Either way, this implies that the Receiver of $G$ assigns a strictly greater probability to Strong than to Weak, upsetting the equilibrium.
	
	We need a combination of both the dominated action for Weak and the effect of the unused actions of $o_1$ for this proof to work. Even if we did not have one of these  two features, the Receiver could assign more weight to Weak because there is nothing disciplining the relative trembles of Strong and Weak.  $\Box$
\end{example}

\begin{example}\label{example nongeneric}
	In a perfect information game with generic payoffs, there is a unique backward induction outcome, which is also the stable outcome.  However, with nongeneric payoffs, a stable set could generate other outcomes, as the following example taken from Mertens \cite{M1995} shows.
	\begin{figure}[htbp]
		\centering
				\centering
				\begin{tikzpicture}[scale=1.2, line width=0.6pt]
					
					\coordinate (P1a) at (0, 0);
					\coordinate (P2a) at (2.5, 0);
					\coordinate (P1b) at (5, 0);
					\coordinate (P2b) at (7.5, 0);
					
					\node[below left=2pt] at (P1a) {$1$};
					\node[below left=2pt] at (P2a) {$2$};
					\node[below left=2pt] at (P1b) {$1$};
					\node[below left=2pt] at (P2b) {$2$};
					
					\draw (P1a) -- (P2a) -- (P1b) -- (P2b) -- (9.5, 0) node[right] {$-1, 1$};
					
					\draw (P1a) -- (2, 1.5) node[right] {$0, 0$};
					\draw (P2a) -- (4.5, 1.5) node[right] {$0, 0$};
					\draw (P1b) -- (7, 1.5) node[right] {$-2, -1$};
					\draw (P2b) -- (9.5, 1.5) node[right] {$-3, 0$};
					
					\draw (P2a) -- (4, -1.2) node[right] {$1, -1$};
					
				\end{tikzpicture}
		\caption{}
		\label{fig Mertens}
	\end{figure}
	
This game has a unique backward induction solution, where player 1 ends the game immediately to give the payoffs $(0,0)$. However, there is an entire component of admissible equilibria, which is the unique stable set of the game, yielding a continuum of outcomes. $\Box$

\end{example}

Essentiality satisfies Axioms B, $D^\ast$, and $I^\ast$ and thus chooses solutions that contain stable sets. However,    
there are stable sets contained in inessential sets, and even those stable sets that reside in essential sets typically tend to be strict subsets.  Axiomatically,  essentiality is  characterized by using Axioms D and I in place of $D^*$ and $I^*$, respectively, as our second theorem shows. 

Since the game in normal form has a lot of redundant strategies and a solution may make  selections from equivalence classes of strategies, we will phrase the theorem using  the reduced normal form that eliminates duplicate pure strategies. Alternatively, it could be done in the space of enabling strategies. 

\begin{theorem}\label{thm main 2}
 Essentiality satisfies Axioms B, D, and I.  If $\varphi$ is a solution concept that satisfies these axioms, then for each generic game $\G(u)$, each solution of the reduced normal form of $\G(u)$ is an essential set.   
\end{theorem}

One obtains the same theorem on replacing connectedness with minimality in the requirement for solutions.\footnote{A relationship between connectedness and minimality exists in fixed-point theory: sets that are minimal with respect to robustness of perturbations of the fixed-point map are connected components of fixed points with nonzero index.}   The proof of Theorem \ref{thm main 2} uses the same line of reasoning as that of Theorem \ref{thm main 1}. The main difference between them arises from the different characterizations of the two concepts---compare Propositions \ref{prop Q} and E.3 of the online appendix.

For a generic extensive-form game, say that an outcome is essential if it is the outcome induced by an essential set.  As with stability, we have the following corollary.

\begin{corollary}\label{corollary 2}
	Any solution concept that satisfies Axioms B, D, and I must, for each generic game $\G(u)$, select from among its essential outcomes. 
\end{corollary}

\subsection{Comparison with the Two-Player Case}
Theorem \ref{thm main 1} differs from the corresponding theorem in GW for two-player games both in terms of the assumptions and the conclusion.  GW also impose Axiom B, but they invoke a weaker version of Axiom $I^*$---the definition of an embedding is as in Definition \ref{def embedding}---and they impose admissibility instead of Axiom $D^*$.  The result in GW is that the solutions themselves have to be stable, rather than just contain stable sets. To understand why we get a weaker statement, consider what happens if in the GW setup we drop admissibility, keeping just their invariance axiom and Axiom B.  Then, we get exactly the conclusion of Theorem \ref{thm main 1}: every solution contains a stable set.  Admissibility's role is to strengthen this conclusion.  And it does so because of two key properties of two-player games: (1)  stable sets are connected components of perfect equilibria; (2) an equilibrium is admissible iff it is perfect.  

In the $N$-player case, even if we add admissibility, we cannot improve on the characterization. There appear to be two ways in which we could proceed. One is to impose a minimality condition, i.e., ask that each solution be as small as possible. Minimality is, in fact, one way to force connectedness of solutions.  We do not find this option palatable, since it serves only the purpose of  keeping solutions as small as possible and is, thus, devoid of any decision-theoretic meaning. The other is to view a (perfect) equilibrium as a lexicographic equilibrium \cite{BBD1991,GK2002} and require connectedness of such solutions.\footnote{This fix brings with it mathematical issues: do we allow for only LPSs with finitely many levels?  infinite levels? what is the appropriate topology for such sets? Conditional probability systems, CPSs, are better behaved than LPSs, as they form a ball \cite{M1989a}, but while they are useful for this paper, they seem inadequate for the general case.}   We believe that this latter approach will ultimately prove fruitful in extending our work to all finite games.

Setting aside the axiom of admissibility, there remains the question of whether we need to strengthen the invariance axiom to I$^\ast$ and whether we need the additional axiom concerning dominated strategies.  The fact is that if we simply use invariance and Axiom B, there is no way to rule out the Quiche outcome in the Beer-Quiche game when we treat it as a three-player game: as we noted in our earlier discussion of this example, there is nothing in the set up that disciplines the Receiver's beliefs about Strong and Weak after observing Quiche.  Even assuming that we do need to strengthen the invariance axiom and/or add another axiom , do we really need these specific formulations?  We are unable to answer this question; an answer will become clear when we do finally have  an axiomatic theory on the domain of all finite games.  The narrower question of whether one of the two axioms, Axioms $D^\ast$ and $I^\ast$, in combination with Axiom B suffices is taken up later in this section.

\subsection{Related Refinements}

Two recent refinements focus on outcomes of solutions in extensive-form  games, as we do. One is the concept of solid outcomes, proposed by Ritzberger et al. \cite{RWW2025}.  The definition of solid outcomes involves two steps.  First, a {\it game block} is a face of the strategy space such that: (1) every Nash equilibrium component is either contained in it or disjoint from it; (2) the sum of the indices of the components contained in it is $+1$. Second, a {\it solid outcome  set} is defined as the set of outcomes of an essential component that resides in a  game block.   When we have a generic extensive-form game, this solution concept gives us finitely many solid outcomes.

The components that define solid outcome sets are essential and, thus, this concept is a refinement of essentiality. Proposition 2 of  Ritzberger et al.~\cite{RWW2025} shows that solid outcome sets contain the outcome of a quasi-perfect equilibrium and hence that the concept satisfies Axiom B.   Our preliminary analysis of this concept shows that it also satisfies Axioms D and I. 

The second solution concept is that of sequentially stable outcomes, defined by Dilm\'e \cite{D2024}. An outcome is sequentially stable if for all sufficiently small trembles and payoff perturbations there is an $\e$-perfect equilibrium that induces a nearby outcome.   It satisfies Axiom B when we invoke sequentiality rather than quasi-perfection.  It is not clear how one verifies Axiom D$^\ast$ since sequential stability is an extensive-form solution concept and Axiom $D^\ast$ deals with the normal form. As for invariance, as in the example that follows, there exist extensive-form games that are equivalent up to duplicate strategies but where the sequentially stable outcomes are not the same.

In terms of the outcome predictions of our axioms and sequential stability, the relationship is clear. While all stable outcomes are sequentially stable, the converse is not true. In fact, there are sequentially stable outcomes that are not even KM-stable. To see this point, consider the outside option game in  Figure \ref{fig sequentially stable}, which is from Pahl and Pimienta \cite{PP2024}. The outside option outcome is sequentially stable; however, it fails forward induction and therefore is not KM-stable.\footnote{If we collapse the two binary choices into a three-way choice, then the outside option is not sequentially stable, which shows that sequential stability does not satisfy invariance.}      

\begin{figure}
	\centering
	
	\begin{tikzpicture}[scale=1.5]
		\centering
		
		\tikzstyle{solid node}=[circle, draw, inner sep=1.2, fill=black];
		
		\node (root) [solid node, label=above:1] at (0,0) {};
		
		\node (out) at (-1.5,-1) {2,2};
		\draw (root) -- (out) node [midway, above left] {Out};
		
		\node (in) [solid node, label=above:1] at (1.5,-1) {};
		\draw (root) -- (in) node [midway, above right] {In};
		
		\node (P2-L) [solid node, label=above:2] at (0,-2.5) {};
		\node (P2-R) [solid node, label=above:2] at (3,-2.5) {};
		\draw (in) -- (P2-L) node [midway, above left] {L};
		\draw (in) -- (P2-R) node [midway, above right] {R};
		
		\draw[dashed, thick] (P2-L) -- (P2-R);
		
		\node (L-l) at (-0.75,-4) {3,1};
		\node (L-r) at (0.75,-4) {0,0};
		\draw (P2-L) -- (L-l) node [midway, above left] {l};
		\draw (P2-L) -- (L-r) node [midway, above right] {r};
		
		\node (R-l) at (2.25,-4) {0,0};
		\node (R-r) at (3.75,-4) {1,3};
		\draw (P2-R) -- (R-l) node [midway, above left] {l};
		\draw (P2-R) -- (R-r) node [midway, above right] {r};
		
	\end{tikzpicture}
	\caption{}
	\label{fig sequentially stable}
\end{figure}

\subsection{The Logical Independence of the Axioms}
In a classical axiomatic treatment, one obtains a characterization of an object via axioms that are logically independent: dispensing with just one of them breaks the uniqueness property.  Given the nature of the theorems here, we could check for such independence of the axioms in a couple of ways, one stronger than the other.  First, we could ask if a solution concept that satisfies two of the three axioms must necessarily satisfy the third. Second, we could ask if a solution concept that satisfies two of the axioms must already have as solutions sets that contain stable sets.  A negative answer to the second, which we now provide, obviously implies a negative answer to the first. For simplicity, we focus just on the axioms for stability.

Perfection as a solution concept satisfies Axioms $D^\ast$ and $I^\ast$: a perfect equilibrium remains so if we embed a game in a larger game in the sense of these two axioms.  However, perfect equilibria need not induce quasi-perfect equilibria in the extensive form. Consider the centipede game in Figure \ref{fig centipede game}. The game has a unique component of equilibria. It also has a  unique backward induction equilibrium, in which player 2 plays $D_2$ at his first node. But the other equilibria in the component are perfect as well.  Thus Axiom B is independent of the other two axioms.

\begin{figure}
	\centering
		
		\begin{tikzpicture}[
			dot/.style={circle, fill=black, inner sep=0pt, minimum size=4pt},
			action/.style={font=\small},
			payoff/.style={font=\small},
			line/.style={draw=black, semithick}
			]
			
			\node[dot, label=above:{$1$}] (N1) at (0,0) {};
			\node[dot, label=above:{$2$}] (N2) at (2.5,0) {};
			\node[dot, label=above:{$1$}] (N3) at (5,0) {};
			\node[dot, label=above:{$2$}] (N4) at (7.5,0) {};
			
			\node[payoff] (T1) at (0,-1.5) {$(1, 0)$};
			\node[payoff] (T2) at (2.5,-1.5) {$(0, 2)$};
			\node[payoff] (T3) at (5,-1.5) {$(3, 0)$};
			\node[payoff] (T4) at (7.5,-1.5) {$(0, 4)$};
			\node[payoff, right] (T5) at (9.7,0) {$(0, 0)$};
			
			\draw[line] (N1) -- (T1) node[midway, left, action] {$D_1$};
			\draw[line] (N2) -- (T2) node[midway, left, action] {$D_2$};
			\draw[line] (N3) -- (T3) node[midway, left, action] {$D_3$};
			\draw[line] (N4) -- (T4) node[midway, left, action] {$D_4$};
			
			\draw[line] (N1) -- (N2) node[midway, above, action] {$A_1$};
			\draw[line] (N2) -- (N3) node[midway, above, action] {$A_2$};
			\draw[line] (N3) -- (N4) node[midway, above, action] {$A_3$};
			\draw[line] (N4) -- (T5) node[midway, above, action] {$A_4$};
			
		\end{tikzpicture}
		
	\caption{}
\label{fig centipede game}
\end{figure}

To show the indispensability of Axiom $I^\ast$,   take the variation of the game in Figure \ref{fig sequentially stable} where player 1 decides simultaneously among the three choices, Out, L, and R, rather than via a sequence of binary choices.  This game has the following normal form $G$:
\[
\left(
\begin{array}{ll}
	2,2 & 2,2\\
	3,1 & 0,0  \\
	0,0 & 1,3   
\end{array}
\right)
\]	
Consider a solution concept that picks stable sets in all games except in the game $G$ and all games obtained from it by the addition of strongly irrelevant dominated strategies;  in those  games, it picks both components of equilibria, i.e., the component yielding the outside option outcome, $(2,2)$, and the pure-strategy equilibrium $(L,l)$. This solution concept is easily seen to satisfy Axiom $D^\ast$.  It also satisfies Axiom B.  Indeed, it does so in the game $G$ itself.  If $\bar G$ is obtained from $G$ by the addition of strongly irrelevant dominated strategies, then, necessarily, in any extensive-form game with $\bar G$ as its normal form, and not just its reduced normal form, the choice  between  such added strategies and the original strategies must be made at the same information sets as in the original game, implying that the extensive-form game is once again trivial, without subgames, making the outside option quasi-perfect.  Observe now that this solution concept  fails Axiom $I^*$: in the normal form of the game in Figure \ref{fig sequentially stable}, which is equivalent to $G$, the outside option is not chosen as a solution, as it is not stable.

Finally, to illustrate the role of Axiom $D^\ast$, view the Beer-Quiche game as a three-player game, and mimic the above procedure to violate only   Axiom D$^\ast$. We choose stable sets in all games except in the Beer-Quiche game and all games that strongly embed it. In those games we choose as solutions sets that project to both components of equilibria in the Beer-Quiche game.    

As can be seen from the examples, we have a lot of latitude in specifying solutions to games. Given that the two invariance axioms deal with different embeddings, logical independence is to be expected.

\section{The Structure of Consistent Beliefs}

McLennan \cite{M1989a, M1989b} shows that the space of conditional probability systems is a ball.  His result is for behavioral strategies and it uses the entire space of conditional probability systems over profiles of pure strategies in the agent-normal form.  From a decision-theoretic viewpoint,  a conditional probability system has a lot of redundant information---for instance, the relative likelihoods of two profiles of pure strategies that are not relevant to the decision problem of any player at any information set.  McLennan remarks in \cite{M1989b} that  if we were to work with just the space of consistent beliefs, as opposed to conditional probability systems, the result that the space is a ball would hold,  but that a formal analysis would require a lot of work. In his setup, it would actually be quite easy  to show that the space of consistent beliefs is a pseudomanifold with boundary.  For our characterization of stability, we do need just the relevant relative likelihoods and, in terms of the property of the space, we do not need the full force of the feature that the space is a ball, just that it is a pseudomanifold with boundary, which is exactly what we establish for enabling strategies.

For each $n$ and $h_n$, let $L_n(h_n)$ be the set of  last actions  following $h_n$ and let $Z(h_n)$ be the set of terminal nodes following $h_n$. Let $P_n(h_n)$ be the projection of $P_n$ to ${[0, 1]}^{L_n(h_n)}$. Let $\S_n(h_n)$ be the set of mixed strategies of $n$ that choose all the actions leading to $h_n$ with probability one. Recall that $\pi_n: \S_n \to P_n$ is the map sending mixed strategies to enabling strategies. The projection $P_n^1(h_n)$ to $[0, 1]^{L_n(h_n)}$ of $\pi_n(\S_n(h_n))$ is the  maximal face of $P_n(h_n)$ that does not contain zero. (If $\S_n \neq \S_n(h_n)$, then $P_n^1(h_n)$ is a proper face of $P_n(h_n)$.) 

Suppose $p$ is a profile such that $p_{-n}(z) > 0$ for some $z \in Z(h_n)$. Then the optimal continuation from $h_n$ for player $n$ is obtained by maximizing $\sum_{z \in Z(h_n)}p_{-n}(z)p_n^1(z)u_n(z)$ by choosing $p_n^1 \in P_n^1(h_n)$.  What matters for optimal behavior is the distribution  over $Z(h_n)$ that is induced by $p_{-n}$.  By going to limits of completely mixed strategies we get consistent beliefs encoded in these distributions. We study the structure of these beliefs below.

Let $\Lmath$ be the set of all ordered pairs $(a_m, a_n)$ with $a_m \in L_m$,  $a_n \in L_n$ for some $m, n$ (possibly with $m = n$).  For each $n$, let $\Z_{-n}$ be the set of all pairs $(z_0, z_1)$ such that both $z_0$ and $z_1$ belong to $Z(h_n)$ for some $h_n \in H_n$, and write $\Z$ for $\vee_n \Z_{-n}$ Each $p$ in the relative interior $P^0$ of $P$ induces for each $(a_m, a_n) \in \Lmath$ a conditional probability
\[
\l(a_m, a_n) \equiv \frac{p_m(a_m)}{p_m(a_m) + p_n(a_n)},
\] 
and for each $n$ and $(z_0, z_1) \in \Z_{-n}$
\[
\l_{-n}(z_0, z_1) \equiv \frac{p_{-n}(z_0)}{p_{-n}(z_0) + p_{-n}(z_1)}.
\]

Let $\bbP^0$ be the set of all pairs $(p, \l) \in P^0 \times {[0, 1]}^{\Lmath \cup \Z}$ where $\l$ is the vector of conditional probabilities derived from $p$. Let $\bbP$ be the closure of $\bbP^0$.  Denote by $\partial \bbP$ the set $\bbP \backslash \bbP^0$. Let $\proj:(\bbP, \partial \bbP) \to (P,\partial P)$ be the natural projection.   Since $\bbP^0$ is homeomorphic to $P^0$ under $\proj$, and $(\bbP, \partial \bbP)$ is a semialgebraic pair,  $(\bbP, \partial \bbP)$ is an orientable $d$-pseudomanifold with boundary, where $d$ is the dimension of $P$.  But that does not imply that $\partial \bbP$ is also a pseudomanifold, a weaker version of which property we need and which we now prove. 

For each triplet $R = (R^0, R^1, R^2)$ that partitions $\Lmath \cup \Z$, let $\bbP(R)$ be the set of $(p, \l) \in \partial \bbP$ such that $\l(\cdot)$ belongs to $(0, 1)$ on $R^0$, is one on $R^1$ and zero on $R^2$.  Obviously, the nonempty sets $\bbP(R)$ partition $\partial \bbP$.  Moreover, for each  subset $\bbP(R)$ of $\partial \bbP$ there is a face of $\partial P$ whose interior contains $\proj(\bbP(R))$.  The following proposition describes the structure of these sets.

\begin{proposition}\label{prop belief manifold}
 	 Each $\bbP(R)$ is a boundaryless manifold and belongs to the closure of one that is a $(d-1)$-dimensional manifold; moreover, if $\bbP(R)$ is $(d-2)$-dimensional, it belongs to the closure of exactly two such manifolds. 
\end{proposition}

\section{A Characterization of Stability}

Given $(q, \l) \in \bbP$, we have, for each $n$ and information set $h_n$, a well-defined belief over the continuation strategy of $n$'s opponents.  Indeed, since we have all the conditional probabilities $\l(z_0, z_1)$ as $(z_0, z_1)$ range over pairs of terminal nodes in $Z(h_n)$, we can compute a ``probability distribution'' $q_{-n}(z; \l, h_n)$ over all terminal nodes  $z$ in $Z(h_n)$. $p_n^1 \in P_n^1(h_n)$ is optimal against $(q, \l)$ at $h_n$ if it maximizes $\sum_{z \in Z(h_n)}\tilde p_n^1(z)q_{-n}(z; \l, h_n)u_n(z)$ subject to  $\tilde p_n^1 \in P_n^1(h_n)$.   We say that an action $a_n \in L_n$ is optimal against $(q, \l)$ in $\G(u)$ if for each $h_n$ that precedes $a_n$ (including the information set where $a_n$ is available) there is an optimal continuation strategy $p_n^1 \in P_n^1(h_n)$ for which $p_n^1(a_n) = 1$.  When $q_{-n}$ is a completely mixed strategy profile, the requirement that $a_n$ is optimal is simply a restatement of the requirement that there exists $s_n \in S_n(a_n)$ that is a best reply to $q_{-n}$. We say that $p$ is a best reply to $(q, \l)$ if for each $n$, each  $a_n$ in the support of $p_n$ (i.e.,~$p_n(a_n) > 0$) is optimal against $(q, \l)$.  

For each game $\G(u)$, let $\calQ(u)$ be the closure of the set of $(p, (p^*, \l)) \in  (P \backslash \partial P) \times \bbP$ such that:  (1) $p^*$ is a best reply to $(p^*, \l)$ in the game $\G(u)$; (2) for every  $a_m$ and $a_n$,  we have $\l(a_m, a_n) \ge \frac{p_m(a_m)}{p_m(a_m) + p_n(a_n)}$ if $a_n$ is not a best reply to $(p, \l)$ in $\G(u)$.  Let $\partial \calQ(u)$  be the set of $(p, (p^*, \l)) \in \calQ(u)$ such that $p \in \partial P$. We write $(\calQ, \partial \calQ)$ for $(\calQ(u), \partial \calQ(u))$ when the game $\G(u)$ is implicitly given. 

\begin{proposition}\label{prop Q}
	If $u$ is generic, there exists a finite partition $(\calQ^1, \partial \calQ^1), \ldots, (\calQ^k, \partial \calQ^k)$ of $(\calQ, \partial Q)$ such that:
	\begin{enumerate}
		\item for each $i$, $(\calQ^i, \partial Q^i)$ is an orientable $d$-pseudomanifold;
		\item $\S^*$ is a stable set of $\G(u)$ iff there exists $i$ such that: 
		\begin{enumerate}
			\item $\pi(\S^*) = \{\, q \in P \mid  \exists \, (p, (q, \l)) \in \calQ^i \, \}$;
			\item $\proj: (\calQ^i, \partial \calQ^i) \to (P, \partial P)$ is essential.
		\end{enumerate}	
	\end{enumerate}
\end{proposition}

Since each $\calQ^i$ is an orientable pseudomanifold with boundary, the essentiality condition of Proposition \ref{prop Q} is equivalent to the fixed-point property:  every continuous function $f: \calQ^i \to P$ has a point of coincidence with $\proj$.  Also, if a set is $p$-stable for some prime $p$, then it is also $0$-stable. Thus, for generic extensive-form games, $0$-stability allows for a larger set of solutions.

\section{Proof of Theorem \ref{thm main 1}}

Theorem 6, Section 3, of Mertens \cite{M1989} shows that a stable set  contains a proper equilibrium. A proper equilibrium is equivalent to  a quasi-perfect equilibrium in the extensive-form representation of the game. Therefore, stability satisfies Axiom B.  In the online appendix we show that stability satisfies Axioms $D^\ast$ and $I^\ast$.  Here we prove the second statement of Theorem \ref{thm main 1}, which is that any solution concept satisfying Axioms B, $D^\ast$, and $I^\ast$ must select as solutions for generic games sets that contain stable sets.

Let $\varphi$ be a solution concept that satisfies Axioms $D^\ast$ and $I^\ast$. Suppose that for a generic game $\G(u)$, it assigns a solution $\S^*$ that does not contain a stable set.  We will show through a sequence of steps that $\varphi$ fails Axiom B.  We first provide a brief overview of the proof, which generalizes the ideas presented in Example \ref{example main}. Letting $P^*$ be the set of enabling strategies equivalent to $\S^*$, and $\calQ^*$ be the set of $(p, (p^*, \l)) \in \calQ$ such that $p^* \in P^*$, the fact that $\S^*$ does not contain a stable set means that the projection from $\calQ^*$ to $P$ is inessential. Hence there is a map $f: \calQ^* \to P$ that does not have a point of coincidence with the projection map from $\calQ^*$.  This fact is established in Step 1. 

The map $f$  corresponds to the best-reply function of an outsider, $\bar o_2$, while the coordinate $p$ in $\calQ^*$ is controlled by an outsider $\bar o_1$ who wants to mimic $\bar o_2$.  Given that $\proj$ does not have a point of coincidence with $f$, there would be no Nash equilibrium with coordinates in $\calQ^*$ in the game between these two players.  To realize $f$ as a best-reply function, we have to discretize the problem, which is the task of Step 2.  It decomposes the ambient space of $\calQ^*$ into a union of simplices and provides a multilinear approximation $g$ of $f$.

As in Example \ref{example main}, we introduce a signaling game, now with $N$ senders, one per player of $\G$.  Rather than take the composition of $\G$ with this game, we construct a family of compositions of these games,  parameterized by $p \in P$. This is the task of Step 3.  Step 4 adds weakly dominated strategies for the players in $\G$, again as in Example \ref{example main}. 

The outsiders $\bar o_1$ and $\bar o_2$, as well as a host of other players, are introduced in Step 5 in order to mimic the beliefs that arise in a quasi-perfect equilibrium, as encoded in the space $\calQ$. Steps 6 and 7 set up a finite game involving all these players, while Step 8 shows that the resulting game is obtained from $\G$ via a strong embedding.

Step 9 provides the coda. A quasi-perfect equilibrium induced by a point in $\S^*$ must yield a point in $\calQ^*$, which would trigger an equilibrium for the outsiders of the sort we have ruled out.

\smallskip

\noindent {\bf Step 1---Fixed-Point Preliminaries.}  Let $P^*$ be the set of enabling strategies that are equivalent to $\S^*$, a solution assigned by $\varphi$ to $\G(u)$.  For each $i$, let $\calQ^{i,*}$ be the subset of the manifold $\calQ^i$ (from Proposition \ref{prop Q}) consisting of all points $(p, (p^*, \l))$ such that $p^* \in P^*$. $\calQ^{i,*}$ is a closed subset of $\calQ^i$ since $\S^*$ is closed. Moreover, since $P^*$ does not contain a stable set, if the projection from $\calQ^i$ to $P$ is essential, then $\calQ^{i,*}$ is a strict subset of $Q^{i}$.

For each $\calQ^i$ for which the projection to $(P, \partial P)$ is essential, pick a point $x^i \in \calQ^i \backslash \partial \calQ^i$ with a neighborhood $V(x^i)$ that is disjoint from $\calQ^{i,*}$. As in Appendix E of Govindan and Wilson \cite{GW2012}, we can now construct a map $f: \calQ \to P$ such that $f(p, (q,\l)) = p$ only if $(p, (q,\l))$ belongs to one of the sets $V(x^i)$.  In particular, letting $\calQ^* = \cup_i \calQ^{i,*}$, the restriction of $f$ to $\calQ^*$ does not have a point of coincidence with the restriction of $\proj$.  We can assume that $f$ maps $\calQ$ into $P \backslash \partial P$, while preserving this last property: indeed, for a small $\e' > 0$ and $p^0 \in P \backslash \partial P$, we can use the map $(1-\e')f + \e' p^0$ whose image is contained in the interior of $P$.   We exploit the function $f$ to first construct a strong embedding of $\G(u)$ and then add some dominated strategies such that the resulting game has no quasi-perfect equilibrium that projects to an equilibrium in $P^*$.

Observe that for any $(p, (q, \l)) \in P \times \bbP$,  the vector $\l$ of conditional probabilities already has all the information required to compute $q$, i.e.,~the projection from $P \times \bbP$ to $\calP \equiv  P \times {[0, 1]}^{\Lmath \cup \Z}$ is a homeomorphism onto its image $\hat \calP \subsetneq \calP$. Therefore, we can view $\calQ$ as a subset of $\calP$ and $f$ as being defined from this set to $P$. For any point $(p, \l) \in \hat \calP$, we use $q(\l)$ to denote the enabling strategy defined by $\l$. 

The map $f$ can be extended to the whole of $\calP$ such that, letting $\proj$ still  be  the natural projection onto  $P$: there is no point of coincidence between $f$ and $\proj$ in a neighborhood $U$ of $\calQ^*$; and  $f$ maps $\calP$ into $P \backslash \partial P$. In particular, there exist $U \subset \calP$ and $\d > 0$ such that: (1) $U$ contains the $2\d$-neighborhood of $\calQ^*$; (2) ${\Vert f(p, (q, \l)) - p' \Vert}_\infty > 2\d$ for all $(p, (q, \l)) \in \calP$ and $p' \in \partial P$;  (3) ${\Vert f(p,  \l) -  p \Vert}_\infty > 2\d$ for all $(p,  \l) \in U$. 

For each $\e > 0$, let $\Q^\e$ be the set of $(p, \l) \in P \times \hat \calP$ such that: (1) $q(\l) \in P^*$ and  it is a best reply to $(q(\l), \l)$; (2) if $a_n$ is not an $\e$-best reply to $(q(\l), \l)$, then for each $a_m$, $\l(a_m, a_n) \ge \frac{p_m(a_m)}{p_m(a_m) +  p_n(a_n)}$. There exists $\e > 0$ small enough such that $Q^\e$ is contained in the $\d$-neighborhood of $\calQ^*$.

\smallskip

\noindent {\bf Step 2---Triangulations.} Take a simplicial triangulation $\K(P)$ of $P$ such that the diameter of each simplex is less than $\d$.   $\calP$ is a product space $P \times {[0, 1]}^{\Lmath \cup \Z}$ and therefore we can take a multisimplicial triangulation $\Lmath(\P)$ of $\P$ such that the diameter of each multisimplex is less than $\d$ and such that  there is a multisimplicial approximation $g: \calP \to P$ to $f$ when $\calP$ is the space of $\Lmath(\P)$ and the range is the space of $\K(P)$ (see Appendix B of Govindan and Wilson \cite{GW2005}).  Then, for every $(p, \l) \in \calP$,   $g(p, \l)$ belongs to a simplex that does not intersect $\partial P$; and for no $(p, \l)$ in a multisimplex that intersects $\calQ^\e$ does there exist a simplex of $\K(P)$ that contains both $g(p, \l)$ and $p$.  Let $\T(\calP)$ be a polyhedral subdivision of $\Lmath(\calP)$. Let $\g: \calP \to \Re_+$ be a piecewise-linear convex function that is linear precisely on the polyhedra of $\T(\calP)$---see Appendix B of \cite{GW2005} for details on how to construct such a function. With these preliminaries out of the way, we are ready to construct an embedding of $\G$.

\smallskip

\noindent {\bf Step 3---A Family of Embeddings.}  We construct a family of strong embeddings of $G$. The procedure involves first adding a signaling game $\G^o$ that is completely independent of $\G$ and that has a unique equilibrium component. 

In the signaling game $\G^o$, there are $N$ types for the sender, indexed $o_n$, $n \in \N$, and there is a receiver, called $o_{N+1}$.  The distribution over the types of the sender is uniform.  The sender has two signals, $h$ and $l$. The game ends if signal $h$ is sent, with the sender, independently of his type, and the receiver receiving a payoff of $N$ (the receiver's payoff here is irrelevant).  After observing a signal of $l$, the receiver has $N$ possible responses, indexed $r_1, \ldots, r_N$. If the receiver chooses $r_n$, then: (1)  type $o_n$ receives a payoff of $0$ from $l$ while all other types receive a payoff of $N+1$ from $l$; (2) the receiver receives a payoff of $1$  if $l$ was sent by type $o_n$ and receives $0$ otherwise. We will view $\G^o$ as an $(N+1)$-person game, with player set $O$: each type of the sender is treated as a separate player, i.e.,~we will use the agent-normal form of the game. 

$\G^o$ has a unique equilibrium component, where all types of the sender pool in sending $h$.  Every equilibrium is sequential, but  the beliefs at the receiver's information set following a message of $l$ are unique across all the equilibria, and they assign equal probability to all types. 
 
We  construct a family of games $\tilde G(p)$, indexed by  $p \in P$, that are derived from the product game $\G \times \G^o$  in two steps.  Fix $p \in P$ and consider first the following game $\tilde G^1(p)$.  The senders $o_n$ move first, then the players in $\N$ play their game $G$ in normal form, and finally the receiver moves. Each player $n$ is told whether player $o_n$ has chosen $h$.  If $o_n$ chose $h$, then $n$ has access to all his strategies and chooses a pure strategy $s_n \in S_n$; otherwise, the mixed strategy equivalent to the enabling strategy $p_n$ is automatically implemented. Finally, the receiver moves if one of the senders has chosen $l$.  The game $\tilde G^1(p)$ strongly embeds $G$ for each $p \in P$. By Axiom $I^*$, there is a solution of $\tilde G^1(p)$ that projects to $\S^*$.

\smallskip

\noindent {\bf Step 4---Adding Weakly Dominated Strategies.} Now we modify $\tilde G^1(p)$ to obtain the game $\tilde G(p)$ as follows.   Each player $n$ chooses a strategy $s_n \in S_n$ in two stages. In the first stage, he makes a provisional choice of a strategy $s_n$. In the second stage, following the choice of a strategy $s_n$, he gets an opportunity to revise his choice. If he does not want to revise his strategy, then $s_n$ is implemented.  If he wants to change his choice, he gets an opportunity to choose one of a finite set of strategies labeled $(t_n, h_n)$, where $t_n \in S_n$ and $h_n$ is an information set enabled by $t_n$. Choosing $(t_n, h_n)$ has the following consequences.  (1) It gives $n$ the same payoff as choosing $t_n$ except that it is reduced by $\e$ at nodes following $h_n$ (where $\e$ is the positive number chosen at the end of Step 1); (2) from the perspective of players $m \neq n$, the strategy $(t_n, h_n)$ implements $p_n$; (3) player $o_n$'s choice of $h$ in the signaling game $\G^o$  is nullified and $l$ is chosen instead.

If $p \notin \partial P$, then we claim that the set of Nash equilibria of $\tilde G(p)$ is the same as that of $\G \times \G^o$.  Indeed, if some player $n$ plays some $(t_n, h_n)$ with positive probability in an equilibrium $\tilde \s$ of $\tilde G(p)$, then for each $m$ who does not play a $(t_m, h_m)$ with positive probability, the corresponding sender $o_m$ must play $l$ with positive probability, implying that every information set of $\G$ that is enabled by $(t_n, h_n)$ is reached with positive probability and $(t_n, h_n)$ would be an inferior reply.  Thus, in no equilibrium of $\tilde G(p)$ does a player $n \in \N$ use a dominated strategy $(t_n, h_n)$. By Axiom $D^\ast$,  therefore, the solutions of $G$ are the projections of the solutions of $\tilde G(p)$ for each $p \notin \partial P$.

\smallskip

\noindent{\bf Step 5---Adding More Players.} Over the next three steps, we construct a game $\bar G$ by adding a set $\bar O$ of players to  $\N \cup O$, which is the set of players in the game $\G \times \G^o$. $\bar O$ consists of the following players: three players, denoted $\bar o_0$, $\bar o_1$, and $\bar o_2$;  for each $(a_m, a_n) \in \Lmath$, a player $\bar o_{a_m, a_n}$; and for each $n$ and $(z_1, z_2) \in \Z_{-n}$, a player $\bar o_{n,z_1, z_2}$.

The pure strategy sets of the players in $\bar O$ are defined by making use of the triangulations constructed in Step 2 and are as follows.  Outsider $\bar o_0$ chooses a full-dimensional polyhedron $T$ of $\T(\P)$. $\bar o_{1}$ and $\bar o_2$ choose a vertex of $\K(P)$, which is a triangulation of $P$. The other outsiders choose a vertex of the triangulation of their  copy ${[0,1]}$ (in the multisimplicial triangulation $\Lmath(\P)$) corresponding to the coordinate (which is either $(a_m, a_n)$ or $(n, z_1, z_2))$.  For each outsider $\bar o$, we use $\bar S_{\bar o}$ and $\bar \S_{\bar o}$ to denote his pure and mixed strategy sets, respectively.

\smallskip

\noindent{\bf Step 6---Timing of the Moves.}
For each outsider $\bar o_{i}$, $i = 1, 2$, a pure strategy $\bar s_{\bar o_i}$ corresponds to a point $\phi_{\bar o_i}(\bar s_{\bar o_i}) \equiv p \in P$ with $\phi_{\bar o_i, n} \in P_n$ being the coordinates of $L_n$ Hence there exists a unique affine extension of this function to $\bar \S_{\bar o_i}$, denoted still by $\phi_{\bar o_i}$.  Likewise, for each $i$ of the form $(a_m, a_n)$ or $(n, z_1, z_2)$, there is an affine function $\phi_{i}: \bar \S_{\bar o_i} \to [0, 1]$.  We use $\phi_{-\bar o_2}$ to denote the collection of all these functions other than $\phi_{\bar o_2}$ and view $\phi_{-\bar o_2}$ as a function from $\S_{\bar O}$, the mixed strategy space of the players in $\bar O$,  to $\calP$.

The timing of the game in the extensive form $\bar \G$ of $\bar G$ is as follows. Players in $\bar O$ play simultaneously and ahead of the other players. The players $\N \cup O$ then play their game unaware of the choices of the players in $\bar O$, and their strategy sets are as in the game $\tilde G(p)$ (for any $p$). The choice of outsider $\bar o_{1}$ determines which of the games $\tilde G(p)$ is actually played by the players in $\N \cup O$, i.e.,~if he chooses a pure strategy $\bar s_{\bar o_1}$, then the game played is $\tilde G(\phi_{\bar o_1}(\bar s_{\bar o_1}))$.  

\smallskip

\noindent{\bf Step 7---Payoffs of the Outsiders.} Let $\bar S$ be the set of pure strategy profiles of $\bar G$ and let $\bar \S$ be the set of mixed strategy profiles.  For each $n$,  each profile $\bar s$ induces a point $q_n \equiv \psi_n(\bar s)$ as follows: $q_n = \phi_{\bar o_1, n}(\bar s_{\bar o_1})$ if $\bar s_{o_n} = l$ or $\bar s_n = (t_n, h_n)$ for some $(t_n, h_n)$; otherwise it equals $\pi_n(\bar s_n)$, the enabling strategy equivalent to $\bar s_n$. Likewise, each profile $\bar s$ also induces a strategy $s_{o_n} = \psi_{o_n}(\bar s)$ of player $o_n$: $s_{o_n} = l$ if $\bar s_{o_n} = l$ or $\bar s_n = (t_n, h_n)$ for some $(t_n, h_n)$; otherwise it is $h$. The maps $\psi_n$ and $\psi_{o_n}$ extend multilinearly to the whole of $\bar \S$.

We now describe the payoffs of the outsiders. $\bar o_0$'s payoff function depends on the choices of players in $\bar O \backslash \{\, \bar o_2 \, \}$.  Fix a full-dimensional polyhedron $T$ of $\T(\P)$. Recall from Step 2 that we have a piecewise-linear function $\g: \P \to \Re_{+}$ that is linear precisely on the polyhedra of $\T(\calP)$. The restriction of  $\g$ to $T$ is therefore linear, and it has a unique extension $\g_T$ to the whole of $\P$.  Given a strategy profile $\bar s$ for the outsiders where $\bar s_{\bar o_0}$ is $T$, $\bar o_{0}$'s payoff from choosing $T$ is $\g_T( \phi_{-\bar o_2}(\bar s))$.  
 
Player $\bar o_1$'s payoff depends on his choice and that of player $\bar o_2$. If he chooses a vertex $v$ and $\bar o_{2}$ chooses $w$, then his payoff is one if $v = w$ and zero otherwise. Thus, $\bar o_{1}$ wants to mimic $\bar o_{2}$.  

Player $\bar o_{a_m, a_n}$'s payoffs depend on the relative probabilities of the last actions $a_m$ and $a_n$ for the original players.  Given a profile $\bar s$, letting $\l = \phi_{\bar o_{a_m, a_n}}(\bar s_{\bar o_{a_m, a_n}})$, $q_m = \psi_m(\bar s)$, and $q_n = \psi_n(\bar s)$,   $\bar o_{a_m, a_n}$'s payoff is $2\l q_m(a_m) + 2(1-\l)q_n(a_n) - (\l^2 + {(1-\l)}^2)(q_m(a_m) + q_n(a_n))$. Player $\bar o_{n, z_1, z_2}$'s payoffs are defined exactly like those of $\bar o_{a_m, a_n}$, but using the probabilities $q_{-n}(z_1)$ and $q_{-n}(z_2)$ in place of $q_m$ and $q_n$.

Finally, player $\bar o_{2}$'s payoff depends on all the other outsiders.  Suppose $\bar o_{2}$ chooses a vertex $w$.  If $\bar o_0$ chooses a polyhedron $T$, there exists a multisimplex of $\Lmath(\calP)$ that contains $T$. Let $u_{2}(\cdot, w, T)$ be the function on the vertex set of $L$ that assigns one to a vertex $v$ if $w$ is the image of $v$ under $g$ and zero otherwise.  $u_{2}(\cdot, w, T)$ extends uniquely to a multilinear function on the whole of $\calP$. Given a profile $\bar s$, if $\bar s_{\bar o_0}$ is a polyhedron $T$ of $\bar o_0$,  the payoff from choosing $w$ is $u_{2}(\phi_{-\bar o_{2}}(\bar s),  w, T)$.

\smallskip

\noindent {\bf Step 8---The Relationship between $\bf G$ and $\bf \bar G$.}  The game, call it $\hat G$, obtained from $\bar G$ by deleting the dominated strategies $(t_n, h_n)$ of each player $n$ strongly embeds $\G$ as well as the signaling game $\G^o$. By Axiom $I^\ast$, there is a solution $\hat \S^*$ of $\hat G$ whose projection to $\G$ is $\S^*$ and whose projection to $\G^o$ is its unique equilibrium component. In the game $\bar G$, player $\bar o_2$'s pure strategies that are best replies to a profile are all vertices of simplices that are contained in $P \backslash \partial P$. Hence, in any equilibrium $\bar \s$,  $\phi_{\bar o_2}(\bar \s_{\bar o_2}) \in P \backslash \partial P$. As $\bar o_1$ mimics $\bar o_2$,  $\bar p \equiv \phi_{\bar o_1}(\bar \s_{\bar o_1}) \in P \backslash \partial P$ as well.  As we saw in Step 4, the game $\tilde G(\bar p)$ does not have an equilibrium where an insider plays a dominated strategy.  Therefore,  by Axiom $D^\ast$, there is a  solution $\bar \S^*$ of $\bar G$ that projects to $\hat \S^*$. 

\smallskip

\noindent{\bf Step 9---Failure of Axiom B.}  To complete the proof of Theorem \ref{thm main 1}, we will show that $\bar \S^*$ does not contain a profile that is equivalent to a quasi-perfect equilibrium of $\bar \G$, i.e., $\varphi$ fails Axiom B. Suppose to the contrary that there is a backward induction equilibrium in $\bar \S^*$. Then there is a  quasi-perfect equilibrium $\bar b$ in behavioral strategies of $\bar G$ such that letting $\bar \s$ be the equivalent mixed strategy profile, we have that $\psi_{\N}(\bar \s) \in P^*$.  Let $\bar b^\eta$ be a sequence of $\eta$-quasi-perfect equilibria converging to $\bar b$ and let   $\bar \s^\eta$ be an equivalent sequence of mixed strategies converging to  $\bar \s$. 

 For each $\eta$, let $q^\eta = \psi_{\N}(\bar \s^\eta)$, $p^\eta = \phi_{\bar o_1}(\bar \s_{\bar o_1}^\eta)$, and $\tilde \s_{O}^\eta = \psi_{O}(\bar \s^\eta)$. Also, for each $n$ and $\eta$, we can express $\bar \s_n^\eta$ as a convex combination $(1-\zeta^\eta)\tilde \s_n^\eta + \zeta^\eta \bar \t_n^\eta$, where $\tilde \s_n^\eta$ has as its support the original strategy set $S_n$ and $\bar\t_n^\eta$ has as its support the strategies $(t_n, h_n)$.  For each $n$, then, $q_n^\eta$ is of the form  $(1-\tilde \s_{o_n}^\eta(l))\tilde q_n^\eta + \tilde \s_{o_n}^\eta(l) p_n^\eta$, where $\tilde q_n^\eta = \pi_n(\tilde \s_n^\eta)$. Let $q$ and $p$ be the limits of $q^\eta$ and $p^\eta$, respectively.  Let $\l^\eta$ be the conditional probabilities induced by $q^\eta$ and, if necessary, by going to a subsequence, let $\l$ be their limit. 

If a strategy $s_n$ of player $n$ is not $\e$-optimal against $\bar \s^\eta$, then at the stage where he is considering playing $s_n$ or one of the dominated strategies $(t_n, h_n)$, he would assign zero probability to $s_n$ under $\bar b$. Therefore, for each $n$ and each action $a_n$ that is not $\e$-optimal against $\l$, $\lim_{\eta \to 0} \frac{\tilde q_n^\eta(a_n)}{\tilde \s_{o_n}^\eta(l)} = 0$. For each pair $m, n$, $\tilde \s_{o_m}^\eta(l) = \tilde \s_{o_n}^\eta(l)$ all along the sequence, since otherwise one of the responses, say $r_n$, is not optimal against $\bar \s^\eta$, implying that $l$ is a better response for $o_n$ than $h$, which is impossible. Putting these two facts together,  for each $m, n$ and action $a_n$ that is not optimal, 
\[
\lim_{\eta \to 0} \frac{(1-\tilde \s_{o_m}^\eta(l))\tilde q_m^\eta + \tilde \s_{o_m}^\eta(l) p_m^\eta}{(1-\tilde \s_{o_n}^\eta(l))\tilde q_n^\eta + \tilde \s_{o_n}^\eta(l) p_n^\eta} \ge \lim_{\eta \to 0}  \frac{\tilde \s_{o_m}^\eta(l) p_m^\eta}{(1-\tilde \s_{o_n}^\eta(l))\tilde q_n^\eta + \tilde \s_{o_n}^\eta(l) p_n^\eta} = \frac{p_m(a_m)}{p_n(a_n)},
\]
showing that $(p, (q, \l)) \in \calQ^\e$. 

For each $(a_m, a_n)$, letting $\bar \l(a_m, a_n) = \phi_{\bar o_{a_m, a_n}}(\bar \s_{\bar o_{a_m, a_n}})$, the payoff structure implies that $\bar \l(a_m, a_n)$ is within $\d$ of $\l(a_m, a_n)$. A similar conclusion holds for the players $\bar o_{n,z_1, z_2}$.  Thus $(p, \bar \l)$ belongs to $U$. The support of $\bar o_0$'s strategy $\bar s_{\bar o_0}$ is a subset of the  polyhedra containing $(p, \bar \l)$.   Therefore, the support of $\bar o_{2}$'s  strategy is contained in the set of  vertices that are in the support of $g(p, \bar \l)$ in the triangulation $\K(P)$.  We now have that $\bar o_{1}$ would mimic $\bar o_{2}$, implying that $p$ belongs to the simplex of vertices in the support of  $\bar o_{2}$'s strategy. Consequently, $g(p, \bar \l)$ and $p$ belong to a common simplex, which is impossible.  Therefore,  $\bar \S^*$ does not contain a backward induction equilibrium, i.e., $\varphi$ fails Axiom B.

\section{Some Remarks on the Nongeneric Case}

The essentiality condition in Definition \ref{def stability} for the projection map $\proj: X_\e \to T_\e$ implies a strong fixed-point property: every continuous function $f: X_\e  \to T_\e$ has a point of coincidence with $\proj$, i.e.,~there exists $(\eta, \s) \in X_\e$ such that $\proj(\eta, \s) = f(\eta, \s)$.  When $X_\e$ and $T_\e$ have the same dimension, the essentiality condition is equivalent to the fixed-point property---see the Theorem of Section 4.E of \cite{M1991} and Lemma A.4 of \cite{GW2008}.  For a generic game $\G(u)$, our characterization of stability in Proposition \ref{prop Q} shows that for a transformation of $X_\e$ and $T_\e$, we do indeed have sets of the same dimension. Thus, if a solution does not contain a stable set, we can construct a fixed-point problem without a solution. Axioms B, $D^*$, and $I^*$ turn the fixed-point problem into a game-theoretically meaningful one.

When the given game is nongeneric (even in the space of extensive-form games), $X_\e$ typically has a different dimension from $T_\e$. Stability asks that $X_\e$ contain a subset of the same dimension as $T_\e$ and for which the restriction of the projection is essential---see, again, the Theorem of Section 4.E of \cite{M1991}.  As yet, we do not understand axiomatically the need for this dimensionality requirement. In \cite{GW2008}, we pursued a refinement, called metastability, based directly on the fixed-point property of the projection, which in a sense drops the dimensionality requirement from the definition of stability.  Metastable sets satisfy a lot of the same properties as stable sets, and they coincide with stable sets for generic extensive-form games. But, importantly, metastability satisfies only  weaker forms of the  small-worlds axiom \cite{M1992} and the decomposition property , which are subsumed by Axiom $I^*$---see Theorems 4.4 and 4.5 of \cite{GW2008}.   Perhaps metastability and Axiom $I^*$ imply stability. If that were true, then there would be hope of using the axioms here---and maybe other axioms?---to extend the result to all games.\footnote{We are not wedded to the concept of stability---it might be that something stronger than metastability and weaker than stability is the right concept.}

The fact that our axiomatization of stability does not tell us that solutions have to be stable sets, merely that they should contain one, is another problem for the general case. We could impose a minimality axiom to deal with it, as we could here. But as we argued in the paper, any such result would not be completely satisfactory, since the only rationale for a minimality axiom is that it would whittle down solutions.  Another approach, which we footnoted above, is to view an equilibrium not just as a profile but as a lexicographic probability system (LPS). While this approach is intuitively appealing, it seems beset with technical problems.

With regard to the parallel axiomatization of essentiality, the picture is a little clearer.  While we did make use of the genericity assumption in our characterization of essentiality in the appendix, the definition of essentiality directly involves a fixed-point property---see Theorem 5.2 of \cite{O1953}---and there are no issues relating to the dimensionality of the objects. Alternative characterizations of essentiality would show what additional axioms one might need, and they would pave the way for an extension of the result to all games.

\appendix

In this appendix, we provide proofs of results that were stated without proof in the main text.

\section{Verifying Axioms D and $D^\ast$}
In this section, we verify that essentiality and stability satisfy Axioms D and $D^\ast$, respectively.  

\begin{proposition}\label{prop Axiom D}
Essentiality satisfies Axiom D.
\end{proposition}

\begin{proof}
	Let $\bar G$ be a game obtained from a game $G$ by the addition of weakly irrelevant dominated strategies.  The two games have the same set of components of equilibria.  It suffices to show that the index of each component is the same in each game.  Let $f$ and $\bar f$ be the best-reply correspondences of $G$ and $\bar G$.  Let $\tilde f: \bar \S \to \bar \S$ be the correspondence that assigns to each $\bar \s$ the set of $\s \in \S$ that are best replies against it in $\bar G$. $\tilde f$ is an upper semicontinuous correspondence with nonempty, compact, and convex values. It is linearly homotopic to $\bar f$ and the set of fixed points along the homotopy is constant.  Hence, by the homotopy axiom for index, the fixed-point index of each component of equilibria under the correspondence $\tilde f$ is the same as the index of the component in the game $\bar G$. Since $\tilde f$ maps into $\S$, the index of each component of equilibria under $\tilde f$ can be computed as the index under the restriction of $\tilde f$ to $\S$, which is in fact $f$.   
\end{proof}

We now show that stability satisfies Axiom $D^\ast$. The proof uses two definitions of stability that are equivalent to the original definition proposed by Mertens---see Theorems 4 and 10 of Govindan and Mertens \cite{GM2003}.  First, we need some notation. For a game $G$, let $\BR$ be the graph of the best-reply correspondence, i.e.,~
\[
\BR = \{ \, (\s, \t) \in \S \times \S \mid \t \, \text{is a best reply to} \, \s \, \text{in} \, G \, \},
\]
and let $\D$ be the diagonal of $\S\times \S$, i.e.,~
\[
\D = \{\, (\s, \s) \in \S \times \S \, \}.
\]
For $0 \le \e \le 1$, let
\[
\BR_\e = \{\, (\s, \t) \in \BR \mid \s \ge (1-\e)\t \, \},
\]
and when $\e > 0$,
\[
\partial \BR_\e = \{\, (\s, \t) \in  \BR_\e \mid \s_{n,s_n} =  (1-\e)\t_{n,s_n} \, \text{for some} \, n, s_n \, \}.
\]
For $\e > 0$, we say that a strategy profile $\s$ is an $\e$-perfect equilibrium if it is completely mixed and assigns no more than $\e$ probability to the set of non-best replies, i.e.,~$\s$ belongs to $\proj(\BR_\e) \backslash \partial \S$, where $\proj$ is the projection of $\BR$ onto the first factor.
  
Let $A$ be the affine space generated by $\S$ and let $\partial A = A \backslash (\S \backslash \partial \S)$. For $(X_\e, \partial X_\e) \subseteq (\BR_\e, \partial \BR_\e)$ and $K \ge \e^{-1}$, we define a map $\varphi_K : (X_\e, \partial X_\e) \to (A, \partial A)$  by $\varphi_K(\s, \t) = \t + K(\s - \t)$. If $K = \e^{-1}$, the image of $\varphi_K$ is contained in $\S$, but not necessarily otherwise.  Observe, too, that the maps $\varphi_K$ are (linearly) homotopic to one another.

\begin{proposition}\label{prop stability} 
	Let $\S^*$ be a subset of $\S$. The following statements are equivalent:
	\begin{enumerate}
		\item $\S^*$ is a stable set of $G$;
		\item There exist  $\e_0 > 0$ and, for each $0 < \e \le \e_0$, a  closed semialgebraic set $X_\e \subseteq \BR_\e$  such that:
			\begin{enumerate}
				\item $X_\e \backslash \partial X_\e$ is connected and dense in $X_\e$, where $\partial X_\e = X_\e \cap \partial \BR_\e$;
				\item  $\varphi_{\e^{-1}}: (X_\e, \partial X_\e) \to (\S, \partial \S)$ is essential in \v{C}ech cohomology with rational coefficients;
				\item $\e'  \le \e$ implies $X_{\e'} \subseteq X_\e$;
				\item $\cap_{\e > 0} X_\e = \{\, (\s, \s) \mid \s \in \S^*\, \}$.
			\end{enumerate}	
		\item There exists a sequence $\S^k$ of closures of $\e_k$-perfect equilibria that converges to $\S^*$ in the Hausdorff topology as $\e_k \to 0$ and such that:
			\begin{enumerate}
				\item  $\S^k \backslash \partial \S$ is connected and dense in $\S^k$;
				\item $\varphi_K: (\proj^{-1}(\S^k), \proj^{-1}(\S^k) \cap \partial \BR_{\e_k}) \to (A, \partial A)$ is essential for some $K \ge \e_k^{-1}$.
			\end{enumerate}	
	\end{enumerate}
\end{proposition}

\begin{proposition}\label{prop Axiom D*}
	Stability satisfies Axiom $D^\ast$.
\end{proposition}

Let $\bar G$ be obtained from $G$ by the addition of strongly irrelevant dominated strategies. Let $f: \bar \S \to \S$ be a map satisfying condition (2) of Definition 4.4.  We show Proposition \ref{prop Axiom D*} using two claims.  In what follows, $\bar \BR$ refers to the graph of the best-reply correspondence of $\bar G$.

\begin{claim}
	Every stable set of $\bar G$ is stable in $G$ as well.
\end{claim}

\begin{proof}
	Let $\S^*$ be a stable set of $\bar G$. There exist $\e_0$ and semialgebraic sets $\bar X_\e \subseteq \bar \BR_\e$ for $0 < \e \le \e_0$ satisfying the conditions of statement (2) of Proposition \ref{prop stability}. For each $0 < \e \le \e_0$, let $\bar \S^\e$ be the set of $\s$ such that $(\s, \s) \in \bar X_\e$. We will show that there exists a stable set $\S^\e$ of $G$ with $\S^* \subseteq \S^\e \subseteq \bar \S^\e$. Since $\cap_{0 < \e \le \e_0} \bar \S^\e  = \S^*$, the compactness property of stable sets implies the stability of $\S^*$ in $G$. 
	
	Fix $0 < \e \le \e_0$ such that $\bar \s \in \bar U$ for each $(\bar \s, \t) \in \bar X_\e$, where $\bar U$ is the neighborhood given by condition (2) of Definition 4.4.  For each $(\bar \s, \t) \in \bar X_\e$, $(f(\bar \s), \t) \in \BR$. Let $\hat X$ be the set of $(\bar \s, \t) \in \bar X_\e$ such that $\bar \s \notin \partial \bar \S \backslash (\S \backslash \partial \S)$.  Define $\a: \hat X \to [0, 1]$ by $\a(\bar \s, \t) = \min \{\, \d \in [0, 1] \mid (f(\bar \s), \t) \in \BR_{\d}\, \}$. $\a$ is clearly a continuous function.  Let $W$ be the closure of the set of $(\d, f(\bar \s), \t)$ such that $(\bar \s, \t) \in \hat X$ and $\d \ge \a(\bar \s, \t)$. Observe that if $(\d, \s, \t) \in W$, then $(\d', \s, \t) \in W$ for all $\d' > \d$; also if $(\s, \t) \in \bar X_\d$ for $\d \le \e$, then $(\d, \s, \t) \in W$. Let $\partial W$ be the set of $(\d, \s, \t) \in W$ such that $\s \in \partial \S$. $f(\bar \s) \notin \partial \S$ if $\bar \s \notin \partial \bar \S \backslash (\S \backslash \partial \S)$;  therefore, $W \backslash \partial W$ is dense in $W$. 
	
	Let $\g: W \to [0, 1]$ be the projection onto the first factor: $\g(\d, \s, \t) = \d$.  As $W$ is semialgebraic, by the Generic Local Triviality Theorem (cf.~Theorem 9.3.2 of \cite{BCR1998}) there exist: (a) $0 < \d_0 \le 1$; (b) a semialgebraic pair $(F, \partial F)$; (c)  a homeomorphism $h: (0, \d_0] \times (F, \partial F) \to (\g^{-1}((0, \d_0]), \g^{-1}((0, \d_0]) \cap \partial W)$ such that $\g\circ h$ is the projection onto the first factor.  $W \backslash \partial W$ being dense in $W$, we have that $F \backslash \partial F$ is dense in $F$. There now  exist finitely many pairs $(F^k, \partial F^k) \subseteq (F, \partial F)$ such that: the sets $F^k \backslash \partial F^k$ are the connected components of $F \backslash \partial F$; $F^k \backslash \partial F^k$ is dense in $F^k$ for each $k$; $\cup_k (F^k, \partial F^k) = (F, \partial F)$. For each $k$ let $(W^k, \partial W^k) = h((0, \d_0] \times (F^k, \partial F^k))$. As with $W$, if  $(\d, \s, \t) \in W^k$ for some $k$, then $(\d', \s, \t) \in W^k$ for all $\d < \d' \le \d_0$; and if $(\s, \t) \in \hat X \cap \bar X_\d$ for $\d \le \d_0$, then $(\d, \s, \t)$  belongs to $W^k$.
	
	There exists $0 < \e_1 < \min (\e, \d_0)$ such that $\a(\cdot) \le \d_0$ on the  connected set $\hat X \cap \bar X_{\e_1}$.  Hence there exists a unique $k$, say $k = 1$, such that $(\a(\bar \s, \t), f(\bar \s), \t) \in W^1 \backslash \partial W^1$ for all $(\bar \s, \t) \in \hat X \cap \bar X_{\e_1}$. 
	
	For each $0 < \d \le \e_1$, let $X_\d$ be the set of $(\s, \t)$ such that $(\d, \s, \t) \in W^1$.  Let $\S^\e$ be the set of $\s$ such that $(\s, \s) \in \cap_{0 < \d \le \e_1} X_\d$. We will show that $\S^\e$ contains $\S^*$ and that it is a stable set, i.e.,~that the sets $X_\d$ satisfy the properties of Proposition \ref{prop stability}.  To show that $\S^*$ is contained in $\S^\e$, take $\s \in \S^*$. There exists a sequence $(\d^k, \bar \s^k, \t^k)$  converging to $(0, \s, \s)$ such that  $(\bar \s^k, \t^k) \in \bar X_{\d^k} \backslash \partial \bar X_{\d^k}$ for each $k$. Therefore, for each $0 < \d \le \e_1$, $(\d, f(\bar \s^k), \t^k) \in W^1$ for large $k$, giving us that $(\d, \s, \s) \in X_\d$, which shows that $\s$ belongs to $\S^\e$. 
	
	The set $\S^\e$ was constructed to satisfy condition (d) of statement (2) of Proposition \ref{prop stability}. We have to verify the other three properties.
	The $X_\d$'s are a nested sequence of sets in $\BR$ since $(\d, \s, \t) \in W^1$ implies $(\d', \s, \t) \in W^1$ for all $\d \le \d' \le \e_1$.  The connectedness property follows from the fact that $F^1 \backslash \partial F^1$ is connected and dense in $F^1$ and $(F^1, \partial F^1)$ is homeomorphic to $(X_\d, \partial X_\d)$.    As for essentiality, for each $0 < \d \le \e_1$,  the essentiality of $\bar \varphi_{\d^{-1}}: (\bar X_\d, \partial \bar X_\d) \to (\bar \S,  \partial \bar \S)$ implies the essentiality of its restriction to the inverse image, call it $(\hat X_\d, \partial \hat X_\d)$, of $(\S, \partial \S)$ (see the Theorem in  Mertens \cite{M1992b}).  The pair $(\hat X_\d, \partial \hat X_\d)$, viewed now as a subset of $(\BR_\d, \partial \BR_\d)$, is contained in $(X_\d, \partial X_\d)$. Therefore, the map $\varphi_{\d^{-1}}$ from $(X_\d, \partial X_\d)$ to $(\S, \partial \S)$ is essential, and the proof of the claim is complete.
\end{proof} 

We need a preliminary lemma before our second claim.  The map $f: \bar \S \to \S$ has a unique multilinear extension to a map from $\bar A$ to $A$, where $\bar A$ is the affine space generated by $\bar \S$. We still use $f$ to denote this extension. For each $0 \le \d \le 1$, let $\bar A^\d$  be the set of $\bar \s$ in $\bar A$ such that: (1) for each $n$, $\bar \s_{n, \bar s_n} \ge 0$ for all $\bar s_n \notin S_n$, and $\sum_{\bar s_n \notin S_n} \bar \s_{n, \bar s_n} \le \d$; (2) $f(\bar \s) \in \S$.  Let $\bar \S^\d = \bar A^\d \cap \bar \S$. For $\d > 0$, let $\partial \bar A^\d$ (resp.~$\partial \bar \S^\d$) be the relative boundary of $\bar A^\d$ (resp.~$\bar \S^\d$) in $\bar A$.  Obviously, $\bar A^\d$ is a compact subset of $\bar A$ for all $\d$, and when $\d = 0$, it is just the set $\S$ viewed as a subset of $\bar A$.  

Let $d$ be the difference in the dimensions of $\bar \S$ and $\S$.  Recall from Appendix IV.2 of \cite{M1991} that a $d$-ball bundle is a triplet $((E, \partial E), (B, \partial B), p)$ where $(E, \partial E)$ and $(B, \partial B)$ are compact pairs, $p: E \to B$ is a continuous map with $p^{-1}(\partial B) \subseteq \partial E$ and with the following property: for each $x \in B \backslash \partial B$, there is a neighborhood $U$ of $x$ in $B \backslash \partial B$ and a homeomorphism $h: U \times ([0, 1]^d, \partial ([0, 1]^d)) \to (p^{-1}(U), p^{-1}(U) \cap \partial E)$ such that $p \circ h$ is the projection onto the first factor.

\begin{lemma}\label{lem ball bundle}
	For each sufficiently small  $\d > 0$, $((\bar A^\d, \partial \bar A^\d), (\S, \partial \S), f)$ is a $d$-ball bundle.
\end{lemma}

\begin{proof}
	Let $N_1$ be the nonempty subset of players $n$ for whom $\bar S_n \neq S_n$. Define $\g: \bar A^1 \to \S \times [0, \d]^{d}$ by $\g(\bar \s) = (f(\bar \s), {(\bar \s_{\bar s_n})}_{n \in N_1, \bar s_n \notin S_n})$.  It is easily checked using the formula for $f$ that  the derivative of $\g$ is nonsingular at each point in $\bar A^0$. Hence for small $\d$,  $\g$ is a homeomorphism of $\bar A^\d$  to its image, which in turn is easily seen to be homeomorphic to $\S \times {[0, \d]}^d$.
\end{proof}	

\begin{claim}
	Every stable set of $G$ is stable in $\bar G$ as well.
\end{claim}

\begin{proof}
	Let $\S^*$ be a stable set of $G$.  There exist $\e_0 > 0$ and, for every $0 < \e \le \e_0$, a set $X_\e$ satisfying the properties of statement (2) of Proposition \ref{prop stability}.  Take a sequence  $\e_i$ of positive numbers converging to zero. For each $i$, let $\S^i$ be the projection of $X_{\e_i}$ to $\S$. Each $\S^i$ is semialgebraic with $\S^i \backslash \partial \S$  connected and dense in it. The sequence $\S^i$ is nested and converges to $\S^*$.  For each $i$, there exists $\d_i > 0$ such that if $\bar \s \in \bar \S^{\d_i} \backslash \partial \bar \S$ and $f(\bar \s) \in \S^i$, then $\bar \s$ is a $2\e_i$-perfect equilibrium of $\bar G$.  We can choose the sequence $\d_i$ to be monotonically decreasing and converging to zero.  $f^{-1}(\S^i \backslash \partial \S) \cap (\bar \S^{\d_i} \backslash \partial \bar \S^{\d_i})$ has finitely many connected components, exactly one of which contains $\S^i \backslash \partial \S$ in its closure. Call the closure of this component $\bar \S^i$.
	
	We now verify that the sequence of sets $\bar \S^i$ satisfies the conditions of statement (3) of Proposition \ref{prop stability}. For each $i$, $\bar \S^i \backslash \partial \bar \S$ is a connected set of $2\e_i$-perfect equilibria and is dense in $\bar \S^i$. The sets $\bar \S^i$ are nested and converge to $\S^*$. Hence, to prove that $\S^*$ is stable in $\bar G$, it is sufficient to show that each $\bar \S^i$ satisfies condition (3b) of Proposition \ref{prop stability} for large $i$. 
	
	Fix $i$ large enough such that $\d_i$ is small enough for Lemma \ref{lem ball bundle} to hold.  Let $\hat \S^i$ be the set of $\bar \s$ in $\bar A^{\d_i}$ such that $f(\bar \s) \in \S^i$.  Let $\hat X^i$ be the set of $(\bar \s, \t) \in \hat  \S^i \times \S$ such that $(f(\bar \s), \t) \in X_{\e_i}$.  Let $\partial \hat X^i$ be the union of  $f^{-1}(\partial X_{\e_i})$ and the set of points $(\bar \s, \t)$ for which $\bar \s \in \partial \bar \S^{\d_i}$. Then $((\hat X^i, \partial \hat X^i), (X_{\e_i}, \partial X_{\e_i}), \bar f)$ is a $d$-ball bundle, where $\bar  f: \hat X^i \to X_{\e_i}$ is the map that sends $(\bar \s, \t)$ to $(f(\bar \s), \t)$.  Fix $K \ge 2\e_i^{-1}$. The claim is now proved if we show that $\bar \varphi_K: (\hat X^i, \partial \hat X^i) \to (\bar A, \partial \bar A)$ is essential. 

	By a retraction of a point $x$ to a  compact convex set $C$, we mean that we map $x$ to the point $y \in C$ that minimizes the $\ell_2$-distance between $x$ and points in $C$. Let $r: A \to \S$ be the map that retracts each $\s \in A$ to $\S$,  and define $\bar r: \bar A \to \bar \S$ similarly.  Let $\tilde f: \bar \S \to \S$ be  the retraction sending points $\bar \s \in \bar \S$ to $\S$. (As usual, we view $\S$ as a subset of $\bar \S$.)  $((\bar \S, \partial \bar\S), (\S, \partial \S), \tilde f)$ is a $d$-ball bundle.  In fact, there exists a function $\tilde h: \S \times [0, \d]^d \to \bar \S$ such that $\tilde f \circ \tilde h$ is the projection onto the first factor and $\tilde h$ maps $(\S \backslash \partial \S) \times ([0, \d]^d, \partial ([0, \d]^d))$ homeomorphically onto $(\bar \S \backslash \partial \S, \partial \bar \S \backslash \partial \S)$.
	
	Define $\tilde \varphi^i: (\hat X^i, \partial \hat X^i) \to (\bar \S, \partial \bar \S)$ by $\tilde \varphi^i(\bar \s, \t) = \tilde h(r \circ \varphi_{K}(f(\bar \s), \t), {(\bar \s_{n,\bar s_n})}_{n\in N, \bar s_n \notin S_n})$, where $K \ge \e_i^{-1}$ is such that $\varphi_K(f(\bar \s), \t) \notin (\S \backslash \partial \S)$ if $(\bar \s, \t) \in \partial \hat X^i$.   Then $r \circ \varphi_{K} \circ \bar f  = \tilde f \circ \tilde \varphi^i$.   Since $\varphi_{\e_i^{-1}}$ is essential, $\varphi_K$ and, therefore, $r \circ \varphi_K$ are essential. Hence,  by the Thom Isomorphism Theorem (cf.~the Theorem of Appendix IV.3 in \cite{M1991}),  $\tilde \varphi^i$ is essential. 
	
	To finish the proof, we show that $\bar \varphi_{K}$ and  $\tilde \varphi^i$ are homotopic as maps between pairs. The two maps agree on the set of points of the form $(\s, \t) \in \hat X^i$; therefore, if we construct a homeomorphism $\bar \psi$ between $\bar \S$ and a convex set $\tilde \S$ such that $\partial \bar \S \backslash (\S \backslash \partial \S)$ is mapped to a convex set, then there is a linear homotopy between $\bar \psi \circ \bar r \circ \bar \varphi_{K}$ and $\bar \psi \circ \tilde \varphi^i$.  Thus, $\bar \varphi_{K}$ and  $\tilde \varphi^i$ are homotopic, which gives us that $\bar \varphi_K$ is essential as well.
\end{proof}
	
\section{Verifying Axioms I and $I^\ast$}
	
	In this appendix, we prove that essentiality and stability satisfy Axioms I and $I^\ast$, respectively.  We start with Axiom I.
	
	\begin{proposition}\label{prop Axiom I}
	Essentiality satisfies Axiom I.
	\end{proposition}
	
	\begin{proof}
		Let $\tilde G$ weakly embed $G$ via the map $f:\tilde \S_\N \times \tilde \S_O \to \S$.  Let $\tilde \S^*$ be an essential component of $\tilde G$ and let $\S^* = f(\tilde \S^*)$.  $\S^*$ is obviously a connected set of equilibria. Let $\hat \S^*$ be the component of equilibria of $G$ that contains $\S^*$.  We will show that $\S^* = \hat \S^*$ and that its index is nonzero. We use $\BR$ and $\tilde \BR$ to denote the best-reply correspondences in $G$ and $\tilde G$, respectively.
		
		Let $r_O: \tilde \S_O \to \bar \S_O$ be a retraction and let $r: \tilde \S_{\N} \times \tilde \S_{O} \to \tilde \S_{\N} \times \bar \S_O$ be the map $r(\tilde \s_{\sn}, \tilde \s_{\so}) = (\tilde \s_{\sn}, r_{\so}(\tilde \s_{\so}))$.  Choose a closed neighborhood $\bar U$ of $\tilde \S^*$ in $\tilde \S_{\N} \times \bar \S_O$ that is disjoint from the other components of equilibria of $\tilde G$. Let $\tilde U = r^{-1}(\bar U)$.  For each $n$, there is a linear injection $i_n: \S_n \to \tilde  \S_n$ such that for each $\bar \s_O \in \bar \S_O$, $f(i_n(\s_n), \bar \s_O)  = \s_n$ for all $\s_n$. Define $\hat \BR_\N: \tilde  U \to  \tilde \S$ by $\hat \BR_\N = i\circ \BR \circ f\circ r$, where $i = {(i_n)}_n$. Let $\hat \BR =  \hat \BR_\N \times \tilde  \BR_O$. The $f$-projection of the set of fixed points of $\hat \BR$ is contained in $\S^*$. $\hat \BR$ is homotopic to $\tilde \BR$ under a linear homotopy with no fixed points on the boundary of $\bar U$. Hence the index of $\tilde U$ under $\hat \BR$ is the same as that  under $\tilde \BR$, and in particular nonzero.  
		
		Suppose either that $\S^* \neq \hat \S^*$ or that the index of $\hat \S^*$ is zero.  Choose a closed neighborhood $\bar V \subseteq \bar U$ of $\tilde \S^*$ in $\tilde \S_{\N} \times \bar \S_O$ such that its $f$-projection to $\S$ does not contain $\hat \S^*$ if $\S^* \neq \hat \S^*$. Let $\tilde V = r^{-1}(\bar V)$. For each $\e > 0$, there exists a function $g: \S \to \S$ such that: (1) the graph of $g$ is within $\e$ of the graph of $\BR$; (2) $g$ has no fixed points in $f(\tilde V)$. Define $\hat g_\N:  \tilde V \to \tilde \S_\N$ by  $\hat g_\N = i\circ g \circ f \circ r$ and $\hat g = \hat g_\N \times \tilde \BR_O$.  Then $\hat g$ is within $\e$ of the graph of $\hat \BR$ and has no fixed points in $\tilde V$, which is impossible since the index of $\hat \BR$ is nonzero.  Hence $\S^* = \hat \S^*$ and its index is nonzero.    	
		
		To prove the other direction, take an essential component $\S^*$ of $G$, whose index is, say, $c \neq 0$.  Let $U$ be a closed neighborhood of $\S^*$ that is disjoint from the other components of equilibria of $G$.   $\BR \times \text{Id}: U \times \tilde \S_O \to \S \times \tilde \S_O$ has index $c$ as well, where $\text{Id}$ is the identity map.  The index does not change if we replace $\text{Id}$ with $\bar \BR_O$ given by $\bar  \BR_O(\s, \bar \s_O) = \tilde \BR_O(i(\s), \bar \s_O)$ or if we replace $\S$ with $i(\S)$. Thus, for the correspondence $\bar \BR: i(\S) \times \tilde \S_O \to i(\S) \times \tilde \S_O$ given by $\bar \BR(i(\s), \tilde \s_{\so}) = i(\BR(\s)) \times \bar \BR_O(i(\s), \tilde  \s_{\so})$, the index over $i(U) \times \tilde \S_O$ is still $c$.  
		
		Let $\bar U$ be the set of $\tilde \s \in  \tilde \S_{\N} \times \bar \S_O$ such that $f(\tilde \s_{\sn}, \tilde \s_{\so}) \in U$, and let $\tilde U = r^{-1}(\bar U)$, where $r$ is the retraction defined earlier in the proof. Let $\hat \BR_\N = i \circ \BR \circ f \circ r$ and $\hat \BR = \hat \BR_{\N} \times \tilde \BR_O$. The restriction of $\hat \BR_\N$ to $\tilde U \cap (i(\S) \times \tilde \S_O)$ is $\bar \BR$ and hence its index is $c$.  As $\hat \BR$ maps into $(i(\S) \times \tilde \S_O)$, the index of $\hat \BR$ over $\tilde U$ is also $c$.  As $\tilde \BR$ is linearly homotopic to $\hat \BR$ over $\tilde U$, its index, too, is $c$.   In particular, there is at least one essential component $\tilde \S^*$ whose $f$-image is contained in $\S^*$.  The same proof as above shows that $f(\tilde \S^*)$ must in fact be $\S^*$.    
	\end{proof}

\begin{proposition}\label{prop Axiom I*}
	Stability satisfies Axiom $I^*$.
\end{proposition}

Let $\tilde G$ strongly embed $G$ via the map $f$.  As with Axiom $D^*$, we prove Proposition \ref{prop Axiom I*} via two claims. The claims use alternative characterizations of stability that obtain when we use other spaces of perturbations. In Definition 2.3 we used the space $T = \{\d \t \mid 0 \le \d \le 1, \t \in \S \, \}$ as the space of perturbations.  We could equivalently have used $[0, 1] \times \S$ as the space of perturbations, i.e.,~vectors $(\d, \t) \in [0, 1] \times \S$ rather than vectors $\d\t$; or $[0, 1]^\N \times \S$, where each player $n$ has a different probability of tremble; or vectors ${(\d_n\t_n)}_{n \in \N}$ with $\d_n \in [0, 1]$ and $\t_n \in \S_n$ for each $n$. Each of these perturbations defines a perturbed game in one of two ways. Take a vector $\d\t$ in $T$, for example:  either we change the payoff at a profile $\s$ to the payoff at $(1-\d)\s + \d\t$, as we did in our definition, or we restrict the strategy space $\S(\d \t)$ to vectors $\s$ such that $\s_{n,s_n} \ge \d\t_{n,s_n}$ for all $n \in \N$ and $s_n \in S_n$.  The latter is what we call a perturbed game $G(\d\t)$.  A definition of stability based on any of these spaces of perturbations is equivalent to that based on any other in this class of perturbations---see Corollary 6 of Section 5 of \cite{M1991}.

In our case, with outsiders, we can add to the variants of the set of perturbations (using the same logic as in the above case). We could use perturbations $({(\d_{n}\tilde \t_{n})}_n, \d_{\so}\tilde \t_{\so})$, or perturbations in which the outsiders have a common tremble probability and the insiders a player-specific one.  
Finally, we also use the following tremble probabilities:  for a fixed $0 < \beta <1$,  trembles $\hat P_\e^\b$, which is the set of $ ({(\d_{n}\t_{n})}_n, \d_{\so} \t_{\so})$ satisfying the condition $\d_{\so}^\beta = \min_n \d_{n} \le \e$. For the last set of trembles, if there is a subset $\hat X_\e^\b$ of the graph over $\hat P_\e^\b$ that satisfies the essentiality condition, then there is a subset $\tilde X_\e \supset \hat X_\e^\b$ of the graph of equilibria over the set  $\tilde P_\e$ of perturbations $({(\d_{n}\t_{n})}_n, \d_{\so} \t_{\so})$ with $\d_n \le \e$ and $\d_{\so} \le \e$ that satisfies the conditions for stability as well.

The proofs of the claims also use the concept of $0$-essentiality from Mertens \cite{M1992b}. A map $f: X \to Y$ between compact spaces is $0$-essential if for every compact pair $(Z, \partial Z)$, and map $g: Z \to Y$, letting $Q$ be the fibered product $\{\, (x, z) \in X \times Z \mid f(x) = g(z) \, \}$ of $f$ and $g$, and  $q: (Q, \partial Q) \to (Z, \partial Z)$ the projection, with $\partial Q = q^{-1}(\partial Z)$, $\check{H}^*(q)$ is one-to-one.  An important property of $0$-essentiality is that  a composition of $0$-essential maps is $0$-essential.

In our context, typically $Y$ is a simplex of perturbations $P$, while $X$ is a subset of the graph of equilibria and $f$ the projection map. When $\proj:(X, \partial X) \to (P, \partial P)$ is essential in cohomology, it is $0$-essential. $0$-essentiality also gives us that if $R$ is a simplex contained in $P$, then the essentiality of $\proj$ from $X$ to $P$ implies the essentiality of the restriction of the projection to the inverse image of $(R, \partial R)$ as well.

\begin{claim}
	Let $\tilde \S^*$ be a stable set of $\tilde G$. Then $\S^* \equiv f(\tilde \S^*)$ is stable in $G$.
\end{claim}

\begin{proof}
Let $\tilde T = {[0, 1]}^{\N} \times [0, 1] \times \tilde \S_\N \times \tilde \S_{\so}$.  For each $(\tilde \e_{\scriptscriptstyle \N}, \tilde \e_{\so}, \tilde \t_{\sn}, \tilde \t_{\so})$ in $\tilde T$, there is a perturbed game where the strategy sets are the same as in $\tilde G$ but where the payoffs to a profile $(\tilde \s_{\sn}, \tilde \s_{\so})$ are the payoffs in $\tilde G$ to the profile $({((1-\tilde \e_n)\tilde \s_n + \tilde \e_n \tilde \t_n)}_{n}, (1-\tilde \e_{\so})\tilde \s_{\so} + \tilde \e_{\so} \tilde \t_{\so})$.  Let $\tilde E \subseteq \tilde T \times \tilde \S_{\sn} \times \tilde \S_{\so}$ be the closure of the graph of the equilibria of perturbed games $(\tilde \e_{\sn}, \tilde \e_{\so}, \tilde \t_{\sn}, \tilde \t_{\so}) \in \tilde T \backslash \partial \tilde T$, and let $\tilde \proj: \tilde E \to \tilde T$ be the natural projection.  

For each $\e \ge 0$,  let $\tilde T_{\e}$ be the set  of $(\tilde \d_{\sn}, \tilde \d_{\so}, \tilde \t_{\sn}, \tilde \t_{\so}) \in \tilde T$ such that $\tilde \d_n \le \e$ for each $n$ and $\tilde \d_{\so} \le \e$; for $\e > 0$, let $\partial \tilde T_{\e}$ be the boundary of $\tilde T_\e$. For each subset $\tilde Y$ of $\tilde E$ and $\e > 0$, let $(\tilde Y_{\e}, \partial \tilde Y_{\e}) = {\tilde \proj}^{-1}(\tilde T_{\e}, \partial \tilde T_{\e}) \cap \tilde Y$ and let $\tilde Y_0 = \proj^{-1}(\tilde T_0) \cap \tilde Y$. 

Since $\tilde \S^*$ is stable, there exists a subset $\tilde Y$ of $\tilde E$ such that: (1) $\tilde \S^*$ is the set of $(\tilde \s_{\sn}, \tilde \s_{\so})$ for which $(0, 0, \t_{\sn}, \t_{\so}, \tilde \s_{\sn}, \tilde \s_{\so}) \in \tilde Y_0$ for some $(\t_{\sn}, \t_{\so})$; (2) for each neighborhood $\tilde V$ of $\tilde Y_{0}$ in $\tilde Y$, there is a connected component of $\tilde V \backslash \partial \tilde Y_{1}$ whose closure is a neighborhood of $\tilde Y_{0}$ in $\tilde Y$; (3) the projection from $(\tilde Y_{\e}, \partial \tilde Y_{\e})$ to $(\tilde T_{\e}, \partial \tilde T_{\e})$ is essential in \v{C}ech cohomology with rational coefficients for some (and then all smaller) $\e > 0$. We can assume, if necessary by replacing $\tilde Y$ with a subset $\tilde Y_\e$ for a sufficiently small $\e$, that for $(\tilde \d_{\sn}, \tilde \d_{\so}, \tilde \t_{\sn},  \tilde \t_{\so}, \tilde \s_\N, \tilde \s_O)$ in $\tilde Y$, $({((1-\d_{n})\s_n + \d_n \t_n)}_n, (1-\tilde \d_{\so})\tilde \s_{\so} + \d_{\so}\tilde \t_{\so})$ belongs to $\tilde U$, the neighborhood of $\tilde \S_\N \times \bar \S_O$ given in Definition 5.2.

Let $P$ be the set of ${(\e_n\t_n)}_{n \in \N}$ such that $0 \le \e_n \le 1$ and $\t_n \in \S_n$ for each $n$.  Each $\eta \in P$ defines a perturbed game where the set of strategies is $\S(\eta)$.  Let $E$ be the equilibrium graph over $P$.  Define a map $\varphi: \tilde Y \to E$ as follows.  $\varphi(\e_{\sn}, \e_{\so}, \tilde \t_{\sn}, \tilde \t_{\so}, \tilde \s_{\sn}, \tilde \s_{\so}) = (\eta, \s)$ where, letting $\hat \s_{\so} = (1-\e_{\scriptscriptstyle O})\tilde \s_{\so} + \e_{\so} \tilde \t_{\so}$, for each $n$: (1) $\eta_{n} = \e_n f_n(\tilde \t_n, \hat \s_{\so}) + (1-\e_n) \sum_{\tilde s_{\so} \notin \bar S_O} \hat \s_{\so}(\tilde s_{\so}) f_n(\tilde \s_n, \tilde s_{\so})$; (2) $\s_n = f_n((1-\e_n)\tilde \s_n + \e_n \tilde \t_n, \hat \s_{\so})$.    It is easily checked that $\varphi$ is a well-defined map into $E$.  Let $X = \varphi(\tilde Y)$. $\S^*$ is the set of $\s$ such that $(0, \s) \in X_0$.   Also, the connexity property holds for $X$.  We merely have to show that the projection from $X_\e$ to $P_\e$ is essential for all small $\e > 0$.

For  each $\e \ge 0$,  let $\tilde P_{\e}$ be the set  of $(\tilde \d_{\sn}, \tilde \d_{\so}, \tilde \t_{\sn}, \tilde \t_{\so}) \in \tilde T_1$ such that $\tilde \d_{\so} \le \min_n \tilde \d_n$ and $\max_n \tilde \d_n \le \e$; for $\e > 0$, let $\partial \tilde P_{\e}$ be the boundary of $\tilde P_\e$. Let $(\tilde X_\e, \partial \tilde X_\e) = \tilde Y_\e \cap (\tilde \proj^{-1}(\tilde P_\e, \partial \tilde P_\e))$.  As $\tilde P_\e$ is a simplex contained in $\tilde T_\e$ and the projection from $\tilde Y_\e$ to $\tilde T_\e$ is essential for all small $\e > 0$,  the projection from $(\tilde X_\e, \partial \tilde X_\e)$ to $(\tilde P_\e, \partial \tilde P_\e)$ is essential for all such $\e$. 

Define a function $h: \tilde P_{\e} \to \prod_n \Re^{S_n}$ by: $h(\d_{\sn}, \d_{\so}, \tilde \t_{\sn}, \tilde \t_{\so}) = \eta$ where, letting $\hat \t_{\so}^*$ be the uniform profile of the outsiders, for each $n$, $\eta_{n} = \d_n f_n(\tilde \t_n, \hat \t_{\so}^*) + \d_{\so} \sum_{\tilde s_{\so} \notin \bar S_O} \hat \t_{\so}^*(\tilde s_{\so}) f_n(\tilde \t_n, \hat \t_{\so}^*)$. Let $\hat P_\e = h^{-1}(P_\e)$. $\hat P_\e$ is a full-dimensional subset of $\tilde P_{\e}$ for small $\e$, and the map $h: \hat P_\e \to P_\e$ is $p$-essential. Observe that for $\eta = h(\d_{\sn}, \d_{\so}, \tilde \t_{\sn}, \tilde \t_{\so})$ if $\sum_{s_n} \eta_{n, s_n} = \e$ for some $n$, then necessarily $\d_n \ge \frac{\e}{2}$. Construct a continuous function $t:[0, \e] \to [0, 2]$ such that $t(\d) = 1$ if $\d \le \frac{\e}{4}$; $t(\d) = \frac{\e}{\d}$ if $\d \ge \frac{\e}{2}$; and  $t$ is linear on $[\frac{\e}{4}, \frac{\e}{2}]$. Modify the map $h$ to a map $\hat h$ with $\hat h_n(\cdot) = t(\sum_{n,s_n} h_{n,s_n}(\cdot))h_n(\cdot)$.  Clearly, the map $\hat h$ is also $0$-essential.

Obviously, the projection from $\hat X \equiv \tilde \proj^{-1}(\hat P_\e) \cap \tilde X_\e$ to $\hat P_\e$, and thus its composition with $\hat h$, is also $0$-essential.  Hence $\hat h \circ \tilde \proj: (\hat X, \partial \hat X) \to (P_\e, \partial P_\e)$ is essential, where $\partial \hat X$ is the inverse image of $\partial P_\e$ under $\hat h \circ \tilde \proj$.  When comparing $\hat h \circ \tilde \proj$ with the map $\proj \circ \varphi$ as maps from $\hat X$ to $P_1$, we see that for $\eta = \hat h \circ \tilde \proj(\d_{\sn}, \d_{\so}, \tilde \t_{\sn}, \tilde \t_{\so}) \in \partial P_\e$ : (1) if $\eta$ belongs to $\partial P_1$ (i.e.,~$\eta_{n,s_n} = 0$ for some $n, s_n$), then under $\proj \circ \varphi$, the same property holds; (2) if  $\sum_{s_n} \eta_{n,s_n} = \e$, then the corresponding sum under the image of $\proj \circ \varphi$ is at least $\frac{\e}{2}$.  Therefore, by retracting $\proj \circ \varphi$ from $P_1$ to $P_\e$, we see that the two maps are linearly homotopic as maps from $(\hat X, \partial \hat X)$ to $(P_\e, \partial P_\e)$.  Therefore, $\proj \circ \varphi$ is essential, implying that $\proj$ is an essential map.  Thus, we have shown that  $\S^*$ is stable.
\end{proof}

We need the following lemma to prove the claim in the other direction.  It is similar in spirit to Lemma \ref{lem ball bundle} and, especially, Lemma C.2 of GW; therefore, its proof is omitted.    For $\d > 0$, let $\tilde \S_O^\d$ be the set of $\tilde \s_{\so} \in \tilde \S_O$ such that for each outsider the probability under $\tilde \s_{\so}$ of a pure strategy that does not belong to his factor of $\bar \S_O$ is no more than $\d$.  Let $d = \text{dim}(\tilde \S_\N) - \text{dim}(\S)$ and $(B, \partial B) = ({[0, 1]}^d, \partial {[0, 1]}^d)$, with the convention that $B$ is a singleton set and $\partial B = \emptyset$ if $d = 0$.

\begin{lemma}\label{lem ball bundle 2}
	
	There exist $\d_{\so} > 0$ and a semialgebraic function $h: \S \times B \times \tilde  \S_O^{\d_{\so}} \to \tilde \S$ such that for each $\tilde \s_{\so} \in \tilde \S_O^{\d_{\so}}$:
	\begin{enumerate}
		\item $f\circ h(\s, b,  \tilde \s_{\so}) = \s$  for all $(\s, b) \in \S \times B$;
		\item  for each $\s \in \S \backslash \partial \S$ and $\tilde \s_{\so}$ such that $\sum_{\tilde s_{\so} \notin \bar S_O} \tilde \s_{\so}(\tilde s_{\so}) <  \min_{n,s_n} \s_{n,s_n}$,  $h(\s, \cdot, \tilde \s_{\so})$ maps $B \backslash \partial B$ homeomorphically onto the set of $\tilde \s_{\sn} \in \tilde \S_\N \backslash \partial \tilde \S_\N$ such that $f(\tilde \s_{\sn}, \tilde \s_{\so}) = \s$.
	\end{enumerate}
\end{lemma}

Let $C\S$, $C\tilde \S_\N$, and $C\tilde \S_O$ be, respectively, the convex cones generated by $\S$, $\tilde \S_\N$, and $\tilde \S_O^{\d_{\so}}$. Then $f$ extends to a multilinear function from $C\tilde \S_\N \times C \tilde \S_O^{\d_{\so}}$ to $C\S$ in the obvious way. $h$ also extends to a map from $C\S \times B \times (C\tilde \S_O \backslash 0)$ to $C\tilde \S_\N$, still denoted $h$,  using the formula: $h(\l\s, b, \tilde \mu \tilde \s_{\so}) = \tilde \mu^{-1}\l h(\s, b, \tilde \s_{\so})$ for $\l \ge 0$ and $\tilde \mu > 0$.  $f \circ h$ is still the projection on the first factor. The map continues to be a homeomorphism between $B \backslash \partial B$ and its image for any fixed $\l, \tilde \mu > 0$, $\s \in \S$, $\tilde \s_{\so} \in \tilde \S_O$, with  $\sum_{\tilde s_{\so} \notin \bar S_O} \tilde \s_{\so}(\tilde s_{\so}) <  \min_{n,s_n} \s_{n,s_n}$.

\begin{claim}
	Let $\S^*$ be a stable set of $G$.  There exists a stable set $\tilde \S^*$ of $\tilde G$ such that $f(\tilde \S^*) = \S^*$.
\end{claim}

\begin{proof}
	It is sufficient to prove the claim under the additional assumption that $\S^*$ is semialgebraic, as the collection of stable sets is compact in the Hausdorff topology and the subcollection of semialgebraic stable sets is dense in it.  The proof proceeds in steps.
	
	\smallskip
	
	\noindent{\bf Step 1.} For each $0 < \e \le 1$, let $\tilde P_\e$ be the set of $({(\tilde \d_n \tilde \t_n)}_n, \tilde \d_{\so}\tilde \t_{\so})$ such that   $0 \le \tilde \d_n \le \e$ for all $n$,  $\tilde \d_{\so}  \le \e$, $\tilde \t_{\sn} \in \tilde \S_{\sn}$, and $\tilde \t_{\so} \in \tilde \S_{\so}$. Let $\partial \tilde P_\e$ be its topological boundary.  Let $\tilde E$ be the  closure of the set of $(\tilde \eta_{\sn}, \tilde \eta_{\so}, \tilde \s_{\sn}, \tilde \s_{\so}) \in (\tilde P_1 \backslash \partial \tilde P_1) \times \tilde \S_{\sn} \times \tilde \S_{\so}$ such that $(\tilde \s_{\sn}, \tilde \s_{\so})$ is an equilibrium of the game where the strategy sets are  $\tilde \S_{\N}(\tilde \eta_{\sn}) \times \tilde \S_O(\tilde \eta_{\so})$.  Let $\tilde \proj: \tilde E \to \tilde P_1$ be the projection map.
	
	Let $\tilde E_0$ be the set of points in $\tilde E$ of the form $(0,0, \tilde \s_{\sn}, \tilde \s_{\so})$. Let $\tilde E_0^*$ be the set of points of that form with $f(\tilde \s_{\sn}) \in \S^*$. We can take a semialgebraic triangulation of $\tilde E$ such that $\tilde E_0^*$ is a full subcomplex.  Let $\tilde X^*$ be the union of the simplices of this triangulation that have a face in $\tilde E_0^*$. $\tilde X^*$ is semialgebraic.  By Lemma 2 of \cite{M1992a}, there exist $\tilde \e_0$ and a finite number of subsets $\tilde X^k$ of $\tilde X^*$ such that: (1) $\cup_k \tilde X^k = \tilde X^*$; (2) for all $0 < \e \le \tilde \e_0$ and $k$, $\tilde X_\e^k \backslash \partial \tilde X_\e^k$ is connected and dense in $\tilde X_\e^k$, and $\tilde X_\e^k \cap \tilde X_\e^l \subseteq \partial \tilde E_1$ for $l \neq k$. To prove our theorem, it is sufficient  to show that for some $k$: (1) $\S^*$ is the set of profiles $f(\tilde \s_{\sn}, \tilde \s_{\so})$ such that   $(0,0, \tilde \s_{\sn}, \tilde \s_{\so}) \in \tilde X^k$; (2) $\tilde X^k$ satisfies the essentiality property.

	\smallskip

	\noindent{\bf Step 2.} Define two functions $\g_1, \g_2: \tilde X^* \to [0,1]$ as follows.  $\g_1(\tilde \eta_{\sn}, \d_{\so}\tilde \t_{\so}, \tilde \s_{\sn}, \tilde \s_{\so}) = \d_{\so}$; and $\g_2(\tilde \eta_{\sn}, \d_{\so}\tilde \t_{\so}, \tilde \s_{\sn}, \tilde \s_{\so}) = \sum_{s_{\so} \notin \bar S_O} \tilde \s_{\so}(s_{\so})$, the probability that a profile not in $\bar S_O$ is chosen under $\tilde \s_{\so}$. Then $\g_1^{-1}(0) \subseteq \g_2^{-1}(0)$, as for each $\tilde \s_{\sn}$, the equilibria of the game $\tilde G_{\so}^{\tilde \s_{\sn}}$ are contained in $\bar \S_O$. By the Lojasiewicz Inequality---cf.~Corollary 2.6.7 of \cite{BCR1998}---there exist $0 < \tilde \e_1 \le \tilde \e_0$ and $0 < \a < 1$ such that $\g_2 \le \g_1^{\a}$ on $\tilde X_{\tilde \e_1}^*$.

	\smallskip

	\noindent {\bf Step 3.}
	Fix now $0 < \beta < \a$, where $\a$ is given in Step 2. For each $0 < \e  \le \tilde \e_1$, where $\tilde \e_1$ is as in Step 2, let $\tilde P_\e^\beta$ be the set of $({(\d_n\tilde \t_n)}_n, \d_{\so}\tilde \t_{\so}) \in \tilde P_\e$ such that $\d_{\so}^\b \le \d_n$ for all $n$, and let $\partial \tilde P_\e^\b$ be its boundary.  For each $k$, let $(\tilde X_\e^{k, \b}, \partial \tilde X_\e^{k, \b}) = \tilde X_\e^k \cap \tilde \proj^{-1}(\tilde P_\e^\b, \partial \tilde P_\e^\b)$ be the corresponding subset of $\tilde X_\e^k$ for perturbations in $\tilde P_\e^\b$. We prove in this step  that for each $k$, the essentiality of the projection from $\tilde X_\e^k$ to $\tilde P_\e$ is equivalent to that of the projection from $\tilde X_\e^{k, \b}$ to $\tilde  P_{\e}^{\b}$. 
	
	For a proof in one direction, observe that the essentiality of the projection from $\tilde X_\e^k$ implies that  it  is $0$-essential,  and thus the projection from $(\tilde X_\e^{k, \b}, \partial \tilde X_\e^{k, \b})$ to $(\tilde P_{\e}^\b, \partial \tilde P_{\e}^\b)$ is essential. Going the other way, the essentiality of the projection from $\tilde X_\e^{k, \b}$ implies its $0$-essentiality and thus the essentiality of the projection from $\hat X_\e^{k, \b}$ to $\hat P_\e^{\b}$, where $\hat P_\e^{\b}$ is the set of $({(\d_n\tilde \t_n)}_n, \d_{\so}\tilde \t_{\so})$ such that $\d_{\so}^\beta = \min_n \d_n$, and $\hat X_\e^{k,\b}$ is the inverse image in $\tilde X_\e^{k, \b}$ of $\hat P_\e^\b$ under the projection map. The set $\hat P_\e^{\b}$ is another set of perturbations for stability, and essentiality with this set of perturbations is equivalent to that for the set $\tilde P_\e^{\b}$ of perturbations. Thus, the essentiality of the projection from $\tilde X_\e^k$ to $\tilde P_\e$ is equivalent to that from $\tilde X_\e^{k, \b}$ to $\tilde P_\e^\b$.
	
	\smallskip
	
	\noindent {\bf Step 4.}
	For each positive integer $L$ and $\e > 0$, let $\tilde P_{\e, L}^\b$ be the set of $({(\d_n\tilde \t_n)}_n, \tilde \d_{\so} \tilde \t_{\so}) \in \tilde P_\e^\b$ such that $\tilde \t_{n,s_n} \ge L^{-1}$ for each $n$ and $s_n$ and let $\tilde X_{\e, L}^{k, \b}$ be the corresponding subset of $\tilde X_\e^{k, \b}$.  If $L > |\tilde S_n|$, then $\tilde P_{\e, L}^\b$ is a nonempty, full-dimensional subset of $\tilde P_\e^\b$. For such large $L < L'$, the inclusion $(\tilde P_\e^\b, \tilde P_\e^\b \backslash (\tilde P_{\e,L'}^\b \backslash \partial \tilde P_{\e, L'}^\b)) \subset (\tilde P_\e^\b, \tilde P_\e^\b \backslash (\tilde P_{\e,L}^\b \backslash \partial \tilde P_{\e, L}^\b))$ induces an isomorphism of their cohomology groups, and $(\tilde P_\e^\b, \partial \tilde P_\e^\b) = \cap_L (\tilde P_\e^\b, \tilde P_\e^\b \backslash (\tilde P_{\e,L}^\b \backslash \partial \tilde P_{\e, L}^\b))$.  We have a similar situation with the sets $\tilde X_{\e, L}^{k, \b}$.  It is then clear that essentiality of the projection from $(\tilde X_{\e}^{k, \b}, \partial \tilde X_{\e}^{k, \b})$ to $(\tilde P_\e^\b, \partial \tilde P_\e^\b)$  is equivalent to the essentiality of the projection from $(\tilde X_{\e, L}^{k, \b}, \partial \tilde X_{\e, L}^{k, \b})$ to $(\tilde P_{\e, L}^\b, \partial \tilde P_{\e, L}^\b)$  for large $L$. For the rest of the proof, fix a large enough $L$ so that this property holds.

	\smallskip

	\noindent{\bf Step 5.} We use $P$ to denote the space of perturbations $\eta \in \prod_n \Re^{S_n}$ and consider perturbed games  obtained by restricting strategy sets to be  $\S(\eta)$  for each perturbation vector $\eta$.  $E$ is the  graph of perturbed equilibria over $P$. There exists  a semialgebraic set $X \subseteq E$ and $0 < \e_0 < 1$ such that: (1) $X_0  = \{\, 0 \, \} \times \S^*$; (2) for each $0 < \e  \le \e_0$, $X_\e \backslash \partial X_\e$ is connected and dense in $X_\e$; (3) the projection map $\proj: (X_\e, \partial X_\e) \to (P_\e, \partial P_\e)$ is essential for all $0 < \e \le \e_0$.

	\smallskip

	\noindent {\bf Step 6.}
	Fix $0 < \e_1 \le \min\{\e_0, \tilde \e_1, \d_{\so}\}$, where $\e_0$ is from Step 5, $\tilde \e_1$ is from Step 2, and $\d_{\so}$ is as in Lemma \ref{lem ball bundle 2}.
	Then the projection from $X_{\e_1}$ to $P_{\e_1}$ is essential.   Let $\tilde P_{\so, \e}$ be the set of $\d_{\so}\t_{\so}$ such that $\d_{\so} \le \e_1$. Let $\hat P_{\e_1}$ be the set of $(\eta_{\sn}, b, \tilde b, \d_{\so} \tilde \t_{\so}) \in P_{\e_1} \times B \times B \times \tilde P_{\so,\e}$ such that  $\d_{\so}^{\a} \le \min_n \eta_{n,s_n}$  for all $n, s_n$, where $B$ is as in Lemma \ref{lem ball bundle 2}.  Let $\partial \hat P_{\e_1}$ be the  set of points $(\eta_{\sn}, b, \tilde b, \d_{\so}\tilde \t_{\so}) \in \hat P_{\e_1}$ such that: (1) $(\eta_{\sn}, b) \in \partial (P_{\e_1, L} \times B)$; or (2) $\d_{\so}\tilde \t_{\so} \in \partial \tilde P_{\so, \e_1}$; or (3) $\d_{\so}^{\a} = \eta_{n,s_n}$ for some $n, s_n$. Let $\hat X_{\e_1}$ be the set of $((\eta_\N, b,  \tilde b, \d_{\so}\tilde \t_{\so}), \s) \in \hat P_\e \times \S$ such that $(\eta_\N, \s) \in X_\e$ and let $\hat \proj$ be the projection map from $\hat X_{\e_1}$ to $\hat P_{\e_1}$.  Let $\partial \hat X_{\e_1}$ be the inverse image in $\hat X_{\e_1}$ of $\partial \hat P_{\e_1}$ under $\hat \proj$.  The essentiality of $\proj$ from $(X_{\e_1}, \partial X_{\e_1})$ to $(P_{\e_1}, \partial P_{\e_1})$ implies that of $\hat \proj$ from $(\hat X_{\e_1}, \partial \hat X_{\e_1})$ to $(\hat P_{\e_1}, \partial \hat P_{\e_1})$.
	
	\smallskip
	
	\noindent {\bf Step 7.}
	Let $\tilde \GG_{\so}$ be the space of all the payoffs for the outsiders over the strategy space $\tilde S_{\sn} \times \tilde S_{\so}$.  For each $\hat G \in \tilde \GG_{\so}$, and $\tilde \s_{\sn} \in \tilde \S_{\sn}$,  we have a well-defined game $\hat G^{\tilde \s_{\sn}}$. Given $\hat G$, $\tilde \s_{\sn} \in \tilde \S_{\N}$, and $(\eta_{\sn}, b, \tilde b, \tilde \eta_{\so}) \in \hat P_{\e_1}$, we define a perturbed game $\hat G^{\tilde \s_{\sn}}(\tilde \eta_{\so})$ as the game where the strategy space is the set $\tilde \S_O(\eta_{\sn}, \eta_{\so})$ of profiles in  $\tilde \S_O(\eta_{\so})$ such that for each outsider $o_i$ and $\tilde s_{o_i}$ that is not in $\bar S_{o_i}$, $\tilde \s_{o_i, \tilde s_{o_i}} \le {|\bar S_O|}^{-1}\eta_{n,s_n}^{\a^{-1}}$ for all $n, s_n$. 
	
	Given $(\eta_{\sn}, b, \tilde b, \eta_{\so})$ and $(\s, \tilde \s_{\so})$ with $\s_{\sn} \ge \eta_{\sn}$ and  $\tilde \s_{\so} \in \tilde \S_O(\eta_{\sn}, \eta_{\so})$, define $\hat \s_{\N}(\eta_{\sn}, \tilde b, \s, \tilde \s_{\so}) = h(\s_{\sn} - \eta_{\sn}, \tilde b, \hat \s_{\so})$ where $\hat \s_{\so}$ is the projection of $\tilde \s_{\so}$ to the coordinates of strategies in $\bar S$; and define $\tilde \eta(\eta_{\sn}, b, \tilde b, \s, \tilde \s_{\so})$ to be $h(\eta - \sum_{\tilde s_{\so} \notin \bar S_O} \tilde \s_{\so}(\tilde s_{\so})f(\tilde s_{\so}, \hat \s_{\so}), b, \tilde \s_{\so})$.
	
	Let $\E$ be the set of $((\eta_{\sn}, b, \tilde b, \d_{\so} \tilde \t_{\so}), \s,  \hat G_O, \tilde \s_{\sn}, \tilde \s_{\so}) \in \hat P_{\e_1} \times \S  \times \tilde \GG_O \times \tilde \S_{\N} \times \tilde \S_O$ such that: (1) $\tilde \s_{\so}$ is an equilibrium of $\hat G^{\tilde \s_{\sn}}(\tilde \eta_{\so})$; (2) $\tilde \s_{\sn} = \hat \s_{\sn}(\eta_{\sn}, \tilde b, \s, \tilde \s_{\so}) + \tilde \eta_{\sn}(\eta_{\sn}, b, \tilde b, \s, \tilde \s_{\so})$. A variation of the KM structure theorem---see the Proposition in \cite{KM1986}---shows that $\E$ is homeomorphic to $\hat P_{\e_1} \times \S \times \tilde \GG_{\so}$ and that the projection map $\tilde g$ has degree one; thus $\tilde g$ is $p$-essential.
	
	\smallskip
	
	\noindent{\bf Step 8.} 
	Let $g: \hat X_{\e_1} \to \hat P_{\e_1} \times \S \times \tilde \GG_{\so}$ be given by $g(\eta_{\sn}, b, \tilde b, \tilde \eta_{\so}, \s) = (\eta_{\sn}, b, \tilde b, \tilde \eta_{\so}, \s, \tilde G)$. Let  $\hat Y_{\e_1}$ be the fibered product of $\tilde g$ and $g$. Since $\tilde g$ is $p$-essential, the projection $\proj^g$ from $(\hat Y_{\e_1}, \partial \hat Y_{\e_1})$ to $(\hat X_{\e_1}, \partial \hat X_{\e_1})$ is essential, where $\partial \hat Y_{\e_1}$ is the inverse image of $\partial \hat X_{\e_1}$ under this projection.  (In particular, this projection is onto.) Therefore, for all $0 < \e \le \e_1$, the projection $\proj^g$ from $(\hat Y_{\e}, \partial \hat Y_{\e})$ to $(\hat X_{\e}, \partial \hat X_{\e})$, and then $\hat \proj \circ \proj^g$ to $(\hat P_\e, \partial \hat P_\e)$, is essential as well.
	
	Again applying Lemma 2 of \cite{M1992a}, there exist $0 < \e \le \e_1$ and a subset $\hat Z_\e$ of $\hat Y_\e$  such that for each $0 < \d \le \e$, the projection from $\hat Z_\d$ to $\hat X_\d$ is essential and $\hat Z_\e$ satisfies the connexity condition, i.e., for all $\d \le \e$, $\hat Z_\d \backslash \partial \hat Z_\d$ is connected and dense in $\hat Z_\d$.   
	
	\smallskip

	\noindent {\bf Step 9.}  Define $H: \hat Z_\e \to \tilde E$ by 
	$$H(\eta_\N, b, \tilde b, \d_O\tilde \t_O, \s, \tilde G, \tilde \s_\N, \tilde \s_O) = (\tilde \eta(\eta_{\sn}, b, \tilde b, \s, \tilde \s_{\so}), \d_0\tilde \s_0, \tilde \s_\N, \tilde \s_O))$$ where $(\tilde \eta(\eta_{\sn}, b, \tilde b, \s, \tilde \s_{\so})$ is defined in Step 7. $H$ maps $\hat Z_\e$ homeomorphically onto image, call it $\hat E_\e$.   As $\proj^g$ from $\hat Z_\e$ to $\hat X_\e$ is onto, we also have that $\S^*$ is the set of $f(\tilde \s_{\sn}, \tilde \s_{\so})$ such that $(0, 0, \tilde \s_{\sn}, \tilde \s_{\so}) \in \hat E_\e$. Thus, $\hat  E_\d \equiv H(\hat Z_\d)$ is contained in $\tilde X^*$ for all small $\d$. By the connectedness property of $\hat E_\d$,  there exists a unique $k$, say $k = 1$, such that $\tilde X^1$ contains $\hat E_\d$ for all small $\d > 0$. Obviously, letting $\tilde \S^*$ be the set of $\tilde \s^*$ such that  $(0, 0, \tilde \s^*) \in X_0^1$, we have $\S^* = f(\tilde \S^*)$. It remains to show that $\tilde X^1$ satisfies the essentiality condition for stability.
	
	\smallskip
	
	\noindent{\bf Step 10.}
	Choose $\d_1 > 0$ small such that  $\hat E_{\d_1}$ is contained in $\tilde X_{\tilde \e_1}^1$.  Choose also $0 < \d < \d_1$ such that $\d^\a \le (1-\d^\a)\d^\b L$ and $\d^{\frac{\a}{\b}} \le L^{-1}\d$. We conclude the proof by showing that the projection from $\tilde X_{\d, L}^{1, \b}$ to $\tilde P_{\d, L}^\b$ is essential.  Define $\psi: \hat P_{\d_1} \to \tilde P_{\d_1}$ by $\psi(\eta_{\sn},  b, \tilde b, \tilde \eta_{\so}) = (\tilde \eta_{\sn}, \tilde \eta_{\so})$ where $\tilde \eta_{\sn}  = \tilde \eta(\eta_{\sn}, b, \tilde b, \s^*, \tilde \s_{\so}^*)$, with 
	$\s^*$ being the profile of uniform distributions  and $\tilde \s_{\so}^*$ a strategy profile with support  $\bar S_O$.
	Let $(Q_{\d_1}, \partial Q_{\d_1}) = \psi(\hat P_{\d_1}, \partial \hat P_{\d_1})$. $\psi$ induces  a homeomorphism between $\hat P_{\d_1} \backslash \partial \hat P_{\d_1}$ and its image $Q_{\d_1} \backslash \partial Q_{\d_1}$.  Then, $\psi \circ \proj^g \circ H^{-1}$ is an essential map from $(\hat E_{\d_1}, \partial \hat E_{\d_1})$ to $(Q_{\d_1},  \partial Q_{\d_1})$. 
	
	We now have that  $\tilde P_{\d,L}^\beta$ is contained in $Q_{\d_1}$ by our choice of $\d$. Since $Q_{\d_1}$ and $\tilde P_{\d, L}^\beta$ are both simplices, the inclusion $(Q_{\d_1},  \partial Q_{\d_1}) \subseteq (Q_{\d_1}, Q_{\d_1} \backslash (\tilde P_{\d, L}^\b \backslash \partial \tilde P_{\d, L}^\b))$ induces an isomorphism of the cohomology groups.  Thus, $\psi \circ \proj^g \circ H^{-1}$ is essential as a map from $(\hat E_\d, \partial \hat E_\d) \to (Q_{\d_1}, Q_{\d_1} \backslash (\tilde P_{\d,L}^\b \backslash \partial \tilde P_{\d,L}^\b))$.
	
	For $t \in [0,1]$, and $(\eta_{\sn}, \eta_{\so}, \tilde \s_{\sn}, \tilde \s_{\so}) \in \hat E_{\d_1}$, define $F^t(\eta_{\sn}, \eta_{\so}, \tilde \s_{\sn}, \tilde \s_{\so}) = (\tilde \eta_{\sn}^t, \eta_{\so})$ where $\tilde \eta_{\sn}^t = \tilde \eta(\eta_{\sn}, b, \tilde b, t\s + (1-t)\s^*,  t\tilde \s_{\so} + (1-t)\tilde \s_{\so}^*)$.
	$F^t$ defines a homotopy between $\psi \circ \proj^g \circ H^{-1}$ and $\tilde \proj$, which is the projection map from $\tilde X_\e$ to $\tilde P_\e$.  Hence, $\tilde \proj$ is essential as a map from $(\hat E_{\d_1}, \partial \hat E_{\d_1})$ to $(Q_{\d_1}, Q_{\d_1} \backslash (\tilde P_{\d,L}^\b \backslash \partial \tilde P_{\d,L}^\b))$.  This implies that the projection map from $\tilde X_{\d, L}^{1,\b}$ to $\tilde P_{\d, L}^\b$ is essential.
\end{proof}

\section{Proof of Proposition 7.1}

We restate Proposition 7.1 and prove it.

\begin{proposition}
	Each $\bbP(R)$ is a boundaryless manifold and belongs to the closure of one that is a $(d-1)$-dimensional manifold; moreover, if $\bbP(R)$ is $(d-2)$-dimensional, it belongs to the closure of exactly two such manifolds. 
\end{proposition}

\begin{proof}
	We will first study, for a fixed $p \in \partial P$, the structure of the set of $\l$ such that $(p, \l) \in \bbP$. Let $A^0$ be the set of all last actions across all players that are used with zero probability under $p$.  As $p$ belongs to $\partial P$, it  cannot be completely mixed, and therefore $A^0$ is nonempty.  Given $\theta \in \Re_{++}^{A^0}$, for each $n$ and $a_n \in L_n$, let $\tilde p_n(a_n; \theta)$ equal $p_n(a_n)$ if $p_n(a_n) > 0$; otherwise let it equal $\theta_n(a_n)$. We define $\tilde p_{-n}(z; \theta)$ similarly, using $\tilde p_m$ rather than $p_m$ for each $m \neq n$. Let $\tilde \l(\theta)$ be the likelihood ratios defined as follows. For any pair $(a_m, a_n) \in \Lmath$:  
	\[
	\tilde \l(a_m, a_n; \theta) = \frac{\tilde p_m(a_m; \theta)}{\tilde p_n(a_n; \theta)};
	\] 
	and for each $n$ and $(z_0, z_1) \in \Z_{-n}$:
	\[
	\tilde \l_{-n}(z_0, z_1; \theta) = \frac{\tilde p_{-n}(z_0; \theta)}{\tilde p_{-n}(z_1; \theta)}.
	\]
	From $\tilde \l$ we can compute $\l(a_m, a_n; \theta)$ as usual: $\l(a_m, a_n; \theta) = \tilde \l(a_m, a_n; \theta){(1+\tilde \l(a_m, a_n; \theta))}^{-1}$ and similarly for $\l_{-n}(z_0, z_1; \theta)$. 
	
	We claim that $(p, \l)$ belongs to $\bbP$ iff there exists a sequence $\theta^k \to 0$ such that $\l(\cdot; \theta) \to \l$. Indeed, to prove the necessity of the condition, suppose $(p, \l)$ belongs to $\bbP$. There exists a sequence $(p^k, \l^k)$ converging to $(p, \l)$ with $p^k \in P^0$ for all $k$. For each $k$ and $a^0 \in A^0$, let $\theta_{a^0}^k = p_{a^0}^k$. Clearly $\theta^k \to 0$. Using $p$ along with $\theta^k$ to generate  sequences $\tilde p_n(\cdot; \theta^k)$ and $\tilde p_{-n}(\cdot; \theta^k)$ for each $n$,   we see that the corresponding sequence $\l(\cdot; \theta^k)$ converges to $\l$.  To prove sufficiency, given a sequence $\theta^k \to 0$ generating a convergent sequence of beliefs $\l(\cdot, \cdot; \theta^k) \to \l$, define for each $k$ a completely mixed behavioral strategy profile $b^k$  using the likelihood ratios $\l(\cdot, \cdot; \theta^k)$. The sequence of $b^k$'s induces a corresponding sequence $p^k$ of enabling strategies, which converges to $p$. The beliefs associated with $p^k$ are $\l(\cdot, \cdot; \theta^k)$, which completes the proof of the claim.
	
	To compute the set of $(p, \l) \in \partial \bbP$ for this given $p$, we therefore need to study the limits of the sequences $\tilde \l(\cdot, \cdot; \theta^k)$ as $\theta^k \to 0$. Taking the log of $\tilde \l(\cdot,\cdot;\theta)$ in the formula defining it, and using a change of variable $\vartheta_{a^0} \equiv \ln(\theta_{a^0})$ for each $a^0 \in A^0$, we get that the log-likelihood ratios are affine in the variable $\vartheta$.  Call this function $L^p(\vartheta)$.  $L^p$ is an affine function from $\Re^{A^0}$ to $\Re^{\Lmath \cup \Z}$.  For each $\vartheta$, it assigns to each $(a_m, a_n)$  (resp.~each  $n$ and $(z_0, z_1)$ in  $\Z_{-n}$) a log-likelihood ratio $L_{a_m, a_n}^p(\vartheta)$ (resp.~$L_{-n, z_0,z_1}^p(\vartheta)$).  
	
	Let $E^p \equiv L^p(\Re^{A^0})$ and let ${[-\infty, \infty]}^{\Lmath \cup \Z}$ be obtained by the two-point compactification $[-\infty, \infty]$ of each factor. Let $\bar E^p$ be the closure of $E^p$ in ${[-\infty, \infty]}^{\Lmath \cup \Z}$. Let $\bar E_\infty^p$ be the set of $\eta \in \bar E^p$ that are limits of sequences $L^p(\vartheta^k)$ with $\vartheta_a^k \to -\infty$ for each $a \in A^0$.  $\bar E_\infty^p$ is homeomorphic to the set of $(p, \l) \in \partial \bbP$. 
	
	Let $\bar \Theta(A^0)$ be the unit simplex in $\Re^{A^0}$, i.e.,~the set of $\vartheta \in \Re_{+}^{A^0}$ such that $\sum_{a^0 \in A^0} \vartheta_{a^0} = 1$, and let $\Theta(A^0)$ be its relative interior.  Denote by $\bar L^p$ the linear function on $\Re^{A^0}$ that is the derivative of the affine function $L^p$. For a triplet  $R \equiv (R^0, R^1, R^2)$ that is a partition of $\Lmath \cup \Z$ into three nonempty subsets, let $\Theta(R; A^0)$ be the set of $\vartheta \in \Theta(A^0)$ such that $\bar L_{r}^p(\vartheta)$ is, respectively, zero, negative or positive, according as $r$ belongs to $R^0$, $R^1$ and $R^2$.  Obviously the sets $\Theta(R; A^0)$ obtained by considering all possible triplets $R$ partition $\Theta(A^0)$. 
	
	For each triplet $R$ for which $\Theta(R; A^0)$ is nonempty, let $E^p(R)$ be the set $L^p_{R^0}(\Re^{A^0}) \times {\{\, \infty \, \}}^{R^1} \times {\{\, -\infty \, \}}^{R^2}$. We claim that the sets $E^p(R)$ partition $\bar E_\infty^p$. The sets $E^p(R)$ are obviously disjoint.  Thus to prove the claim we have to show that $\bar E_\infty^p$ is the union of these sets.  First take $\eta$ in $\bar E_\infty^p$. There is a unique $R$ such that $\eta_r$ is, respectively, finite, $+\infty$ or $-\infty$, according as $r$ belongs to $R^0$, $R^1$ or $R^2$.  Obviously there is a $\vartheta$ such that $L_{R^0}(\vartheta)$ is the projection of $\eta$ to $R^0$ and thus $\eta$ belongs to $E^p(R)$.    Going the other way, given $\eta \in E^p(R)$ for some $R$, there exists $\vartheta$ such that $\eta_{R^0}  = L_{R^0}^p(\vartheta)$. Let $\vartheta'$ be a point in $\Theta(R; A^0)$. Obviously $\eta$ is the limit of $L^p(\vartheta + \a \vartheta')$ as $\a \to -\infty$. Thus, $\bar E^p_\infty$  contains the union of the sets $E^p(R)$ and our claim is proved.
	
	To analyze the structure of the sets $E^p(R)$, we first analyze their duals $\Theta(R; A^0)$. For each $R$, letting $d^0(R; A^0)$ be the dimension of $L_{R^0}^p(\Re^{A^0})$, we see that  $\Theta(R; A^0)$ is a linear manifold of dimension $|A^0| - 1  - d^0(R; A^0)$. The closure $\bar \Theta(R; A^0)$ of each $\Theta(R; A^0)$ in $\bar \Theta(A^0)$  is a polytope and the relative interior of each face of it that does not belong to the boundary of $\bar \Theta(A^0)$ is a set of the form $\Theta(R'; A^0)$.   We claim that $\bar \Theta(R; A^0)$ has at least one vertex in $\Theta(A^0)$ (which then corresponds to some $\Theta(R'; A^0)$). Indeed, to show this, choose $R'$ such that $\Theta(R'; A^0)$ is a face of $\Theta(R; A^0)$ with the smallest possible dimension among faces that are not contained in the boundary of $\bar \Theta(A^0)$. We claim that $\Theta(R'; A^0)$ is a singleton. Otherwise, there is a line segment in  $\bar \Theta(R'; A^0)$ and a pair $(a_m^0, a_n^0)$ of last actions (possibly of the same player) such that $\vartheta_{a_m^0} - \vartheta_{a_n^0} = \bar L_{a_m, a_n}^p(\vartheta)$ is positive at one endpoint of the line segment and negative at the other, which is impossible by the definition of the set $\Theta(R'; A^0)$.  Thus $\bar \Theta(R; A^0)$ has a vertex in $\Theta(A^0)$. If $\Theta(R; A^0)$ is 1-dimensional and has a vertex $\vartheta^0$ in $\bar \Theta(A^0) \backslash \Theta(A^0)$, then $\{\, \vartheta^0 \,\}$ is the singleton $\Theta(\hat R; \hat A^0)$ in $\Theta(\hat A^0)$ for some $\hat R$ and $\hat A^0 \subsetneq A^0$.

	The properties of $\Theta(R; A^0)$ described above yield the following properties for the sets $E^p(R)$. Each set $E^p(R)$ is a linear manifold of dimension $d^0(R; A^0)$. And its closure is a union of manifolds of the form $E^p(R')$ where $R'$ is such that $\bar \Theta(R'; A^0)$ contains $\Theta(R; A^0)$ in its closure.  Every $E^p(R)$ belongs to the closure of a manifold of dimension $|A^0| - 1$. If $d^0(R; A^0) = |A^0|-2$, then it belongs to the closure of either one or two such manifolds $E^p(R')$ of dimension $|A^0|-1$, depending on whether the set $\Theta(R; A^0)$ has one or two endpoints in $\Theta(A^0)$. When the set $\Theta(R; A^0)$ has one endpoint $\vartheta^0$ in $\bar \Theta(A^0) \backslash \Theta(A^0)$, there exist $\hat R$ and $\hat A^0 \subsetneq A^0$ such that $\vartheta^0$ is the set  $\Theta(\hat R; \hat A^0)$. It is easily checked that for each  $\eta \in E^p(R)$, there then  exists a sequence $p^k \to p$ with $p_{a_n}^k > 0$ iff $a_n \notin \hat A^0$, and a  corresponding sequence of $\eta^k$ in the $(|\hat A^0| -1)$-dimensional manifold $E^{p^k}(\hat R)$, converging to $\eta$.
	All the qualitative results about the structure of the set of points $(p, \l)$ for a fixed $p$ that we have described hold uniformly across all profiles with the same support as $p$.
	
	We are now ready to prove the stated properties of the sets $\bbP(R)$.  Take $\bbP(R)$ in $\partial \bbP$.  There exists a unique proper face $P^1$ of $P$ whose interior $P^{1,0}$ contains $\proj(\bbP(R))$.  Let $A^0$ be the set of unused actions in $P^1$.  $\bbP(R)$ is then $\cup_{p \in P^{1,0}} \bbP(R, p)$, where for each $p$,  $\bbP(R, p) = \proj^{-1}(p) \cap \bbP(R)$, and it is homeomorphic to $\cup_{p \in P^{1,0}} (\{\,  p \, \} \times E^p(R))$  and thus a manifold of dimension  $d^0(R; A^0) + \dim(P^1)$. As $d^0(R; A^0) \le |A^0| - 1$ and $\dim(P^1) = d - |A^0|$, $\dim(\bbP(R)) \le d-1$. If $d^0(R; A^0) = |A^0| -1$, then $\dim(\bbP(R)) = d -1$.  If $\dim(\bbP(R)) < d-1$, then it belongs to the closure of a manifold $\bbP(R')$ of dimension $d-1$ as the polytope $\bar \Theta(R; A^0)$ has a vertex  in $\Theta(A^0)$ that is of the form $\Theta(R'; A^0)$. If $\bbP(R)$ has dimension $d-2$,  then it either belongs to the closure of  two such manifolds whose projection to $P$ is $P^{1,0}$,  or it belongs to the closure of just one such manifold, in which case it also belongs to the closure of one projecting to a face $P^2$ that is either a  face of $P^1$ or has $P^1$ as a face. 
\end{proof}

\section{Proof of Proposition 8.1}

This appendix gives a proof of the following proposition characterizing stability in generic extensive-form games.

\begin{proposition}\label{prop Q}
	If $u$ is generic, there exists a finite partition $(\calQ^1, \partial \calQ^1), \ldots, (\calQ^k, \partial \calQ^k)$ of $(\calQ, \partial \Q)$ such that:
	\begin{enumerate}
		\item for each $i$, $(\calQ^i, \partial Q^i)$ is an orientable $d$-pseudomanifold;
		\item $\S^*$ is a stable set of $\G(u)$ iff there exists $i$ such that: 
		\begin{enumerate}
			\item $\pi(\S^*) = \{\, q \in P \mid  \exists \, (p, (q, \l)) \in \calQ^i \, \}$;
			\item $\proj: (\calQ^i, \partial \calQ^i) \to (P, \partial P)$ is essential.
		\end{enumerate}	
	\end{enumerate}
\end{proposition}

We prove Proposition \ref{prop Q} via  a sequence of lemmas. For each $u \in  \Re^{NZ}$, let $E(u)$ be the closure of the set of $(\e, p, (q, \l)) \in (0, 1) \times (P \backslash \partial P) \times \bbP$ such that for all $n$ and $a_n \in L_n$, $q_{n,a_n} \ge \e  p_{n,a_n}$, with an equality if $a_n$ is not a best reply to $q$. For simplicity, we write $E$ for $E(u)$ when the utility function we are talking about is understood. For each $0 < \eta \le 1$ (resp.~$0 \le \eta \le 1$) and each subset $X$ of $E$, let $X_{[0, \eta]}$ (resp.~$X_\eta$) be the set of $(\e, p, (q, \l))$ in $X$ such that $\e \le \eta$ (resp.~$\e = \eta$) and let $\partial X_{[0, \eta]}$ (resp.~$\partial X_{\eta}$) be the set of $(\e, p, (q, \l))$  in $X_{[0, \eta]}$ (resp.~$X_\eta$) such that $(\e, p) \in \partial ([0, \eta] \times  P)$ (resp.~$p \in \partial P$).   

\begin{lemma}\label{lem E partition}
	If $u$ is generic, there exist $0 < \eta_0 < 1$ and a partition of $(E_{[0, \eta_0]}, \partial E_{[0, \eta_0]})$ into finitely many pairs $(X_{[0, \eta_0]}^i, \partial X_{[0, \eta_0]}^i)$, $i = 1, \ldots, k$ such that for all $i$ and $0 < \eta \le \eta_0$:
	\begin{enumerate}
		\item  $(X_{[0, \eta]}^i, \partial X_{[0, \eta]}^i)$ is an orientable $(d+1)$-pseudomanifold and $(X_\eta^i, \partial X_\eta^i)$ is an orientable $d$-pseudomanifold;
		\item $\proj:(X_{[0, \eta]}^i, \partial X_{[0, \eta]}^i) \to ([0, \eta] \times P, \partial([0, \eta] \times P))$ is essential in cohomology iff $\proj:(X_\eta^i, \partial X_\eta^i) \to (P, \partial P)$ is essential in cohomology.   
	\end{enumerate}
\end{lemma}

\begin{proof}
	Lemma E.3 of \cite{GW2008} shows the existence of $\eta_1$ and subsets $X^1, \ldots, X^k$ of $E$ such that: $\cup_i (X_{[0, \eta_1]}^i, \partial X_{[0, \eta_1]}^i) = (E_{[0, \eta_1]}, \partial E_{[0, \eta_1]})$; and for $i \neq j$, $(X_{[0, \eta_1]}^i \cap  X_{[0, \eta_1]}^j) \subseteq \partial E_1$. We will show later in the proof that  these sets are  disjoint.\footnote{In \cite{GW2008} the set $E$ consisted of tuples $(\e, p, q)$ without the induced beliefs; hence this stronger property does not hold in their setup.} The paragraph following the proof of Lemma E.3 (loc.~cit.) shows that for $0 < \eta \le \eta_1$, each  $X_{[0, \eta]}^i$ is a $(d+1)$-homology manifold and that $(X_{[0, \eta]}, \partial X_{{[0, \eta]}})$ is an orientable pseudomanifold as well, which proves property (1) for these sets.

	We now show that there exists $0 < \eta_0 \le \eta_1$ such that for all $i$ and $0 < \eta \le \eta_0$, $X_\eta^i$ is an orientable pseudomanifold and that property (2) holds. To show point (1) for the sets $X_\eta^i$ for small $\eta > 0$, we argue as follows. 
	For each $i$, let $\phi^i:  X_{[0, \eta_1]}^i \to [0, \eta_1]$ be the map that sends $(\e,  p, (q, \l))$ to $\e$. By the Generic Local Triviality Theorem (Theorem 9.3.2 of \cite{BCR1998}) there exists $0 < \eta_0 \le \eta_1$, and for each $i$, a semialgebraic pair $(F^i, \partial F^i)$, and a homeomorphism  $h^i: (0, \eta_0] \times (F^i, \partial F^i) \to (X_{[0, \eta_0]}^i \backslash X_0^i, \partial X_{[0, \eta_0]}^i \backslash \partial X_0^i)$ such that $\phi^i\circ h^i (\eta, f^i) = \eta$ for each $(\eta, f^i) \in (0, \eta_0] \times F^i$.  Looking ahead to our proof of point (2), the homeomorphism $h^i$ can be chosen in such a way that for each face $P_0$ of $P$, there is a subset $F_0^i$ of $F^i$ such that $h^i((0, \eta_0] \times F_0^i) = \proj^{-1}((0, \eta_0] \times P_0) \cap X_{[0, \eta_0]}^i$. As $X_{[0, \eta_0]}^i$ is a (d+1)-homology manifold, $[0, \eta_0] \times F^i$ is a $(d+1)$-homology manifold as well. Hence, $F^i$ is a $d$-homology manifold and $(F^i, \partial F^i)$ is an orientable $d$-pseudomanifold. $(X_\eta^i, \partial X_\eta^i)$ is homeomorphic to $(F^i, \partial F^i)$ by the restriction of $h^i$ to $\{\, \eta \, \} \times (F^i, \partial F^i)$, which completes the proof of point (1).
	
	To prove point (2), now fix $i$ and  $0 < \eta \le \eta_0$. Let $\bar P_\eta$ be the quotient space of $[0, \eta] \times P$ obtained by collapsing $\{\, 0 \, \} \times P$ to a point $0_P$ and let $\partial \bar P_\eta$ be the image of the boundary under this quotient map.  Likewise, let $\bar X_{[0, \eta]}^i$ be the quotient space obtained by collapsing $X_0^i$ to a point $0_{X^i}$ and denote by $\partial \bar X_{[0, \eta]}^i$ its boundary.  Let $\bar \proj$ be the projection from $\bar X_{[0, \eta]}^i$ to $\bar P_\eta$.  The essentiality of $\proj$ for $(X_{[0, \eta]}^i, \partial X_{[0, \eta]}^i)$ is equivalent to that of $\bar \proj$ for the corresponding set $(\bar X_{[0, \eta]}^i, \partial \bar X_{[0, \eta]}^i)$.

	The map $H_\eta^i: (0, \eta] \times X_\eta^i \to \bar X_{[0, \eta]}^i \backslash 0_{X^i}$ that sends $(\e, (\eta, p, (q, \l)))$ to $h^i(\e, {(h^i)}^{-1}(\eta, p, (q, \l)))$ is a homeomorphism that extends to a map to their closures.  Hence $H_\eta^i$ induces an isomorphism of the cohomology groups.  Define $\tilde \proj: [0, \eta] \times X_\eta^i \to \bar P_\eta$ by $\tilde \proj = \bar \proj \circ H^i$. Then $\bar \proj$ is essential iff $\tilde \proj$ is, and thus $\proj:(X_{[0, \eta]}^i, \partial X_{[0, \eta]}^i) \to ([0, \eta] \times P, \partial ([0, \eta] \times P))$  is essential iff $\tilde \proj$ is.  Let $\hat \proj: [0, \eta] \times X_\eta^i \to  \bar P_\eta$ be the map given by $\hat \proj(\e, (\eta, p, (q, \l))) = (\e\eta, p)$ if $\e > 0$ and $0_P$ otherwise. Since for each face $P_0$ of $P$ there exists a subset $F_0^i$ of $F^i$ such that $h^i((0, \eta_1] \times F_0^i) = \proj^{-1}((0, \eta_1] \times P_0) \cap X_{[0, \eta_1]}^i$, $\hat \proj$ is linearly homotopic to $\tilde \proj$.  Hence the latter map is essential iff the former is.  Thus, $\proj:(X_{[0, \eta]}^i, \partial X_{[0, \eta]}^i) \to ([0, \eta] \times P, \partial ([0, \eta] \times P))$ is essential iff $\hat \proj$ is. Obviously, $\hat \proj$ is essential iff $\proj: (X_\eta^i, \partial X_\eta^i) \to (P, \partial P)$ is essential, which proves point (2) of the lemma.
	
	To complete the proof, it remains to show that the sets $X_{[0, \eta_0]}^i$ are disjoint. As we saw earlier,  the sets $X_{[0, \eta_0]}^i \backslash \partial X_{[0, \eta_0]}^i$ are disjoint. Take then $(\eta, p, (q, \l)) \in \partial X_{[0, \eta_0]}^i$ for some $i$.  If it belongs to $X_{[0, \eta_0]}^j$ for some $j \neq i$, then we can find curves $(\eta^i(t), p^i(t), (q^i(t), \l^i(t))$ and $(\eta^j(t), p^j(t), (q^j(t), \l^j(t))$, $t \in [0, 1]$, that lie in $X_{[0, \eta_0]}^i$  and $X_{[0, \eta_0]}^j$, respectively, and such that: (a) their values at $t = 0$ equal $(\eta, p, (q, \l))$;  (b) $p^i(t), p^j(t) \notin \partial E_1$ for all $t$; (c) $\eta^i(t) = \eta^j(t)$ for all $t$.  For $\a \in [0, 1]$, let $p(t; \a)$ (resp.~$(q(t; \a)$) be the point in $P$ that is closest to  ${({(p_{a_n}^i(t))}^{(1-\a)}{(p_{a_n}^j(t))}^\a)}_{a_n \in A}$ (resp.~${({(q_{a_n}^i(t))}^{(1-\a)}{(q_{a_n}^j(t))}^\a)}_{a_n \in A}$) among all $p'$ for which $p_{a_n}^{'}$ equals ${(p_{a_n}^i(t))}^{(1-\a)}{(p_{a_n}^{j}(t))}^\a$ (resp.~${(q_{a_n}^i(t))}^{(1-\a)}{(q_{a_n}^{j}(t))}^\a$) for each $n, a_n$ such that $q_{a_n} = \eta p_{n, a_n}$.  For small $t> 0$, $p(t; \a)$ and $q(t; \a)$ are well-defined for all $\a \in [0,1]$. Moreover, the limit as $t \to 0$ of the sequence $(\eta(t), p(t; \a), (q(t;\a), \l(t;\a)))$ is $(\eta, p, (q, \l))$ for all $\a$.  For each $\a$ and small $t$, there is now a small perturbation  of payoffs to $u(t; \a)$ that puts this point in $E(u(t; \a))$. Moreover, $u(t; \a)$ can be chosen continuously in $(t, \a)$.  In \cite{GW2008}, Lemma E.3 was proved using a generic local triviality argument: in a neighborhood $U$ of $u$, there exist a fiber $(F, \partial F)$ and a homeomorphism $h: U \times  (F, \partial F) \to (E(U), \partial E(U))$ such that $\proj \circ h$ is the projection onto the first factor, and the fiber $(F, \partial F)$ is the union of the fibers $(F^1, \ldots, F^k)$.   Therefore, for small $t$ and all $\a$, $(\eta(t;\a), p(t; \a), (q(t;\a), \l(t; \a))$ belongs to the same component of $(F, \partial F)$, which contradicts the assumption that the two curves are in two different sets.
\end{proof}

For each $i = 1, \ldots, k$, let $\calQ^i$ be the set of $(p, (q, \l))$ such that $(0, p, (q, \l)) \in X_0^i$ and let $\partial \calQ^i$ be the subset of $\calQ^i$ consisting of points $(p, (q, \l))$ such that $p \in \partial P$. Since the sets $(X_{[0, \eta_0]}^1, \partial X_{[0, \eta_0]}^1), \ldots, (X_{[0, \eta_0]}^k, \partial X_{[0, \eta_0]}^k)$ partition $E_{[0, \eta_0]}$, by Lemma \ref{lem E partition} the $\calQ^i$'s are  connected sets that are pairwise disjoint.

\begin{lemma}\label{lem Q partition}	
	If $u$ is generic, then  $(\calQ, \partial \calQ) = \cup_{i = 1}^k (\calQ^i, \partial \calQ^i)$ and for each $i$, $(\calQ^i, \partial \calQ^i)$ is a $d$-pseudomanifold with boundary.
\end{lemma}	

\begin{proof}
	Let $\calQ^*(u)$ be the  set of $(p, (q, \l)) \in (P\backslash \partial P) \times (\bbP \backslash \partial \bbP)$ such that, for every pair $(a_m, a_n)$, if $a_n$ is not a best reply to $(q, \l)$ in $\G(u)$ then $\l(a_m, a_n) \ge \frac{p_m(a_m)}{p_m(a_m) + p_n(a_n)}$.   Let $\calQ_0(u)$ be the closure of the set of points $(p, (q, \l)) \in (P \backslash \partial P) \times  \bbP$ such that there is a sequence $(p^k, (q^k, \l^k))$ in $\calQ^*(u)$ converging to it against which $q$ is a best reply. We claim that $\calQ_0(u)$ is the set of  $(p, (q, \l))$ such that $(0, p, (q, \l))$ belongs to $E_0(u)$. To show that $\calQ_0(u)$ is a subset of $E_0(u)$, it is sufficient to show that $(p, (q, \l)) \in \calQ_0(u)$ with $p \notin \partial P$ belongs to $E_0(u)$, as $\calQ_0(u)$ is the closure of such points and $E_0(u)$ is closed.  Take such a point $(p, (q, \l))$. If $q \in P \backslash \partial P$, then $q$ is a best reply to itself, i.e.,~a completely mixed equilibrium.  Therefore, $(\e, p, (q, \l))$ belongs to $E_\e$ for all small $\e$ and hence $(0, p, (q, \l))$ belongs to $E_0(u)$.  Suppose $q \in \partial P$. Then $q \neq p$ and for all large $k$, $q^k \neq p^k$. For each such $k$, as $q^k$ and $p^k$ are in $P \backslash \partial P$, there exists a unique $\e^k \in (0, 1)$ such that $\frac{1}{1-\e^k}(q^k - \e^kp^k)$ belongs to $\partial P$.  Then $(\e^k, p^k, (q^k, \l^k))$ belongs to $E(u)$. As $q^k$ converges to $q$, $\e^k \to 0$.  Thus we have shown that $\calQ_0(u)$ contains only points $(p, (q, \l))$ such that $(0, p, (q, \l))$ belongs to $E_0(u)$.  To prove the reverse inclusion, if $(\e^k, p^k, (q^k, \l^k))$ is a sequence in $E(u)$ converging to $(0, p, (q, \l))$ with $ p \in \bbP \backslash \partial \bbP$, then obviously $q$ is a best reply to $(q^k, \l^k)$ as $q^k$ puts $\e^k$ weight on suboptimal last actions. Therefore, $(p, (q, \l)) \in \calQ_0(u)$, implying as well that every point in $E_0(u)$ is also in $\calQ_0(u)$, as both sets are closed.
	
	Obviously, $\calQ_0(u)$ is contained in $\calQ(u)$.  Thus, to complete the proof, there remains to show the reverse inclusion and that $\calQ(u)$ is a finite union of pseudomanifolds.  As we saw earlier, $\partial \bbP$ is expressible as the union of manifolds $\bbP(R)$, where the triplets $R = (R^0, R^1, R^2)$ partition $\Lmath \cup \Z$.  For $R = (\Lmath \cup \Z, \emptyset, \emptyset)$, we have the manifold $\bbP(R) \equiv \bbP \backslash \partial  \bbP$. Thus $\bbP$ is a finite disjoint union of semialgebraic manifolds, $\bbP(R)$, where $R$ ranges over all triplets that either partition $\Lmath \cup \Z$ or are the triplet with $R^0 = \Lmath \cup \Z$. Every manifold belongs to the relative interior of a face of $P$, and letting $d(R)$ be the dimension of $\bbP(R)$,   we have that if $d(R) < d(R')$, $\bbP(R)$ is either contained in the closure of $\bbP(R')$ or disjoint from it.  Observe that this decomposition of $\bbP$ depends only on the game tree $\G$ and not on the payoffs.
	
	Fix a  manifold $\bbP(R)$; and fix a face $P^j$ of $P$ such that if $\bbP(R) \neq \bbP \backslash \partial \bbP$, then the face that contains $\bbP(R)$ in its interior is a subset of $P^j$. Let $\M(P^j, \bbP(R))$ be the set of $(u, (p^*, \l)) \in \Re^{NZ} \times \bbP(R)$ such that the strategies in $P^j$ are the best replies to $(p^*, \l)$.
	Clearly, $\M(P^j, \bbP(R))$ is a semialgebraic manifold of dimension $NZ + d(R) - \dim(P^j)$. Suppose $(u, (p^*, \l)) \in  \M(P^j, \bbP(R))$ and $\bbP(R')$ contains $\bbP(R)$ in its closure with $\bbP(R')$ also belonging to the face $P^j$ if it does not equal $\bbP \backslash \partial \bbP$.  If we have a sequence $(p^k, \l^k))$ converging to $(p^*, \l))$ such that $(p^k, \l^k)$ belongs to $\bbP(R')$ for all $k$, then, as in \cite{KW1982}, we can construct a  sequence $u^k \to u$ such that $(u^k, (p^k, \l^k))$ belongs to $\M(P^j, \bbP(R'))$ for all $k$.  Thus, $\M(P^j, \bbP(R))$ belongs to the closure of $\M(P^j, \bbP(R'))$ if $\bbP(R)$ belongs to the closure of $\bbP(R')$. Similarly, if $P^{j'}$ is a face of $P^j$ that also contains $\bbP(R)$, then $\M(P^{j'}, \bbP(R))$ contains $\M(P^j, \bbP(R))$ in its closure.

	For each $u \in \Re^{NZ}$, if we let $M(P^j, \bbP(R), u)$ be the set of $(p^*, \l)$ such that $(u, (p^*, \l)) \in \M(P^j, \bbP(R))$, then by Generic Local Triviality (loc.~cit.) for all $u$ outside of a closed semialgebraic subset $U(P^j, \bbP(R))$, $M(P^j, \bbP(R), u)$ is a semialgebraic manifold of dimension $d(R) - \dim(P^j)$, and the sets $M(P^j, \bbP(R), u)$ and $M(P^j, \bbP(R), u')$ are homeomorphic for $u$ and $u'$ belonging to the same connected component of $\Re^{NZ} \backslash U(P^j, \bbP(R))$. If $u$ does not belong to $U(P^j, \bbP(R))$ for any pair $(P^j, \bbP(R))$, then: (1) each $M(P^j, \bbP(R), u)$ is a manifold of dimension $d(R) - \dim(P^j)$; (2) if $M(P^j, \bbP(R), u)$ is nonempty for some $(P^j, \bbP(R))$, then for each $\bbP(R')$ (respectively $P^{j'}$) whose closure contains $\bbP(R)$ (resp.~that is a face of $P^j$)  the closure of $M(P^j, \bbP(R'), u)$  (resp.~$M(P^{j'}, \bbP(R), u)$) contains $M(P^j, \bbP(R), u)$.      
	
	Fix a manifold $\bbP(R)$ and a nonempty face $P^j$ of $P$, and let $\tilde P^j$ be any (possibly empty) face of $P^j$.  Let $\calQ(P^j, \bbP(R), \tilde P^j, u)$ be the set of $(p, (p^*, \l)) \in (P \backslash \partial P) \times \bbP(R)$ such that: (1) $(p^*, \l) \in M(P^j, \bbP(R), u)$; (2) if $a_n$ is not a best reply, then $\l(a_m, a_n) > \frac{p_m(a_m)}{p_m(a_m) + p_n(a_n)}$ iff $a_m$  is a best reply and has zero probability on the face $\tilde P^j$.   Obviously, $\calQ(P^j, \bbP(R), \tilde P^j, u)$ is nonempty only if $\tilde P^j$ is not a face of the face of $P$ that contains $\bbP(R)$ in its relative interior. 
	
	$\calQ(u)$ is the closure of the union of the sets $\calQ(P^j, \bbP(R), \tilde P^j, u)$ as we range over all triplets $(P^j, \bbP(R), \tilde P^j)$ with $\bbP(R) \neq P \backslash \partial P$. By Property (2) of the last paragraph, if $u$ does not belong to $U(P^j, \bbP(R))$ for any $(P^j, \bbP(R))$, then if $(p, (p^*, \l))$ belongs to $\calQ(P^j, \bbP(R), \tilde P^j, u)$ with $\bbP(R) \neq P \backslash \partial P$, there is a sequence $(p^k, (q^k, \l^k))$ converging to it that belongs to $\calQ^*(u)$.  Hence $\calQ(u)$ is contained in  $\calQ_0(u)$.  To finish the proof we have to show that $\calQ(u)$ is a pseudomanifold, for which it is sufficient to show that each of the sets $\calQ(P^j, \bbP(R), \tilde P^j, u)$  belongs to the closure of one with dimension $d$; and that if the dimension of one is $d-1$, it belongs to the closure of exactly two such $d$-manifolds. 
	
	Fix $P^j, \bbP(R), \tilde P^j$. As we saw, $M(P^j, \bbP(R), u)$ is $(d(R) - \dim(P^j))$-dimensional. For each $(p^*, \l) \in M(P^j, \bbP(R), u)$, the set of $p$ such that $(p, (p^*, \l))$ belongs to $\calQ(P^j, \bbP(R), \tilde P^j, u)$ is a manifold of dimension  $\dim(P^j) - \dim(\tilde P^j)$, with the convention that $\dim(\tilde P^j)$ is $-1$ (resp.~0) if $\tilde P^j$ is an empty set and $\dim(P^j) < d$ (resp.~$\dim(P^j) = d)$.  Thus, $\calQ(P^j, \bbP(R), \tilde P^j, u)$ is a manifold of dimension $d(R)  - \dim(\tilde P^j)$. If $d(R) = d$, or $d(R) =  d-1$ and $\tilde P^j$ is empty, then the dimension of $\calQ(P^j, \bbP(R), \tilde P^j)$ is $d$ and we are done.  If $d(R) < d-1$ or $\tilde P^j$ is nonempty,  there is a $(d-1)$-dimensional manifold $\bbP(\hat R)$ that is  contained in the interior of the same face, call it $\hat P^j$, that contains $\bbP(R)$ and that has $\bbP(R)$ in its closure. $\calQ(P^j, \bbP(R), \tilde P^j)$ is then contained in the closure of the $d$-dimensional manifold $\calQ(\hat P^j, \bbP(\hat R), \emptyset)$.  Finally, if $\calQ(P^j, \bbP(R), \tilde P^j, u)$ is $(d-1)$-dimensional, to show that it belongs to the closure of two $d$-manifolds, observe first that, then, either $d(R) = d-1$ and  $\tilde P^j$ is a singleton or $d(R) = d-2$ and $\tilde P^j$ is empty.  As we saw, there is a $\bbP(\hat R)$ in the same face $\hat P^j$ that contains $\calQ(P^j, \bbP(R), \tilde P^j)$ in its closure. If $\tilde P^j$ is a singleton, there is also a $(d-1)$-dimensional $\bbP(\check{R})$ in the face $\check{P}^j$ spanned by $\hat P^j$ and $\tilde P^j$ such that $\calQ(P^j, \bbP(R), \tilde P^j)$ is in the closure of $\calQ(\check{P}^j, \bbP(\check{R}), \emptyset)$; if $\tilde P^j$ is empty, then there is another $(d-1)$-dimensional manifold $\bbP(\check{R})$ in the relative interior of a face $\hat P^j$  that has $\bbP(R)$ as a face, and then $\calQ(P^j, \bbP(R), \tilde P^j)$ is in the closure of $\calQ(\check{P}^j, \bbP(\check{R}), \emptyset)$, which completes the proof. 
\end{proof}

\begin{lemma}\label{lem Q essentiality}
	For each $i$,  $H_d(\calQ^i, \partial \calQ^i) \approx \mathbb{Z}$; and for all small $\eta > 0$, $\proj: (X_\eta^i, \partial X_\eta^i) \to (P, \partial P)$ is essential iff $\proj: (\calQ^i, \partial \calQ^i) \to (P, \partial P)$ is.
\end{lemma}

\begin{proof}
	It is sufficient to prove the lemma by replacing $\calQ^i$ with $X_0^i$.
	For $0 < \eta \le 1$, let $\bar \partial X_{[0, \eta]}^i$ be the set of $(\e, p, (q, \l)) \in X_{[0, \eta]}^i$ such that $p \in \partial P$; and for $\e = 0, \eta$, let $i_\e:(X_\e^i, \partial X_\e^i) \to (X_{[0, \eta]}^i, \bar \partial X_{[0, \eta]}^i)$ be the inclusion map.  We can choose $\eta$ small enough such that $i_0$ induces an isomorphism in homology. Clearly, the inclusion $i_\eta$ induces a monomorphism in homology. By Lemma \ref{lem Q partition}, $\calQ^i$ is a pseudomanifold, making $X_0^i$ one as well.  As  $i_{\eta, *}(H_d(X_\eta^i, \partial X_\eta^i)) \approx \mathbb{Z}$, we must have that $H_d(X_0^i, \partial X_0^i) \approx \mathbb{Z}$, thus proving the first statement. 
	
	As for the second statement, we also now have that $i_{\eta, *}: H_d(X_\eta^i, \partial X_\eta^i)) \approx H_d(X_{[0, \eta]}^i, \bar \partial X_{[0, \eta]}^i)$.  Moreover, in the space  $[0, \eta] \times P$ of perturbations,  the inclusion $j_\e: \{\, \e, \,\} \times (P, \partial P) \to ([0, \eta] \times (P, \partial P))$ induces  an isomorphism in homology for $\e = 0, \eta$. Therefore, the homological essentiality of the  projection $(X_0^i, \partial X_0^i)$ to $\{\, 0 \,\} \times (P, \partial P )$ is equivalent to that from $(X_\eta^i, \partial X_\eta^i)$ to $\{\, \eta \,\} \times (P, \partial P )$. The result now follows as  essentiality in homology is equivalent to essentiality in cohomology.
\end{proof}

We are ready to assemble the lemmas to prove Proposition \ref{prop Q}.  Lemma \ref{lem Q partition} shows that $\calQ$ is partitioned in a finite union of pseudomanifolds $\calQ^i$.  The first statement of Lemma \ref{lem Q essentiality} shows that each of these sets $\calQ^i$ is orientable. It follows from Lemma E.3 of \cite{GW2008} that a subset $\S^*$ is stable iff there exists $X^i$ such that: (1) $\pi(\S^*)$ is the set of $p^*$ for which there exists a point $(0, p, (p^*, \l)) \in X_0^i$; and (2) the projection from $(X_{[0, \eta]}^i, \partial X_{[0, \eta]}^i) \to ([0, \eta] \times P, \partial ([0, \eta] \times \partial P))$ is essential in cohomology. By Lemma \ref{lem Q partition}, condition (1) here is the same as  the requirement that $\pi(\S^*)$ be the set of $q$ such that $(p, (q, \l)) \in \calQ^i$.  By Lemma \ref{lem E partition}, condition (2) is the same as the requirement that the projection $\proj:(X_\eta^i, \partial X_\eta^i) \to (P, \partial P)$  be essential.  By Lemma \ref{lem Q essentiality}, this requirement is the same as asking for $\proj: (\calQ^i, \partial \calQ^i) \to (P, \partial P)$ to be essential. $\Box$ 

\section{A Characterization of Essentiality}

In this appendix we provide a characterization of essentiality, which is then used in the next appendix to prove Theorem 6.5. The ideas in these two appendices mirror those for stability but are adapted to deal with payoff---rather than strategy---perturbations. Also, we work with behavioral---rather than enabling---strategies.

Conceptually, our characterization is similar to the one that Pahl and Pimienta \cite{PP2024} obtain for the two-player case.  In a generic extensive-form game, the index of a component can be computed as the product of the indices of two fixed-point problems: one for the pruned game  obtained by deleting all paths that are not in the support of the outcome; and another that captures off-the-equilibrium-path behavior of the best-reply correspondence. The index under the first problem is always $\pm 1$; thus, the essentiality of the component depends on the second problem. In the two-player case, unlike here, Pahl and Pimienta are able to decompose this latter problem into two subproblems---called  ``excluded games''---analyzing the strategic interaction that occurs after just one of the two players deviates. 

Let $\G(u)$ be an extensive-form game.  Let $\BR^\G: B \to B$ be  the best-reply correspondence of $\G(u)$ in behavioral strategies.  $\BR^\G$ is an upper semicontinuous correspondence with nonempty, compact, and contractible values.  Therefore, one can assign an index $\text{Ind}(B^*)$ to each component $B^*$ of equilibria in behavioral strategies.  Let $G$ be the normal form of $\G(u)$.  We know that we have an index for each component $\S^*$ of equilibria in mixed strategies.    The following proposition shows that there is a one-to-one correspondence between the components of equilibria in these two representations and that their indices agree.\footnote{This result is certainly part of the folklore of the subject, but we  are recording it here since we need it and have not seen it in print.}  As a matter of notation, for each component $\S^*$ of equilibria of $G$, let $B(\S^*)$ be the set of behavioral strategy profiles that are equivalent to some profile in $\S^*$.

\begin{proposition}\label{prop same index}
	The components of equilibria in behavioral strategies of the extensive form are of the form $B(\S^*)$, where $\S^*$ is a component of equilibria of $G$. Furthermore, $\text{Ind}(\S^*)= \text{Ind}(B(\S^*))$ for each component $\S^*$ of equilibria of $G$. 
\end{proposition}

\begin{proof}
	The components of equilibria in enabling strategies are in a one-to-one correspondence with the components of equilibria of the normal form and their indices coincide under this identification---cf. Proposition 3.10  in \cite{P2023}. Thus it is sufficient to show the claim replacing mixed strategies with enabling strategies in its statement. 
	
	For each component $P^*$ of equilibria in enabling strategies, let $B(P^*)$ be the set of behavioral strategy profiles that are equivalent to it.  Let $\phi: B \to P$ be the map that sends behavioral strategies to equivalent enabling strategies.  Then for each component $B^*$ of equilibria of $\G(u)$, $\phi(B^*)$ is a connected set of equilibria and is contained in a component $P^*$ of equilibria.  Obviously, $B(P^*)$ contains $B^*$.  To prove the first statement, then, it is sufficient to show that $B(P^*)$ is connected for a component $P^*$.  For each $\e > 0$, let $U^{\e}$ be the $\e$-neighborhood of $P^*$.  There is a homeomorphism $\psi: P \backslash \partial P \to B \backslash \partial B$ that sends $p$ to an equivalent profile $\psi(p)$ of behavioral strategies.  For each $\e$, the closure $B^\e$ of $\psi(U^\e \backslash \partial P)$ is a connected set that contains $B(P^*)$.  Since $B(P^*)$ is $\lim_{\e \downarrow 0} B^\e$, it is connected.
	
	To prove the second statement, take a component $P^*$ of equilibria.  Take a neighborhood $U$ of $P^*$ that is disjoint from the other components.  Let $V$ be the closure of $\psi(U \backslash \partial P)$, where $\psi$ is as in the previous paragraph. $V$ is a neighborhood of $B(P^*)$ that is disjoint from the other components. Fix $\e > 0$ and let $B^\e$ be the set of $b \in B$ such that for each $n, a_n$, $b_{n,a_n} \ge \e$. Let $P^\e$ be the set of enabling strategies $p$ such that $\psi(p) \in B^\e$. $P^\e$ is then a polyhedron.  Use $\BR^{P, \e}$ (resp.~$\BR^{\G, \e}$) to denote the best-reply correspondence with this restricted strategy space.  If $\e$ is small, the index of $P^*$ (resp.~$B(P^*)$) can be computed as the index of $\BR^{P, \e}$ (resp.~$\BR^{\G, \e}$) over $U$ (resp.~$V$). The result now follows as $\psi$ is a homeomorphism that commutes with the best-reply correspondences. 
\end{proof}

There is a convex-valued selection $\BR$ of $\BR^\G$ that is useful for this appendix.  For each $b$, $n$, and each information set $h_n$ of $n$ that is reachable by $b_{-n}$ (if $n$ chooses all actions leading to $h_n$), as beliefs are well-defined there, say that an action $a_n \in A_n(h_n)$ is optimal if there is a continuation strategy at $h_n$ that is optimal given his beliefs at $h_n$ and in which he chooses $a_n$. Then $\BR(b)$ assigns to each information set $h_n$ of each player $n$ that is reachable by $b_{-n}$ the set of mixtures over optimal actions against $b_{-n}$.  The fixed points of $\BR$ are a subset of the set of Nash equilibria.

Let $B^*$ be a component  of equilibria in behavioral strategies of a game $\G(u)$ and denote its index by $\text{Ind}(B^*)$. We assume that $u$ is generic so that every equilibrium in $B^*$ induces the same outcome, and hence the equilibrium payoff of each player $n$, call it $v_n^*$, is unique. For each $n$, let $H_n^+$ be the set of information sets of player $n$ that are reached with positive probability under the equilibria of $B^*$. For each $n$, the mixture at each information set in $H_n^+$ is constant across all equilibria in $B^*$. Let $A_n^+$ (resp.~$A_n^0$) be the set of actions $a_n$ that belong to some $h_n \in H_n^+$ and that are chosen with positive (resp.~zero) probability by the equilibria of $B^*$.  For each $n$, let $H_n^1$ be the set of $h_n  \notin H_n^+$ such that at each $h_n' \in H_n^+$ that precedes it, $n$'s action that leads to $h_n$ is in $A_n^+$---thus, the information sets in $H_n^1$ are excluded under $B^*$ by $n$'s opponents. For each $a_n^0 \in A_n^0$, let $H_n(a_n^0)$ be the collection of information sets that succeed $a_n^0$ and let $H_n^0 = \cup_{a_n \in A_n^0} H_n(a_n^0)$.  The information sets in $H_n^0$ are reached only if $n$ chooses a nonequilibrium action $a_n^0$. Let $H_n^- = H_n^0 \cup H_n^1$. Obviously, the sets $H_n^+, H_n^0, H_n^1$ partition $H_n$.\footnote{For simplicity in notation, this appendix and the next are written as if for each $n$, all three collections, $H_n^+$, $H_n^0$, and $H_n^1$, and hence all objects defined using them, are nonempty. If one of them, especially something that shows up as a factor in a product space, is empty, it is to be simply ignored.}

Every $b_n \in B_n$ can be written uniquely as a triplet $(b_n^+, b_n^0, b_n^1)$, where $b_n^+$ is the restriction of $b_n$ to information sets in $H_n^+$, and for $i = 0, 1$,  $b_n^i$ is the restriction of $b_n$ to $H_n^i$; alternatively, we also write $b_n$ as $(b_n^+, b_n^-)$ where $b_n^- = (b_n^0, b_n^1)$. For each $n$, let $B_n^+$ be the set of $b_n^+$ whose support is contained in $A_n^+$, and let $B^+ = \prod_n B_n^+$. Observe that the projection of $B^*$ to $B^+$ is a singleton, denoted $b^{*,+}$, with $n$'s component being $b_n^{*,+}$.   

For $i = 0, 1$, let $B_n^i$ be the restriction of $B$ to the information sets in $H_n^i$.  $B_n^- \equiv B_n^0 \times B_n^1$. Also, for $i = 0, 1$, let $B^i = \prod_n B_n^i$, and $B^- = B^0 \times B^1$. Let the projection of $B^*$ to $B^-$ be denoted $B^{*, -}$. $B^{*,-}$ is of the form $B^0 \times B^{*, 1}$. If the payoffs are generic, then  $B^{*,1}$ is a full-dimensional subset of $B^1$ consisting of strategies that deter deviation by players: indeed, the set is defined by a set of polynomial inequalities, and generically it is the closure of the set of profiles that satisfy the inequalities strictly. Denote by $\partial B^{*,-}$ the boundary of $B^{*, -}$ in $B^-$: $\partial B^{*,-} = B^0 \times \partial B^{*,1}$, where $\partial B^{*, 1}$ is the boundary of $B^{*,1}$ in $B^1$.

Fix a closed neighborhood $U^{*,+}$ of $b^{*,+}$ in $B^+$ such that: (1) for each $b^+ \in U^{*,+}$, the support is $A^+$; (2) for every other equilibrium component, if the equilibrium outcome has the same support as that for $B^*$, then its projection to $B^+$ is outside $U^{*,+}$ or coincides with that of $B^*$.\footnote{For generic two-player games, no two components can induce the same equilibrium distribution.  But with more than two players, that property seems unlikely to be true, though we do not know of an example.}  Given $b^+ \in U^{*,+}$, computing the optimal  action in the set $A_n^+ \cap A_n(h_n)$ for each $h_n \in H_n^+$ involves only the information about $b^+$.  Therefore, we have a well-defined best-reply correspondence $\BR^{*,+}: U^{*,+} \to B^+$, which can be viewed as the best-reply correspondence of the pruned game obtained by deleting all paths that are not in the support of the equilibrium distribution associated with $B^*$.  If $u$ is generic, then $b^{*,+}$ is the unique fixed point of $\BR^{*,+}$ in $U^{*,+}$ and its index, denoted $\text{Ind}(b^{*,+})$, is $\pm 1$---indeed, viewed as an equilibrium of the pruned game tree, it is a regular equilibrium. If $H_n^1$ is empty for each $n$, then either $H_n^0$ is empty for each $n$ and $B^*$ is an isolated equilibrium, or $B^* = \{\, b^{*,+} \, \} \times B^0$ and either way $\text{Ind}(B^*) = \pm 1$. This case is uninteresting for the purposes of Theorem 6.5: in the reduced normal form of the game, the component is an isolated equilibrium.   Therefore, we will assume that $H_n^1$ is nonempty for some $n$, and thus $B^{*,1}$ is a nonempty set.

Fix now a closed neighborhood $U^{*,-}$ of $B^{*, -}$ in $B^-$ such that  $U^{*,+} \times U^{*,-}$ is disjoint from all other components of equilibria of $\G(u)$. Our objective now is to define consistent beliefs at information sets in $H^-$ for profiles of the form $(b^{*,+}, b^-)$ with $b^- \in U^{*, -}$. (The beliefs at information sets in $H^+$ are given by the equilibrium strategies.)  Let $A^0 = \cup_n A_n^0$ and let $\D$ be the set of all probability distributions over $A^0$. Given $(b^0, b^1, \d) \in U^{*,-} \times \D$ and, for $\e$ a small positive number, consider the strategy profile  $b^{\e, \d} \in B$ defined as follows.  For each $n, h_n \in H_n^+$, we perturb the strategy $b_n^{*,+}$ so that each $a_n^0 \in A_n(h_n) \cap A_n^0$ is played with probability $\e\d_{a_n^0}$ and with probability $(1-\sum_{a_n^0 \in A_n(h_n)} \e\d_{a_n^0})$, the mixture $b_n^{*,+}(h_n)$ is played; at all other information sets, players randomize according to $(b^0, b^1)$. For each small $\e$, $b^{\e, \d}$ induces, for each $n$, beliefs at all information sets in $H_n^0 \cup H_n^1$ that are reachable by $b_{-n}^{\e, \d}$ (even if $n$'s strategy excludes them). Letting $\e$ go to zero, we get a limiting vector of beliefs $\mu_n^{b_{-n}^-, \d}$ over all such  information sets in $H_n^-$.  Note that $\mu_n^{b_{-n}^-, \d}$ need not specify beliefs at all information sets in $H_n^-$, only those that are enabled by $b_{-n}^{\e, \d}$ for $\e > 0$. We use $\mu^{b^-, \d}$ to denote the vector of beliefs of all players at the profile $(b^-, \d)$. 

For $(b^-, \d) \in U^{*,-} \times \D$, let $H_n^0(b^-, \d)$ be the information sets in $H_n^0$ where player $n$'s beliefs $\mu_n^{b_{-n}^-, \d}$ are well-defined and assign zero probability to any node that is reached by some opponent $m$ choosing a nonequilibrium action $a_m^0$; likewise let $H_n^1(b^-, \d)$ be the information sets in $H_n^1$   where the beliefs are well-defined and assign zero probability to nodes that are reached only when two or more of his opponents deviate to a nonequilibrium action.   At an information set in $H_n^i(b^-, \d)$, $i = 0, 1$, the beliefs are continuous in $(b^-, \d)$. To see this, suppose $h_n \in H_n^0(b^-, \d)$; then the information set is reachable by $(b_{-n}^*, b_{-n}^1)$ (if $n$ were to choose all his actions leading to it) and $\mu^{b^-, \d}$ assigns positive probability only to nodes that are reached thus; clearly the beliefs are continuous at $(b^-, \d)$; a similar argument works to show continuity at information sets in $H_n^1(b^-, \d)$. 

We will say that a mixture at an information set $h_n \in H_n^0 \cup H_n^1$ is optimal against $(b^-, \d)$ either if $h_n \notin H^0(b^-, \d) \cup H^1(b^-, \d)$ or if it is optimal against the induced strategy-belief pair $(b^-, \mu^{b^-,\d})$.  We now have a best-reply correspondence $\BR^{*,-}: U^{*,-} \times \D \to B^-$ that is {\it well-behaved}, in the sense that it is an upper semicontinuous correspondence with nonempty, compact, and convex values.

We define a correspondence $\BR^{*,\D}: U^{*, -}  \to \D$ by viewing $\d \in\D$ as being a choice variable of an outsider $O_\D$.  His payoff if he chooses $a_n^0$ is $G_\D(a_n^0) \equiv G_n(\tilde b_n(a_n^0, b_n^-), b_{-n}^{*,+}, b_{-n}^-) - v_n^*$, where $\tilde b_n(a_n^0, b^-)$ is the following strategy: at the information set $h_n(a_n^0)$, it plays $a_n^0$ and at all other information sets in $H_n^+$ it plays according to $b_n^{*, +}$; at information sets in $H_n^-$,  it plays according to $b_n^-$. Thus, his best-reply correspondence is:
$
\BR^{*,\D}(b^-) = \text{argmax}_{\tilde \d} \sum_{n,a_n^0} \tilde \d_{n, a_n^0}[G_n(\tilde b_n(a_n^0, b_n^-), b_{-n}^{*,+}, b_{-n}^-) - v_n^*]
$ 

This correspondence, too, is well-behaved. We would now like to compute the index of the fixed points of the product correspondence $\BR^{*,-} \times \BR^{*, \D}$. Since the domain is, possibly, a strict subset of the range, the first property we have to establish is that the correspondence does not have a fixed point on the boundary of the domain, which is set out in the following lemma.  

\begin{lemma}\label{lem boundary of B}
	If $u$ is generic, $\BR^{*,-} \times \BR^{*, \D}$ does not have a fixed point on $\partial B^{*,-} \times \D$.
\end{lemma}   

\begin{proof} 
	Since we are proving a genericity property, all objects will be indexed by payoff functions $u \in \Re^{N \times Z}$. Let $\U^*$ be the set of payoff functions $u$ such that the game $\G(u)$ has finitely many equilibrium outcomes, and  the on-path equilibrium support $A_{n,u}^+$ is constant across $u$ for some equilibrium outcome.  The sets $H_{n,u}^+$, $H_{n,u}^0$, $H_{n,u}^1$ for this outcome are then all constant across $u$ and we can construct the best-reply correspondence $\BR_u^{*,-} \times \BR_u^{*,\D}$ for each $u$ as we have.    Let $\bar \U^*$ be the set of $u \in \U^*$ such that  there exists a fixed point $(b^-, \d) \in \partial B_u^{*,-} \times \D$  of $\BR_u^{*,-} \times \BR_u^{*, \D}$. For each $n$, let $\hat A_{n,u}^0$ be the set of $a_n^0 \in A_n^0$ such that $\tilde b(a_n^0, b^-)$ yields $v_{n,u}^*$ against $(b_{-n}^{*, +}, b_{-n}^-)$, where $v_{n,u}^*$ is $n$'s equilibrium payoff in $\G(u)$, and $\tilde b_n$ is as defined for the payoffs of $O_\D$. Let $\hat A_u^0 = \cup_n \hat A_{n,u}^0$. Since $b^-$ belongs to the boundary of $B_u^{*,-}$, $\hat A_u^0$ is nonempty.  Clearly, the support of $\d$ is contained in $\hat A_u^0$.   $\bar \U^*$ itself can now be partitioned into a finite number of subsets $\bar \U^1, \ldots, \bar \U^k$, where for each $i$, $\hat A_u^0$ is constant across all $u$ in $\bar \U^i$.  It is sufficient to show that $\bar \U^i$ is a lower-dimensional set for each $i$.
	
	Fix a $\bar \U^i$ and let $\hat A^0$ be the constant set $\hat A_u^0$ across $u$ in $\bar \U^i$. For each $n$, let $\hat H_n^0$ be the information sets in $H_n^0(b^-, \d)$ that follow some $a_n^0 \in \hat A_n^0$ and are reached by $b_n^0$.  The beliefs and the continuation payoffs at information sets in $\hat H_n^0$ are multilinear in $(b_{-n}^{*,+}, b_{-n}^1)$.  Likewise let $\hat H_n^1$ be the information sets in $H_n^1(b^-, \d)$. The beliefs and continuation payoffs at  information sets in $\hat H_n^1$ are multilinear in $(b_{-n}^{*, +}, b_{-n}^-, \d)$. Let $\hat H^- = \cup_n (\hat H_n^0  \cup \hat H_n^1)$ and let $\hat A^-$ be the set of actions at the information sets in $\hat H^-$ that are in the support of $b^-$.  
	
	Consider the following set of equations in the variables $(u, \hat b^+, \hat b^-, \hat \d)$ where  the support of $\hat b^+$ is $A^+$, the support of $\hat b^-$ is $\hat A^-$, and the support of $\hat \d$ is $\hat A^0$:   (1) $\hat b^+$ is an equilibrium of the game $\G(u)$ where actions other than $A^+$ are eliminated; (2) the actions in $\hat A^-$ are all optimal given $(b^0, b^1, \d)$; (3) the action $\hat a_n^0$ yields the same payoff as any other $a_n^+ \in A_n(h_n(a_n^0))$. The solutions to the first two sets of equations form a set of dimension $|NZ| + |\hat A^0| - 1$.  However, we also have to satisfy the equations in (3),  and thus the set of solutions has dimension $|NZ| - 1$.  For generic $u$, then, there is no solution $(u, \hat b^+, \hat b^-, \hat \d)$ to the set of equations. The result now follows since $(u, b^+, b^-, \d)$ must satisfy the equations for the fixed point $(b^-, \d)$ for the game $\G(u)$ with $u \in \bar \U^i$.  
\end{proof}

Thanks to the lemma, $\BR^{*, -} \times \BR^{*, \D}$ has a well-defined fixed-point index for the set of fixed points in  $B^{*,-} \times \D$. We denote this index by $\text{Ind}(B^{*,-})$.  Also, if necessary by replacing  $U^{*, -}$  with  a smaller closed neighborhood, we can assume that $\BR^{*,-} \times \BR^{*, \D}$ has no fixed points in  $(U^{*,-} \backslash B^{*, -}) \times \D$. 

Putting together all three correspondences we have defined, we have a correspondence $\BR^*: U^{*,+} \times U^{*,-}\times \D \to B^+ \times B^- \times \D$.  $\BR^*$  now has a well-defined index,  which is $\text{Ind}(b^{*,+})\times \text{Ind}(B^{*,-})$. The following proposition is our characterization of the index of components.

\begin{proposition}\label{prop essentiality}
	$\text{Ind}(B^*) = \text{Ind}(b^{*,+}) \times \text{Ind}(B^{*, -})$. 
\end{proposition}

\begin{proof}
	For each $\d \in \D$ and $\e \in [0, 1)$, consider the following perturbed game $\bar \G^{\e, \d}$.  For each player $n$ and each information set $h_n \in H_n^+$ where there is at least one $a_n^0 \in A_n^0$, player $n$ first decides which of the actions $a_n^0$ there to play; and if he decides against playing any of them, then Nature chooses with probability $\e \d_{a_n^0}$ to automatically implement $a_n^0$ for each $a_n^0$ there and with the complementary probability allows $n$ to choose one of his actions $a_n^+$ there. Even if Nature implements an $a_n^0$, player $n$ still has to choose an action at  each of the succeeding information sets. $n$'s opponents are unable to distinguish whether $n$ chose $a_n^0$ or Nature did. We represent a strategy of player $n$ as a vector $(b_n^+, b_n^0,  b_n^1, \eta_n, \hat b_n^0)$ where: $b_n^+$ gives  $n$'s choices over $A_n^+(h_n)$ for each $h_n \in H_n^+$; $b_n^0$ represents the choices after Nature chooses an $a_n^0$ and $\hat b_n^0$ represents choices after $n$ himself chooses an $a_n^0$; $b_n^1$ represents $n$'s choices at information sets in $H_n^1$; $\eta_n$ is a vector in $\Re^{A_n^0}$ with $\eta_{a_n^0}$ being the probability with which $n$ (and not Nature) chooses $a_n^0 \in  A_n^0$ at the stage where $n$ decides which of the actions in $A_n^0 \cap A_n(h_n)$, if any, to play.  
	
	Let $\bar G^{\e, \d}$ be the normal form of the game $\bar \G^{\e, \d}$. $\bar G^{0, \d}$ is obtained from $G$ by the addition of duplicate strategies (created by the sequential nature of choosing actions by first deciding on nonequilibrium actions).  Hence the index of each  component of equilibria of $\bar G^{0, \d}$ is the same as its equivalent component in $G$.  Let $\bar B^*$ be the component of equilibria of $\bar \G^{0, \d}$ that is equivalent to $B^*$, that is, $\bar B^*$ is the set of all $(b^{*,+}, b^0, b^1, 0, \hat b^0)$ with $b^0, \hat b^0 \in B^0$ and $b^1 \in B^{*, 1}$. We then have that $\text{Ind}(\bar B^*) = \text{Ind}(B^*)$.

	Let $\bar \BR^{\e, \d}$ be the selection of the best-reply correspondence of $\bar G^{\e, \d}$ we defined at the beginning of this appendix. Players optimize at all information sets that are not excluded by their opponents; but we also require each player $n$ to optimize at information sets following Nature's choice of $a_n^0$ (even if  $\d_{a_n^0} = 0$).  There is now a neighborhood $\bar U$ of $\bar B^*$ such that for all small $\e > 0$ and any $\d \in \D$,  there are no fixed points of $\bar \BR^{\e, \d}$ on the boundary of $\bar U$ and the index of $\bar \BR^{\e, \d}$ over $\bar U$ is $\text{Ind}(B^*)$.  For such small $\e > 0$, consider now the correspondence $\bar \BR^\e \times \text{Id}$ on $U \times \D$ where $(\bar \BR^\e \times \text{Id})(\cdot, \d) = \bar \BR^{\e,\d}(\cdot) \times \{\, \d \, \}$.   The index of this correspondence over the set $\bar U \times \D$ is also $\text{Ind}(B^*)$.  Furthermore,  the index remains unchanged if we change $\text{Id}$ to any well-behaved correspondence: indeed, for the linear homotopy between $\text{Id}$ and this new correspondence, there are no fixed points in the boundary of $\bar U$ and hence the index remains constant along the homotopy. In particular, the index remains unchanged if we replace  $\text{Id}$ with the correspondence $\bar \BR^{\e, \D}$ defined as follows: 
	$\bar \BR^{\e, \D}(b^{+}, b^0,  b^1, \eta, \hat b^0, \d) = \BR^{*,\D}(b^0, b^1).$

	Suppose now that for each small $\e > 0$ we have the following property: given a fixed point $(b^+, b^0, b^1, \eta, \hat b^0, \d)$ in $\bar U^* \times \D$ of $\bar \BR^{\e} \times \bar \BR^{\e, \D}$, for each $n$, each  $a_n^0 \in A_n^0$ is an inferior action in the game $\bar \G^{\e, \d}$ against $(b^+, b^0, b^1, \eta, \hat b^0)$.  Then, the index of $\bar \BR^{\e} \times \bar \BR^{\e, \D}$ over $\bar U^* \times \D$ can be computed using the pruned tree $\hat \G^{\e, \d}$ where the actions $a_n^0$ and everything that follows them are deleted. Let $\hat \BR^{\e} \times \hat \BR^{\e, \D}$ be the best-reply correspondence with the pruned tree.  Observe that each neighborhood of the graph of $\BR^{*,+} \times \BR^{*,-} \times \BR^{*, \D}$ over $U^{*,+} \times U^{*,-} \times \D$ contains the graph of $\hat \BR^{\e} \times \hat \BR^{\e, \D}$ over $U^{*,+} \times U^{*,-} \times \D$ for all small $\e$. Therefore, the result follows.

	It remains to show that for small $\e > 0$, in any fixed point $\bar b \in \bar U^*$, the actions $a_n^0$ are not optimal.  To do that, 
	take  a sequence of fixed points $(b^{+, \e}, b^{0, \e},  b^{1, \e},  \eta^\e, \hat b^{0, \e}, \d^\e)$ of $\bar \BR^\e \times \bar \BR^{\e, \D}$ such that, letting $\tilde b^\e$ be the equivalent behavioral strategy profile in $\G$, $\tilde b^\e$ converges to some $\tilde b \in B^*$.   Then $\tilde b$ is of the form $(b^{*,+}, \tilde b^-)$. Define $\tilde \d_{a_n^0}^\e = {\sum_{a_m^0} \tilde b^\e(a_m^0)}^{-1} \tilde b^\e(a_n^0)$  and let $\tilde \d$  be its limit.   Clearly $(b^{*,+}, \tilde b^0, \tilde b^1, \tilde \d)$ is a fixed point of $\BR^{*,+} \times \BR^{*,-} \times \BR^{*, \D}$, implying that $(\tilde b^0, \tilde b^1)$ belongs to the interior of $B^{*,-}$.   In other words, for small $\e$, the actions $a_n^0$ are suboptimal as was to be shown.
\end{proof}

We will now provide an implication---an equivalence, actually, though we do not prove the other direction---of our index computation that is used in the next appendix. For all sufficiently small $\zeta$, if we replace $\BR^{*,-} \times \BR^{*, \D}$ by $\zeta$-best-reply correspondences, i.e.,~where the probability of an action that is not $\zeta$-optimal  is no more than $\zeta$,  still there is no fixed point in $\partial B^{*,-} \times \D$. Fix such a $\zeta$. Let $B^{*, \zeta, -}$ be the set of $(b^0, b^1) \in B^{*,-}$ such that $b^0(a_n) \ge \zeta$ for all $a_n$ that belong to $ A_n(h_n)$ for some $h_n \in H_n^0$. Similarly, let $\D^\zeta$ be the set of $\d \in \D$ such that $\d_{a_n^0} \ge \zeta$ for all $a_n^0$.  Let $\BR^{*, \zeta, -}$ and $\BR^{*, \zeta, \D}$ be the best-reply correspondences when the choice sets are restricted to be  $B^{*,\zeta, -}$ and $\D^\zeta$, respectively.  There are no fixed points in $\partial B^{*, -} \times \D$ and the fixed-point index is still the same.

For each $(b^-, \d) \in B^{*,\zeta, -} \times \D^\zeta$ such that $b^1$ is completely mixed as well, and $\e > 0$, the beliefs $\mu^{\e, b^-,\d}$ are well-defined at all information sets in $H^-$.  Let $\bbB^*$ be the set of $(b^-, \mu, \d)$ that are limits of $(b^\e, \mu^{\e, b^-, \d}, \d^\e)$ as $\e \to 0$.\footnote{  $\bbB^{*}$, or rather just its first two factors, is the projection to the information sets in $H^-$ of pairs of strategy profiles and consistent beliefs where an equilibrium in $\{\, b^{*,+}\, \} \times B^{*, \zeta, -}$ is played.}  $\partial \bbB^{*}$ refers to the points $(b^-, \mu, \d) \in \bbB^{*}$ with $b^- \in \partial B^{*,-}$.

There is now a well-defined best-reply correspondence from $\bbB^{*}$ to $B^{-} \times \D$, where at each information set $h_n$ in $H_n^-$, player $n$ optimizes relative to the beliefs there and subject to choosing each action with probability at least $\zeta$ if $h_n \in H_n^0$,  and the best-reply for the factor $\D^\zeta$ is as before.  We use $\bbBR^*$ to denote this correspondence. Given a vector $g \in \Re^{A \backslash A^+}$, $\G(u) \oplus g$ denotes the game where the payoffs to an action $a_n \notin A^+ \cup A^0$ at an information set $h_n \in H_n^-$ are augmented by a bonus $g_{a_n}$; and for the outsider controlling choices in $\D$, the payoffs to $a_{n^0}$ are augmented by $g_{a_n^0}$. We use $\bbBR_g^{*}$ to denote the best-reply correspondence when the payoffs have been augmented using $g$.

Let $V$ be the closure of a connected open set in  $B^{*, \zeta, -} \times \D^\zeta$  that is contained in its interior and that is nonempty if $\text{Ind}(B^*) \neq  0$. (The set $V$ could be empty if the index of $B^*$ is zero.)    We now have the following proposition, whose  proof adapts the logic of Govindan, Laraki, and Pahl \cite{GLP2023}.

\begin{proposition}\label{prop zero index}
	There exists a function $g: \bbB^{*} \to \Re^{A \backslash A^+}$ such that: 
	\begin{enumerate}
		\item $g$ is zero in a neighborhood of $\partial \bbB^{*}$ in $\bbB^{*}$; 
		\item if $(b^-,  \d) \notin V$, then $(b^-, \d) \notin \bbBR_{g(b^-, \mu, \d)}^{*}(b^-, \mu, \d)$ for any $(b^-,  \mu, \d) \in \bbB^{*}$. 
	\end{enumerate}
\end{proposition}

\begin{proof}
	 We define a map $f^*: B^{*,\zeta,-} \times \D^\zeta \to B^{\zeta,-} \times \D^{\zeta}$, which is a version of the Nash map.  For each $(b^-, \d)$ and $n$, recall that $H_n^1(b^-, \d)$ is the set of information sets where there is some node that is reached when exactly one of $n$'s opponents $m$ chooses an $a_m^0$. For each   $h_n \in H_n^1(b^-, \d)$, and a node that is reached if only one player $m$ deviated to some $a_m^0$,  we can define the probability of reaching this node by the product of the probabilities of all actions leading up to it, with $\d_{a_m^0}$ used for action $a_m^0$; we define the probability $P(h_n \mid b^-, \d)$ as the sum of these probabilities over such nodes. $f^*$ assigns the following $\tilde b(h_n)$ at $h_n$: for each action $a_n \in A_n(h_n)$, 
	\[
	\tilde b_{a_n}(h_n)  \propto b_{a_n}(h_n) + P(h_n \mid b^-, \d)\max[G_n(a_n, b^-) - G_n(b(h_n), b^-), 0]
	\] 
	where $G_n(a_n, b^-)$ is the best payoff achievable given beliefs $\mu$ conditional on being at $h_n$ and playing $a_n$; and similarly for $G_n(b(h_n), b^-)$. 
	
	An information set in $H_n^0(b^-, \d)$ has at least one node that is reached without any opponent playing an action in $A^0$.  For an information set $h_n$ in $H_n^0(b^-, \d)$, let $P_n(b^-, \d)$ be the probability of reaching this information set if $n$ chooses all the actions leading up to it. $f^*$ assigns the following strategy $\tilde b(h_n)$ at $h_n$: 
	\[
	\tilde b_{a_n}(h_n) -\zeta \propto \frac{b_{a_n}(h_n) -\zeta}{1-|A_n(h_n)|\zeta} + P(h_n \mid b^-, \d)\max[G_n(\hat b_n(a_n), b^-, \mu) - G_n(b(h_n), b^-, \mu), 0]
	\] 
	where $\hat b_n(a_n)$ is the strategy that plays actions $a_n' \neq a_n$ in $A_n(h_n)$ with probability $\zeta$ and plays $a_n$ with the remaining probability; $G_n(\hat b_n(a_n), b^-)$ is the best payoff achievable conditional on being at $h_n$ and playing $\hat b_n(a_n)$; and  $G_n(b(h_n), b^-)$ is defined similarly. At all other information sets $h_n \in H_n^-$, $f^*$ is the identity. 
	
	For the $\d$-coordinate,  we use the Nash map: 
	\[
	\tilde \d_{a_n^0} -\zeta  \propto  \frac{\d_{a_n^0} -\zeta}{1-|A^0|\z} + \max[G_\D(\hat \d(a_n^0)) - G_\D(\d), 0],
	\]
	where $\hat \d(a_n^0)$ plays $a_n^0$ with probability $1-(|A^0|-1)\zeta$ and every other $a_m^0$ with probability $\zeta$.
	
	The fixed points of $f^*$ are the same as those of $\BR^{*, \zeta, -} \times \BR^{*, \zeta, \D}$ and there is a linear homotopy between $f^*$ and $\BR^{*,\zeta, -} \times \BR^{*, \zeta, \D}$ that preserves the set of fixed points.  Hence  the index of $f^*$ over $B^{*,\zeta, -} \times \D^\zeta$ is the same as that of $\BR^{*, \zeta, -} \times \BR^{*, \zeta, \D}$, which in turn is $\text{Ind}(B^*)$.  When $V$ is nonempty, we can construct a map $f_V: V \to B^{\zeta, -} \times \D^\zeta$ whose index is the same as that of $f^*$. By the Hopf Extension Theorem (Corollary 8.1.18 in \cite{ES1966}), there now exists a map $f: B^{*, \zeta, -} \times \D^\zeta \to B^{\zeta, -} \times \D^{\zeta}$ that coincides with $f^*$ in a neighborhood $W$ of $\partial B^{*, \zeta, -} \times \D^\zeta$ and with $f_V$ on $V$, when $V$ is nonempty, and that has no fixed points outside $V$. In particular, if $b^-$ belongs to the boundary of $B^{*, \zeta, -}$ and $\d = f_\D(b^-, \d)$, then $f_{h_n}(b^-, \d) \neq b_{h_n}^-$ for some $n$ and $h_n \in H_n^0(b^-, \d) \cup H_n^1(b^-, \d)$.
	
	Define $\bbF: \bbB^* \to B^{\zeta,-} \times \D^\zeta$ by $\bbF(b^-, \mu, \d) = f(b^-, \d)$. Then $f(b^-, \mu, \d) \neq (b^-, \d)$ if $(b^-, \d) \notin V$. For each $h_n \in H_n^-$, let $X^0(h_n)$ be the set of $(b^-, \mu, \d)$ such that $\bbF_{h_n}(b^-, \mu, \d) = b^-(h_n)$ and let $X^1(h_n)$ be its complement.  For each $(b^-, \mu, \d) \in X^1(h_n)$, letting $\tilde b(h_n)$ be the $h_n$-coordinates under $\bbF$, we have that $\tilde b_n \neq b_n$. For such a point, take the ray from $b_n$ through $\tilde b_n$.  There is a unique point of intersection, call it $r^0(b^-, \d, h_n)$, between this ray and the boundary of  the minimal face of the projection of $B^{\zeta, -}$ to $\D(A_n(h_n))$ containing the interval between $b_n(h_n)$ and $\tilde b_n(h_n)$ in the affine space generated by this face.  Let $X_{a_n}^1$ be the closure of the set of points in $X^1(h_n)$ such that $r_{a_n}^0(\cdot) \ge r_{a_n^{'}}^0(\cdot)$ for all $a_n' \in A_n(h_n)$.  We construct the sets $X^0(\D)$, $X^1(\D)$, and $X_{a_n^0}^1$ similarly.
	
	Let $\beta^1_{h_n}: \bbB^* \to [0, 1]$ be a function that is zero on $X^0(h_n)$ and strictly positive elsewhere.  For each $a_n$, let $\beta_{a_n}^2: \bbB^* \to [0, 1]$ be a function that is one on $X_{a_n}^1$ and strictly smaller than one elsewhere.  Let $\beta^3_{h_n}$ be a function that is zero on both $V$ and a neighborhood $U \subsetneq W$ of $\partial \bbB^*$  and one outside $W \cup V$. The functions $\beta_\D^1$, $\beta_{a_n^0}^2$, and $\beta_\D^3$ are defined similarly. We define $g_{a_n}(b^-, \mu, \d)$ recursively backwards for the collection of information sets in $H_n^-$.
	\[
	g_{a_n}(b^-, \mu, \d) = \beta^3_{h_n}(b^-, \mu, \d)\beta_{a_n}^2(b^-, \mu, \d)[v_n(b^-, \mu, h_n) - G_n(a_n, b^-, \mu, \d) + \beta^1_{h_n}(b^-, \mu, \d)]
	\]
	where $v_n(b^-, \mu, h_n)$ is the best payoff at $h_n$ against $(b^-, \mu, \d)$ conditional on reaching $h_n$ and given that bonuses have already been defined for succeeding information sets.  $g_{a_n^0}$ is defined similarly for each $a_n^0 \in A^0$. The function $g$ has the requisite properties.  
\end{proof}

\section{Proof of Theorem 6.5}

Essential components exist and are closed and connected sets of equilibria.  Moreover, each essential set contains a stable set, which in turn contains a proper equilibrium; hence essentiality satisfies Axiom B. Propositions \ref{prop Axiom D} and \ref{prop Axiom I} show  that essentiality satisfies Axioms D and I.  To complete the proof, fix a solution concept $\varphi$ satisfying Axioms D and I. Fix also a generic extensive-form game $\G(u)$ and a solution $\S_\varphi^*$ of the reduced normal form of  $\G(u)$ assigned by $\varphi$.  Let  $\S^*$ be the component of equilibria of the reduced normal form  that contains $\S_\varphi^*$. Suppose either that $\S^*$ has index zero or that $\S_\varphi^* \neq \S^*$.  We will show that $\varphi$ fails Axiom B.   

Let $B^*$ and $B_\varphi^*$ be, respectively, the sets of behavioral strategies that are equivalent to $\S^*$ and $\S_\varphi^*$.   Then,  either $\text{Ind}(B^*) = 0$ or $B_\varphi^*$ is a strict subset of $B^*$. 

\smallskip

\noindent {\bf Step 1---Fixed-Point Preliminaries.} We first apply Proposition \ref{prop zero index} to get a bonus function $g$. To that end, choose $\zeta$ as required there.  When $\text{Ind}(B^*) = 0$, let $V$ be the empty set; otherwise, choose $V$ to be the closure of a nonempty open set that is disjoint from $B_\varphi^* \times \D$ and from the boundary of $B^{*,\zeta, -} \times \D^\zeta$, as specified by Proposition \ref{prop zero index}.  There now exists a function   $g: \bbB^* \to \Re^{A \backslash A^+}$ with the properties specified in Proposition \ref{prop zero index}.   

 The ambient space of $\bbB^*$ is the set $\bbB$ of all $(b^-, \mu, \d)$ where $(b, \d) \in B^{\zeta, -} \times \D^\zeta$ and $\mu$ specifies a probability distribution for each $h_n$ over the nodes in it.  The correspondence $\bbBR^*$ extends to a correspondence $\bbBR$  over $\bbB$ in the obvious way, as does the perturbed correspondence $\bbBR_g^*$ for any bonus vector $g$. We can extend  the function $g$ to the whole of $\bbB$, still denoting it $g$. 
 
Let $\bbB_\varphi^*$ be the subset of $\bbB^*$ consisting of points $(b^-,\mu, \d)$ such that $b^- \in B_\varphi^{*, -}$. There exists $ \eta_0 > 0$ such that if  $(b^-, \mu, \d)$ is within $\eta_0$ of some $(\tilde b^-, \tilde \mu, \tilde \d)$ in $\bbB_\varphi^*$, and $g$ is within $\eta_0$  of $g(\tilde b^-, \tilde \mu, \tilde \d)$ then $(b^-, \mu, \d) \notin \bbBR_g(b^-, \mu, \d)$. There now exists $0 < \eta < \eta_0$ such that: (1) if $(b^-, \mu, \d)$ and $(\tilde b^-, \tilde \mu, \tilde \d)$ are within $\eta$ of one another, then their $g$ values are within $\eta_0$ of one another; (2) $g$ is zero in the $2\eta$-neighborhood of $\partial \bbB^*$ in $\bbB$.   

\smallskip

\noindent{\bf Step 2---Triangulations.} $\bbB$ is  a product of simplices: there is a simplex $\D^\zeta(A_n(h_n))$ of distributions over actions at $h_n$ for each $h_n \in H_n^0$, where the probability of each action is at least $\zeta$; there is a simplex $\D(A_n(h_n))$ of distributions over actions at $h_n \in H_n^1$; there is a simplex $\D(h_n)$ of distributions over the nodes $x$ of $h_n$ for each $h_n \in H_n^-$;  and there is the simplex $\D^\zeta$ of distributions over $A^0$. We can take a triangulation of each of these simplices such that the diameter of each simplex of the triangulation is less than $\eta$.  

For each $h_n \in H_n^-$, let $\X^0(h_n)$ be the polyhedral complex generated by the triangulation of the simplex of actions at $h_n$.  Likewise, let $\X^1(h_n)$ and $\X(\D)$ denote the polyhedral complexes of $\D(h_n)$ and $\D^\zeta$.   Let $\tilde S_{h_n}^0$ (resp.~$\tilde S_{h_n}^1$ and $\tilde S_\D$) be the set of full-dimensional polyhedra of $\X^0(h_n)$ (resp.~$\X^1(h_n)$ and $\X(\D)$). Let $\tilde \g_{h_n}^0$ (resp.~$\tilde \g_{h_n}^1$ and $\tilde \g_\D$) be a piecewise-linear convex function that is linear precisely on the polyhedra of  $\X^0(h_n)$ (resp.~$\X^1(h_n)$) and $\X(\D)$). 

\smallskip

\noindent{\bf Step 3---Adding Players.} As in Step 3 of the proof of Theorem 6.1, we introduce a signaling game $\G^o$ that is independent of $\G$. The only difference between the two games is that the game there has $N$ senders, while this one has $|A^0|$ senders, one per $a_n^0 \in A^0$.    In the unique equilibrium outcome, all senders send the message $h$, and the unique belief at the information set of the receiver following a message of $l$ assigns equal weight to all senders. We use $S_{a_n^0}^o$ to denote the message space of sender $a_n^0$ and $S^o$ to denote the profile of pure strategies of the players in $\G^o$.

We now add the following set of outsiders as well.  For each $n$ and each information set $h_n \in H_n^-$, there are three outsiders, $\hat o_{h_n}$, $\tilde o_{h_n}^0$, and $\tilde o_{h_n}^1$; and  there are two other outsiders $\hat o_0$ and $\tilde o_0$.    The pure strategy sets of these outsiders are as follows. Outsider $\hat o_{h_n}$ chooses a vertex of $\D^\zeta(A_n(h_n))$ (resp.~$\D(A_n(h_n))$) if $h_n \in H_n^0$ (resp.~$H_n^1$).  Outsider $\hat o_0$ chooses a vertex of $\D^\zeta$.  The pure strategy set of outsider $\tilde o_{h_n}^i$, for $i = 0, 1$ (resp.~$\tilde o_0$)  is the set $\tilde S^i(h_n)$ (resp.~$\tilde S_{\D}$) of the full-dimensional polyhedra  of $\X^i(h_n)$ (resp.~$\X(\D))$. The space of profiles of these choices is  denoted $(\hat S^o, \tilde S^o)$.

\smallskip

\noindent{\bf Step 4---A Game Tree $\hat \G$.} We now describe a game tree $\hat \G$ involving players in $\N$ and the  outsiders introduced in Step 3.  First, all the senders $a_n^0$, and the other outsiders $\hat o$, $\tilde o$ move. Each player $n$ gets to observe the messages sent by the senders $a_n^0$. Then the extensive-form game  $\G$ is played by the players in $\N$. And finally the receiver of the signaling game $\G^o$ responds to the senders' messages. 

The game $\G$ is played with a slight modification. If all senders $a_n^0$ chose $h$, then each player $n$ first chooses which actions to take at information sets in $H_n^+$.  Specifically, he decides which of the following types of strategy to use: a choice $s_n^+: H_n^+ \to A_n$ such that $s_n^+(h_n) \in A_n(h_n)$ for all $h_n \in H_n^+$.  If he decides to use one of them, then those choices are automatically implemented and he has to choose at the other information sets.  If one of the senders $a_n^0$ sends $l$, then  $n$ has a strategy, to be described below, that is automatically implemented for him.  A pure strategy for $n$ is thus naturally decomposed into $(s_n^+, s_n^0, s_n^1)$.

Fix a pure strategy profile $(s^o, \hat s^o, \tilde s^o)$ for the outsiders and a profile $s$ for the original players.  If none of the senders sends the message $l$, then the original game is played by the insiders and the game ends (since the receiver does not have a choice after $h$).  If, say, a subset $\hat A_n^0$ of senders $a_n^0$ uses $l$,  then the following strategy is implemented for $n$ and without his opponents knowing about it: at all information sets in $H_n^+$, other than the information sets $h_n$ with $A_n(h_n) \cap \hat A_n^0 \neq \emptyset$,  the equilibrium mixture $b_n^+$ is implemented; at $h_n$ with an action $a_n^0$ in $\hat A_n^0$, with probability $\d_{a_n^0}$, $a_n^0$ is chosen, where $\d_{a_n^0}$ is the probability of $a_n^0$ under the choice $\hat s_{\so}$ of $\hat o_0$, and with the remaining probability,   the equilibrium mixture is implemented; at information sets $h_n$ in  $H_n^-$, player $\hat o_{h_n}$'s choice is implemented.  
\smallskip

\noindent{\bf Step 5---The Payoffs in $\hat \G(u)$.} The payoffs of the players are as follows. For each $h_n \in H_n^-$, pick a point $\tilde b_{h_n}(\tilde s_{h_n}^0)$ in the polyhedron $\tilde s_{h_n}^0$; similarly pick a point $\tilde \mu(h_n) \in \tilde s_{h_n}^1$ and $\tilde \d(\tilde s_0) \in \tilde s_0$.  $(\tilde b^-(\tilde s^o), \tilde \mu(\tilde s^o), \tilde \d(\tilde s^o))$ denotes the point corresponding to the profile $\tilde s^o$.  At nodes following choices of $h$ by all senders, the payoffs are as in $\G$ for $\N$.  When $a_n^0$ chooses $l$, then the payoffs for each information set $h_m^- \in H_m^-$ for each $m$ and nodes at $h_m^-$ that follow the choice of $a_n^0$, the payoff of action $a_m \in A_m(h_m)$ is augmented by $g_{a_m}(\tilde b^-(\tilde s^o), \tilde \mu(\tilde s^o),\tilde \d(\tilde s^o))$.  The payoffs of $\hat o_{h_n}$ coincide with the payoffs of $n$ when sender $a_n^0$ chooses $l$ and he gets to choose at $h_n$. 

The payoff to player $\hat o_0$ is as follows.   If he chooses a vertex $\hat s_0$ which corresponds to a mixture $\d$ in $\D^\z$,  for each $a_n^0 \in A^0$ and for each terminal node that follows just sender $a_n^0$ sending the message $l$, his payoffs are the payoffs of player $n$ but with a bonus $g_{\hat s^0}$ instead of $n$'s bonus. At all other nodes, his payoffs are zero.

It remains to describe the payoffs of the outsiders $\tilde o$. Outsider $\tilde o_0$ wants to mimic the choice of $\hat o_0$. The restriction of the  function $\tilde \g_\D$ to a full-dimensional polyhedron $\tilde s_{\tilde o_0}$ is affine.  It has a unique affine extension $\hat G_{\tilde s_{\tilde o_0}}$ to the whole of $\D^\zeta$. This function represents the payoff to the pure strategy $\tilde s_0$ that is the full-dimensional polyhedron. 


For each polyhedron $\tilde s_{h_n}$ of $\X^0(h_n)$, $\g_{h_n}^0$ defines a linear function $\hat G_{\tilde s_{h_n}}$ over the polyhedron which extends to one over $\D(A_n(h_n))$.  Now extend it to the convex hull of $\D(A_n(h_n))$ and $0$ by setting it to be zero on $0$. 
For each $n$, there is a map $\hat \psi_n: \S_n \times \S^o \times \hat \S^o \to \S_n$ that is just the projection if no sender $a_n^0$ plays $l$ with positive probability, and it is independent of $n$'s choice otherwise. If $\s_n = \hat\psi_n(s_n, s^o,\hat s^o)$,  then for each information set $h_n \in H_n^-$, we have a vector ${(\s_n(a_n))}_{a_n \in A_n(h_n)}$ with $\s_n(a_n) = \sum_{s_n \in S_n(a_n)} \s_{n,s_n}$ being the unconditional probability of playing $a_n$. The payoff against $(s, s^o, \hat s^o)$ for $\tilde s_{h_n}$ is $\hat G_{\tilde s_{h_n}}({(\s_n(a_n))}_{a_n \in A_n(h_n)})$ Thus, $\tilde o_{h_n}$ wants to mimic the implied unconditional probabilities of the actions at $h_n$.

Payoffs of $\tilde o_{h_n}^1$ are defined as mimicking the beliefs at information set $h_n$.   For each polyhedron $\tilde s_{h_n}$ of $\X^1(h_n)$, $\g_{h_n}^1$ defines a linear function $\tilde G_{\tilde s_{h_n}}^1$ over the polyhedron which extends to one over $\D(h_n)$.  Now extend it to the convex hull of $\D(h_n)$ and $0$ by setting it to be zero on $0$.  For each strategy profile $(s, s^0, \hat s^0, \tilde s^0)$ if $\s = \psi(s, s^0, \hat s^o)$, then for each node $x$ of $h_n$ there is a well-defined probability of reaching $x$, giving us a vector of these probabilities.  $\tilde o_{h_n}^1$'s payoff from $\tilde s_{h_n}^1$ is the value of $\tilde G_{\tilde s_{h_n}}^1$  at this vector of probabilities.  

\smallskip

\noindent{\bf Step 6---The Relationship between $\G$ and $\hat \G$.} $\hat \G$ weakly embeds $\G$ via the mappings $\hat \psi_n$.  Indeed, players can implement their strategies in $\G$ if the senders choose $h$.  A fixed strategy is implemented for $n$ if one of his sender counterparts plays $l$. The payoffs of the insiders are as in the original game $\G$ unless some sender sends $l$.  The equilibria of $\hat \G$ involve only the message $h$ from the senders.  Thus, $\hat \G$ embeds $\G$. By Axiom I, the  $\varphi$-solutions of $\G$ are the projections of the $\varphi$-solutions of $\hat \G$.

\smallskip

\noindent {\bf Step 7---Adding Dominated Strategies.} We now modify the game $\hat \G$ by adding a set of dominated strategies for the players in $\N$. For each $n$ and a choice $s_n^+$ at the first stage, if $n$ chooses an $s_n^+$ for which $s_n^+(h_n) \in A_n^0$ for some $h_n$, then we create a choice problem where we add as many dominated strategies as there are $a_n^0$'s in the choices made by $s_n^+$.  Specifically, for each $a_n^0 \in s_n^+(H_n^+)$, we create a dominated strategy $\hat a_n^0$, and  $n$ has to sequentially reject all these choices $\hat a_n^0$, one in each round, before executing $s_n^+$.  Choosing a dominated action $\hat a_n^0$ immediately triggers the message $l$  for the sender $a_n^0$ and the strategy that is implemented for $n$ is the same as well.  Only $n$'s payoffs in the continuation at $h_n(a_n^0)$ following choices in $A_n^+$  are reduced by $0 < \a \le \zeta$ to make it dominated in the continuation game, where $\a$ is chosen such that if the action $a_n^0$ is $\a$-optimal against an equilibrium $b \in B^*$, then $b^-$ is within $\eta$ of $\partial B^{*,-}$.  Call this modified game $\bar \G$. 

Let $\bar S_n$ be the set of pure strategies of $n$ in $\bar \G$. There is a mapping $\bar \psi_n: \bar \S_n \times \S^o \times \hat \S^o  \to  \S_n$ that coincides with $\hat \psi_n$ if $n$ does not choose a dominated strategy $\hat a_n^0$ and when he does, it maps the profile to the corresponding mixture. Likewise, for each sender $a_n^0$ in the signaling game, there is a mapping $\bar \psi_{a_n^0}: \bar \S_n \times \S_{a_n^0}^0 \to \S_{a_n^0}^o$ that  is just the projection if $n$ does not choose one of the dominated actions $\hat a_n^0$ and is just $l$ otherwise. As in Step 4 of the proof of Theorem 6.1,  it is easy to show that the equilibria of the game $\bar \G$ are the same as those of $\hat \G$. Hence, by Axiom D, the $\varphi$-solutions of the two games are the same. 

\smallskip

\noindent{\bf Step 8---Failure of Axiom B.} Since $\varphi$ satisfies Axioms D and I, there is a $\varphi$-solution $\bar \S_\varphi^*$ of the normal form $\bar G$ of $\bar \G(u)$ such that $\bar \psi(\bar \S_\varphi^*) = \S_\varphi^*$. To prove the result, we now show that $\bar \S_\varphi^*$ does not contain a mixed strategy profile that is equivalent to a quasi-perfect equilibrium of $\bar \G(u)$. Suppose  otherwise.  Then  the game $\bar \G(u)$ has a sequence of $\e$-quasi-perfect equilibria $\bar b^\e$ such that, for $\e > 0$, letting $\bar \s^\e$ be an equivalent mixed strategy profile, $\s^\e \equiv \bar \psi(\bar \s^\e)$, and letting $b^\e \in B$ be a profile in $\G$ equivalent to $\s^\e$, $b^\e$ converges to some $b^* = (b^{*,+}, b^{*,-}) \in B_\varphi^*$.  Let $\mu^\e$ be the beliefs induced by $b^\e$ at information sets in $H^-$ and let $\mu^*$ be its limit. Likewise, let $\d^\e = \bar\s_{\hat o_0}^\e$ and let $\d^*$ be its limit.

For each $\tilde o_{h_n}$, a pure strategy $\tilde s_{h_n}^0$ corresponds to a mixture $\tilde b(\tilde s_{h_n}^0)$.  Hence in the sequence of $\e$-quasi-perfect equilibria, his strategy $\bar \s_{\tilde o_{h_n}}^\e$ corresponds to a point $\tilde b^\e \in \D(h_n)$.  Similarly, we have a distribution $\tilde \mu^\e$ for player $\tilde o_{h_n}^1$ and $\tilde \d^\e$ for $\tilde o_0$, with limits $\tilde b^*, \tilde \mu^*, \tilde \d^*$. By the structure of the payoffs of these players, for each $h_n \in H_n^-$, $\tilde b^*(h_n)$ and $\tilde \mu^*(h_n)$ are within $\eta$ of $b^*(h_n)$ and $\mu^*(h_n)$ respectively, and similarly  $\tilde \d^*$  is within $\eta$ of $\d^*$. Obviously, $(b^{*,-}, \mu^*, \d^*) \in \bbBR_{g(\tilde b^{*,-}, \tilde \mu^*, \tilde \d^*)}$ by the structure of the payoffs in $\bar \G$.

We claim now that in $\bar b^\e$, no player $n$ chooses an $s_n^+$ with an $a_n^0$ in its range, either initially or in any of the successive rounds.  Obviously under $\bar b^\e$, no player chooses an $s_n^+$ with an $a_n^0$ in its range, as it would induce a nonequilibrium outcome. Suppose now that conditional on choosing such an $s_n^+$, he  reconfirms it.  Then there is a pair $(s_n^0, s_n^1)$ such that $(s_n^+, s_n^0, s_n^1)$ is $\a$-optimal against $b^*$.   This implies that $b^{*,-}$ is in the $\eta$-neighborhood of $\partial B^{*,-}$.  As $\tilde b(h_n)$ is within $\eta$ of $b(h_n)$ and $g$ is zero in the $2\eta$-neighborhood of $\partial \bbB^{*}$, the bonus vector $g$ is then zero.  For each $\e$, let $\hat \d^\e \in \D$ be such that  $\hat \d_{a_n^0} \propto b_{a_n^0}^\e$ and let $\hat \d$ be its limit.  Now $(b^{*,-},\hat \d)$ is a $\zeta$-best reply to itself in $\G$, which is impossible. Hence, each $n$ avoids reconfirming $s_n^+$ with $a_n^0$ in its range.

Under $\bar b^\e$, let $\bar b_{o_{a_n^0}}^\e$ be the probability that sender $a_n^0$ sends the message $l$; this probability is the sum of two probabilities: (1) the probability that sender $a_n^0$ voluntarily sends the message $l$; (2) the probability that $l$ is triggered by $n$'s use of a dominated choice $\hat a_n^0$. From the receiver's perspective, quasi-perfection requires that $\bar b_{o_{a_n^0}}^\e = \bar b_{o_{a_m^0}}^\e$  for all $a_n^0, a_m^0$. Therefore, $\lim_{\e \to 0} \frac{b_{a_n^0}^\e}{b_{a_m^0}^\e} = \frac{\d_{a_n^0}^*}{\d_{a_m^0}^*}$ for each $a_n^0, a_m^0 \in A^0$. Hence,  $(b^{*,-}, \mu^*, \d^*)$ belongs to $\bbB^*$. 

We now have that $(\tilde b^{*,-}, \tilde \mu^*, \tilde \d^*)$ is within $\eta$  of $(b^{*,-}, \mu^*, \d^*)$, which belongs to $\bbB_\varphi^*$,  and $(b^{*,-}, \mu^*, \d^*) \in \bbBR_{g(\tilde b^{*,-}, \tilde \mu^*, \tilde \d^*)}$, which is impossible. Hence $\bar \G$ does not have a quasi-perfect equilibrium projecting to $B_\varphi^*$, and  $\varphi$ fails Axiom B. $\Box$

\end{document}